\documentclass[draft,12pt,reqno,a4paper]{amsart}

\usepackage[margin=3cm,footskip=1cm]{geometry}

\usepackage{amssymb} 
\usepackage{amsmath} 
\usepackage{amsfonts}
\usepackage{amsthm} 
\usepackage{mathtools}
\usepackage{color}

\usepackage{enumerate} 
\usepackage[latin1]{inputenc}
\usepackage[shortlabels]{enumitem}
\usepackage{color}
\usepackage{graphicx} 
\usepackage{verbatim}

\newcommand{\supp}{\operatorname{supp}}

\DeclareMathOperator\ran{ran}

\newcommand{\spann}{\mathrm{span}}

\newcommand{\ad}{{\operatorname{ad}}}
\newcommand{\N}{{\mathbb{N}}}

\newcommand{\R}{{\mathbb{R}}}

\newcommand{\C}{{\mathbb{C}}}
\renewcommand{\S}{{\mathbb{S}}}

\renewcommand\i{\mathrm{i}}

\renewcommand{\c}{{\mathrm c}}

\newcommand{\e}{{\mathrm e}}
\newcommand{\ess}{{\mathrm {ess}}}

\renewcommand{\d}{{\mathrm d}}
\newcommand{\dol}{{\mathrm {dol}}}

\newcommand{\diag}{{\mathrm {diag}}}

\newcommand{\pupo}{{\mathrm {pp}}}

\renewcommand{\Re}{\operatorname{Re}}
\renewcommand{\Im}{\operatorname{Im}}

\DeclarePairedDelimiter\inp\langle\rangle
\newcommand\paro[2][]{#1  ( #2#1 )}

\newcommand\parb[2][]{#1 \big ( #2#1\big )}

\newcommand\parbb[2][]{#1 \Big ( #2#1\Big )}
 
 \newcommand{\pp}{{\mathrm {pp}}}
  
\newcommand{\mand}{\text{ \,and\, }}

\newcommand{\cas}{{\textrm {the Cauchy-Schwarz inequality }}}
\newcommand{\caS}{{\textrm {the Cauchy-Schwarz inequality}}}

\newcommand{\1}{\hspace{ 1cm}}

\DeclarePairedDelimiter\ket{\lvert}{\rangle}
\DeclarePairedDelimiter\bra{\langle}{\rvert}

\DeclareMathOperator*{\slim}{s-lim}
\DeclareMathOperator*{\wlim}{w-lim}

\DeclareMathOperator*{\wvHlim}{{w-\mathcal H}-lim}

\DeclareMathOperator*{\wLlim}{{w}-{\mathit {L^2(C_a)}}-{lim}}

\DeclareMathOperator*{\swslim}{s- w^\star-lim}

\DeclarePairedDelimiter\abs\lvert\rvert
\DeclarePairedDelimiter\norm\lVert\rVert
\DeclarePairedDelimiter\set{\{}{\}}
\DeclarePairedDelimiter\comm{[}{]}

\newcommand\Step[1]{ 
  \par\bigskip
  \noindent
  \textbf{#1}.\enspace
}

\newcommand{\brR}{{\breve R}}
\newcommand{\brh}{{\breve h}}
\newcommand{\brH}{{\breve H}}

\newcommand{\brI}{{\breve I}}
\newcommand{\brs}{{\breve s}}

\newcommand{\brr}{{\breve r}}
\newcommand{\brq}{{\breve q}}
\newcommand{\brJ}{{\breve J}}

\newcommand{\bD}{{\mathbf D}}

\newcommand{\bX}{{\mathbf X}}

\newcommand{\vA}{{\mathcal A}}
\newcommand{\vB}{{\mathcal B}}
\newcommand{\vC}{{\mathcal C}}
\newcommand{\vD}{{\mathcal D}}
\newcommand{\vE}{{\mathcal E}}

\newcommand{\vG}{{\mathcal G}}

\newcommand{\vH}{{\mathcal H}}

\newcommand{\vL}{{\mathcal L}}

\newcommand{\vN}{{\mathcal N}}
\newcommand{\vO}{{\mathcal O}}

\newcommand{\vQ}{{\mathcal Q}}
\newcommand{\vR}{{\mathcal R}}
\newcommand{\vT}{{\mathcal T}}
\newcommand{\vU}{{\mathcal U}}
\newcommand{\vV}{{\mathcal V}}

\theoremstyle{plain}%default
\newtheorem{thm}{Theorem}[section]

\newtheorem{lemma}[thm]{Lemma} \newtheorem{corollary}[thm]{Corollary}
\theoremstyle{definition} 
\newtheorem{defn}[thm]{Definition} 

 \newtheorem{cond}[thm]{Condition}
 \newtheorem{remark}[thm]{Remark}
\newtheorem{remarks}[thm]{Remarks}
 \newtheorem*{remarks*}{Remarks}
\newtheorem*{remark*}{Remark}

\numberwithin{equation}{section}

\title {stationary scattering theory, the $N$-body long-range case}

\thanks{
Supported by DFF grant nr.\ 8021-00084B}

\author{E. Skibsted} \address[E. Skibsted]{Institut for Matematik\\
Aarhus Universitet\\ Ny Munkegade 8000 Aarhus C, Denmark}
\email{skibsted@math.au.dk}

\begin{document}

\begin{abstract} Within the class of Derezi{\'n}ski-Enss  pair-potentials
   which includes  Coulomb potentials and for which 
    asymptotic completeness is known \cite{De}, we show that all entries
  of the $N$-body quantum scattering matrix have a  well-defined meaning at any given
  non-threshold energy. As a function of  the energy parameter the
  scattering matrix is weakly continuous. This  result generalizes a  
  similar one  obtained previously by Yafaev for systems of particles interacting
  by short-range potentials  \cite{Ya1}. As for Yafaev's  paper  we do not make any
  assumption on the decay of channel eigenstates. The main part of the
  proof consists
  in establishing a number of 
   Kato-smoothness bounds needed for justifying a new formula for
    the
  scattering matrix. Similarly we construct and show strong continuity
  of channel wave matrices for
  all non-threshold energies. Away from a set of measure zero we show that the scattering
  and channel wave matrices constitute a well-defined `scattering
  theory', in particular at such energies the scattering matrix is
  unitary, strongly continuous and characterized by asymptotics of minimum
  generalized eigenfunctions.
\end{abstract}

\allowdisplaybreaks

\maketitle


\tableofcontents

\section{Introduction}\label{sec:Introduction}
Asymptotic completeness of systems of quantum particles interacting
  by long-range potentials (more precisely  by pair-potentials
 decaying like $\vO(\abs{x^a}^{-\mu})$ with  $\mu>\sqrt{3}-1$) was proven by Derezi{\'n}ski \cite{De} by 
   entirely  time-dependent methods. While asymptotic completeness of systems of quantum particles interacting
  by short-range potentials can be proven by a  time-independent
  method, see \cite{Ya3}, there does seem to appear any  time-independent method 
  for the long-range case  in the literature. A notable virtue of
  the mentioned  papers is their generality, thus the completeness results hold
  without implicit assumptions.

One deficiency of this situation is  that  we lack 
  understanding of the
  scattering matrix in a general long-range  setup. This is in contrast to Yafaev's case where his method
  reveals some properties of  the (short-range) 
  scattering matrix. In particular Yafaev showed that this quantity is
  weakly continuous in the spectral parameter away from the threshold
  set, as deduced from a formula,  \cite{Ya1}. 

We show in this paper similar results for
  the long-range case. The continuity  will be derived from a new
  stationary formula, and similarly to   Yafaev's case our  formula is justified
  in terms of various Kato smoothness bounds. Only a few of them appears in
  Yafaev's papers. In particular our paper is strongly dependent of
  certain time-independent bounds/observables which may be understood as
  modifications of Derezi{\'n}ski's key ingredients. In addition to
  \cite{De} and \cite{Ya3} our
  paper relies on  some resolvent bounds first   proven in \cite{GIS} and
  then recently extended to a Besov space setting  along with a (sharp) Sommerfeld
  uniqueness result in \cite{AIIS}.

 For any channel
  $\alpha=(a,\lambda^\alpha, u^\alpha)$ we consider the
 \emph{channel wave operators} 
\begin{align*}
  W_\alpha^{\pm}=\slim_{t\to \pm\infty}\e^{\i tH}J_\alpha\e^{-\i
  (S_a^\pm (p_a,t)+\lambda^\alpha t)},\quad J_\alpha
f_a=u^\alpha\otimes f_a,
\end{align*} where $H$ denoted the full $N$-body Schr{\"o}dinger
operator,  $S_a^\pm (\xi_a,t)$ solves  a certain  Hamilton-Jacobi
equation (it is a `distortion' of the function $t\xi_a^2$). For two
channels, say incoming $\alpha=(a,\lambda^\alpha, u^\alpha)$ and outgoing $\beta=(b,\lambda^\beta, u^\beta)$, the
corresponding entry $S_{\beta\alpha}=(W_{\beta}^+)^*W_{\alpha}^-$  of the scattering operator takes
the form in a diagonalizing (momentum)  space,
\begin{align*}
\hat
  S_{\beta\alpha}=\int^\oplus_{
  (\max\set{\lambda^\beta,\lambda^\alpha}, \infty)}
  S_{\beta\alpha}(\lambda)\,\d \lambda.
\end{align*} One can abstractly construct the fiber operators
$S_{\beta\alpha}(\cdot)$ as  almost everywhere defined quantities,
however  such a construction does not seem to  
 provide  properties 
 of scattering. On the contrary we prove (independently of any abstract construction)  that there exist a (unique)
weakly continuous candidate away from the threshold
  set. Another main result of the paper is that for fixed energy also  the restriction
  of the wave operators and their adjoints (the latter are called wave
  matrices) exist. Given these results it makes sense to ask if there is a
   `scattering theory' at a fixed energy? While this is not
  settled in this paper we do  show that almost all non-threshold
  energies are `stationary scattering energies'. In particular at  any such energy 
  $ S_{\beta\alpha}(\cdot)$ is strongly continuous.

  The weak continuity property of  $S_{\beta\alpha}(\cdot)$ all non-threshold
  energies will be read off from 
an explicit formula of the following form. Up to some trivial  multiplicative
factors in momentum space and localization away from `collision planes'
\begin{align*}
  S_{\beta\alpha}(\lambda)\approx\parbb{(H-\lambda)\Phi_\beta^+
    J_\beta\gamma_b^+(\lambda-\lambda^\beta)^* }^*\delta(H-\lambda) \parbb{(H-\lambda)\Phi_\alpha^-
    J_\alpha\gamma_a^-(\lambda-\lambda^\alpha)^*}.  
\end{align*}  Here $\delta(H-\lambda) =\pi^{-1}\Im\parb{(H-(\lambda+\i
0)^{-1}}$ is the delta-function of $H$ at $\lambda$, $\gamma_b^+(\lambda-\lambda^\beta)^* $ and
$\gamma_a^-(\lambda-\lambda^\alpha)^*$ are certain auxiliary 'one-body wave
matrices' (constructed from inter-cluster potentials), while  the
\emph{ stationary channel modifiers} $\Phi_\alpha^-$ and  $\Phi_\beta^+$ should be
considered as some appropriate, incoming and outgoing phase-space  `distortions of the
identity'. The latter operators  are introduced  so that (looking only at the
right-hand factor)
\begin{align*}
   (H-\lambda)\Phi_\alpha^-
    J_\alpha\gamma_a^-(\lambda-\lambda^\alpha)^*\approx\sum_k Q^-(a,k)^*\parbb{B_-(a,k)Q^-(a,k)J_\alpha{\gamma_a^-(\lambda-\lambda^\alpha)}^*},
\end{align*} Here the sum is finite, the operators $B_-(a,k)$ are bounded
and the only $\lambda$-dependence is sitting in the  argument of the
incoming  one-body wave
matrix. With a similar structure for  the
left-hand factor the weak continuity follows from three properties to be
demonstrated in the paper:
\begin{enumerate}[i)]
\item \label{item:1oo} $ Q^+(b,l)\delta(H-\lambda) Q^-(a,k)^* $ is weakly continuous in $\lambda$.
\item\label{item:2oo} $Q^-(a,k)J_\alpha{\gamma_a^-(\lambda-\lambda^\alpha)}^*$  is
strongly continuous in $\lambda$.
\item\label{item:3oo} $Q^+(b,l)J_\beta{\gamma_b^+(\lambda-\lambda^\beta)}^*$  is strongly continuous in $\lambda$.
\end{enumerate} Note that  uniform boundedness of the operator in
\ref{item:1oo} is strongly related to Kato boundedness of the
$Q$-operators \cite{Ka}. Apart from the limiting absorption principle bound the
corresponding boundedness assertions for \ref{item:1oo}--\ref{item:3oo}
all arise from  a well-known
positive commutator technique,  for example used extensively  in the cited papers
\cite{De, Ya3,Ya1}. 

Although  we  take the existence of the wave operators
$W_\alpha^{\pm}$ for
granted in this paper, which is legitimate due to \cite{De}, we remark that it may
be derived independently from the stationary  setup
of the present  paper. The proof of the existence and  weak continuity of the
scattering matrix at all non-threshold energies may then also be considered as being independent of
\cite{De}. However for deducing unitarity and the strong  continuity
for almost all energies of the scattering matrix, as obtained in the paper,   we do indeed
use the completeness assertion of \cite{De}. We leave it as a
conjecture that those stronger properties are fulfilled  for \emph{all}
 non-threshold energies. 

We  remark that our
method can be viewed as an extension  of the one of 
\cite{Ya1} for the short-range case (although involving  several new
ingredients and a different representation of  the scattering
matrix). In particular Yafaev's
result  on the weak continuity of the  short-range 
  scattering matrix may be viewed as a consequence of our theory, see Remark \ref{remark:Dollard-time-depend-scatt}
  \ref{item:Dol4}.

In Subsection \ref{subsec:A principle example, atomic $N$-body
  Hamiltonians} we  elaborate on  the consequences of our 
long-range theory (for the general class of  Derezi{\'n}ski-Enss  pair-potentials)
applied to  a 
well-known atomic physics model.
 The bulk of the paper is organized as follows. With a decay condition on the
channel bounds states $u^\alpha$ and $u^\beta$ one can obtain the weak
continuity of $ S_{\beta\alpha}(\cdot)$ by a different method with pair-potentials
 decaying only like  $\vO(\abs{x^a}^{-\mu})$ for   $\mu>0$. We discuss
 this in Section \ref{sec:preliminaries}. The reader may skip most
 of this material, however the precise description of the $N$-body
 model is given there and certain  ingredients will be used in
 later sections. Most notably this is the case for the theory of
 Section \ref{sec:preliminaries} specialized to the one-body case. In
 Section \ref{sec:one-body  matrices} we collect most of what is
 needed from the one-body scattering theory. A main result of the
 paper (this is the weak continuity discussed above) is stated in
 Section \ref{sec:Stationary modifier} as Theorem
 \ref{thm:mainstat-modif-n} which  holds under Condition
 \ref{cond:smooth2wea3n12} with  $\mu>\sqrt{3}-1$. In the same section
 we outline the scheme for proving the theorem, in particular we
 introduce  stationary modifiers. The ingredients from \cite{Ya3}  and
 \cite{De} that we need (essentially vector field constructions) are given in Sections \ref{sec:Yafaev's
   construction} and \ref{sec:Derezinski's
  construction}, respectively. In addition the latter section
contains a Subsection \ref{subsec:Limiting absorption principles}
with  two Mourre type
estimates.  We need to do commutation with the
stationary modifiers which are given by functional calculus. We devote
Section \ref{sec:Calculus considerations} to studying a calculus
  facilitating this task. The commutation is then implemented in Sections \ref{sec:Computation of T} and 
\ref{sec:Formula for widetilde} in which the above assertions
\ref{item:1oo} and \ref{item:2oo}-\ref{item:3oo} are proven,
respectively. We finalize the proof of Theorem
 \ref{thm:mainstat-modif-n} in Subsection \ref{subsec:Conclusion of argument, the
              weak continuity}. The exact expression for the scattering matrix, of the
type 
indicated 
above,  is derived in Appendix \ref{sec:AppendixB} with inputs from
\cite{GIS, AIIS} and by using results from Sections \ref{sec:Computation of T} and 
\ref{sec:Formula for widetilde}.  We devote Appendix 
\ref{sec:AppendixO} and Appendix  \ref{sec:AppendixA} to a few  missing
technical details on assertions used previously.

Our second  main result,  which concerns representation
 of channel  states in terms of certain  generalized eigenfunctions,   is given
 in
 Section \ref{sec:Exact channel wave-matrices}  as Theorem
 \ref{thm:chann-wave-matr}. The involved
 integrals are given in terms of \emph{channel wave
  matrices} that we construct for  all  non-threshold energies. Another issue of
  Section \ref{sec:Exact channel wave-matrices} is the 
  construction of the \emph{scattering matrix}  for
  all non-threshold energies (done concretely), as well as  its
  unitarity and strong continuity  away from a null set of
 energies,  see Definition  \ref{defn:scattering matrix}  and
  Corollary \ref{cor:direct}. Such generic energies, introduced more
  precisely in  Definition \ref{defn:scatEnergy}, we call `stationary scattering
  energies'. We characterize for any stationary scattering
  energy  and for any  incoming 
  $\alpha$-channel the associated \emph{minimum
  generalized eigenfunctions} in terms of their explicit  asymptotic
properties, see Theorems 
 \ref{prop:besov-space-setting} and \ref{Cor:besov-space-setting}. The
 latter result
  may also be vieved as a  characterization of the    incoming 
  $\alpha$-channel  part of the
  scattering matrix at any such energy.
 Wave matrices  have been studied in
 many  contexts in the literature, for example abstractly in
 \cite{KK} and concretely for  $N$-body Schr\"odinger operators
 with short-range interactions in \cite{Ya4} (our exposition is
 rather different from Yafaev's). The developed theory for stationary scattering
 energies is similar to the standard
 one  for the one-body problem at any positive energy (in that case all positive energies are stationary scattering energies), see for example
 \cite{DS, GY}. 

As for the above discussion we ignore in   the
 bulk of the paper  the possible existence of non-threshold
 embedded eigenvalues. It is a minor technical exercise to include non-threshold
   embedded eigenvalues for all  results 
  of the paper, which however will not  be discussed further.     
\subsection{A principle example, atomic and molecular $N$-body
  Hamiltonians}\label{subsec:A principle example, atomic $N$-body
  Hamiltonians}

Consider a system of  $N$ charged particles of dimension $n$  interacting by Coulomb forces. The
Hamiltonian  reads 
\begin{equation}
\label{0.1Cou}
H=-\sum_{j = 1}^N \frac{1}{2m_j}\Delta_{x_j} + \sum_{1 \le i<j \le N} q_iq_j|x_i
-x_j|^{-1}, \quad x_j\in\R^n,\,n\geq 3,
\end{equation}
where $x_j$,  $m_j$ and $q_j$ denote the position, mass and charge of
the $j$'th particle, respectively.

 The Hamiltonian $H$ is regarded as a self-adjoint operator
on $L^2(\bX)$, where $\bX$ is the $n(N-1)$ dimensional real vector
space  $ \bX:= \{ \sum_{j=1}^{N} m_j x_j = 0\}$.
Let $\vA$  denote the set of all cluster
decompositions of the $N$-particle system. The notation $a_{\max}$ and
$a_{\min}$ refers to the $1$-cluster and $N$-cluster decompositions,
respectively.
  Let for $a\in\vA$ the notation  $\# a$ denote the number of
clusters in $a$.
For $i,j \in\{1, \dots, N\}$, $i< j$, we denote by $(ij) $ the
$(N-1)$-cluster decomposition given by letting $C=\{i,j\}$ form a
cluster and all other particles $l\notin C$ form singletons. We write $(ij) \leq a$ if $i$ and $j$ belong to the same cluster
in $a$.   More general, we write $b\leq  a$ if each cluster of $b$
is a subset of a cluster of $a$. If $a$ is a $k$-cluster decomposition, $a= (C_1, \dots, C_k)$,
we let 
\[
\bX^a = \set[\big]{ x\in\bX|\,\, \sum_{l\in C_j } m_l x_l = 0,  j = 1, \dots,
k}=\bX^{C_1}\oplus\cdots \oplus\bX^{C_k},
\]
and
\[
\bX_a  =  \set[\big]{ x\in\bX|  x_i = x_j \mbox{ if } i,j \in C_m  \mbox{ for some }
m \in \{ 1, \dots, k\}  }.
\]
 Note that $b\leq a\Leftrightarrow \bX^b\subseteq\bX^a$, and that the
 subspaces  $\bX^a$ and $\bX_a$  define  an orthogonal decomposition
 of  $\bX$
equipped  with
the quadratic form
\[
 q(x)=\sum_j 2m_j|x_j|^2,  \1 x\in {\bX}.
\]
 Consequently any  $x\in \bX$ decomposes orthogonally as 
 $x =x^{a} + x_{a}$ with $x^a =\pi^a x\in\bX^a$ and $x_a =\pi_a x\in \bX_a$.

With these notations, the many-body Schr\"odinger operator
\eqref{0.1Cou}  takes the form $ H = H_0 + V$, 
where  $H_0=p^2$ is (minus)  the Laplace-Beltrami operator on   the
Euclidean space  $(\bX, q)$ and
$V=V(x) =  \sum_{b=(ij)\in\vA} V_{b}(x^{b}) $ with $ V_b (x^b) =
V_{ij} (x_i - x_j)$ for the  $(N-1)$-cluster decomposition
$b=(ij)$. Note for example  that 
\begin{align*}
  x^{(12)}=\parb{\tfrac{m_2}{m_1+m_2}(x_1-x_2),-\tfrac{m_1}{m_1+m_2}(x_1-x_2),0,\dots,0}.
\end{align*}

For any cluster 
decomposition $a\in\vA$ we introduce the  sub-Hamiltonian $H^a$ as follows. 
For $a=a_{\min}$  we define
$H^{a_{\min}}=0$ on $\mathcal H^{a_{\min}}:=\mathbb C. $
For $a\neq a_{\min}$ we introduce the potential 
\begin{equation*}
V^a(x^a)=\sum_{b=(ij)\leq a} V_{b}(x^b);\quad
x^a\in \bX^a.
\end{equation*} 
Then 
\begin{align*}
 H^a:=-
\Delta_{x^a} +V^a(x^a)=
(p^a)^2 +V^a\ \ 
\text{on }\mathcal H^a=L^2(\bX^a).
\end{align*}

A channel $\alpha$ is by
definition given as $\alpha =(a,\lambda^\alpha, u^\alpha)$, where
$a\in\vA'=\vA\setminus \{a_{\max}\}$ and  $u^\alpha\in \mathcal H^a$ obeys
$\norm{u^\alpha}=1$ and 
$(H^a-\lambda^\alpha)u^\alpha=0$ for a real number
$\lambda^\alpha$. For any $a\in \vA'$ the intercluster potential is by definition
\begin{align*}
  I_a(x)=\sum_{b=(ij)\not\leq a}V_b(x^b).
\end{align*} Next we introduce  atomic Dollard type channel wave operators. They  read
\begin{align}\label{eq:Atomwave_op}
  W_{\alpha,{\rm atom}}^{\pm}=\slim_{t\to \pm\infty}\e^{\i
  tH}{u_\alpha}\otimes \parb{\e^{-\i
  (S_{a,{\rm atom}}^\pm (p_a,t)+\lambda^\alpha t)} (\cdot)},
\end{align} where
\begin{align*}
   S_{a,{\rm atom}}^\pm (\xi_a,\pm|t|)=\pm
  S_{a,{\rm atom}}(\pm\xi_a,|t|)\mand S_{a,{\rm atom}}(\xi_a,t)=t\xi_a^2+\int_1^t\,I_a(2s\xi_a)\,\d s.
  \end{align*} The well-definedness of $W_{\alpha,{\rm atom}}^{\pm}$
  follows from the existence of  $W_{\alpha}^{\pm}$, see  Remarks
  ~\ref{remark:Dollard-time-depend-scatt} \ref{item:Dol1} and \ref{item:Dol2}.
 Let $I^\alpha=(\lambda^\alpha,\infty)$ and
$k_\alpha=p_a^2+\lambda^\alpha$. By the intertwining property $H W_{\alpha,{\rm atom}}^{\pm}\supseteq
W_{\alpha,{\rm atom}}^{\pm}k_\alpha $  and the fact that $k_\alpha$ is diagonalized by
the unitary map  $F_\alpha:L^2(\mathbf X_a)\to L^2(I^\alpha
;L^2(C_a))$, $C_a=\mathbf X_a\cap\S^{n_a-1}$ with  $ n_a=\dim \mathbf X_a$,  given by
\begin{align*}
  (F_\alpha u)(\lambda,\omega)=(2\pi)^{-n_a/2}2^{-1/2}
  \lambda_\alpha^{(n_a-2)/4}\int \e^{-\i  \lambda^{1/2}_\alpha \omega\cdot
  x_a}u(x_a)\,\d x_a,\quad \lambda_\alpha=\lambda-\lambda^\alpha,
\end{align*} we can for \emph{any} two given  channels $\alpha$ and
$\beta$  write 
\begin{align*}
  \hat
  S_{\beta\alpha,{\rm atom}}:=F_\beta(W_{\beta,{\rm atom}}^+)^*W_{\alpha,{\rm atom}}^-F_\alpha^{-1}=\int^\oplus_{
  I_{\beta\alpha} }
  S_{\beta\alpha,{\rm atom}}(\lambda)\,\d \lambda,\quad
  I_{\beta\alpha}=I^\beta\cap I^\alpha.
\end{align*} The fiber operator $ S_{\beta\alpha,{\rm atom}}(\lambda)\in
\vL(L^2(C_a),L^2(C_b))$ is from an abstract point of view (our point
of view is different as discussed previously) 
a priori defined only for
a.e. $\lambda \in I_{\beta\alpha}$. It is the
{$\beta\alpha$-entry of the  atomic  Dollard type scattering matrix} $S_{{\rm atom}}(\lambda)=\parb{S_{\beta\alpha,{\rm atom}}(\lambda)}_{\beta\alpha}$.

Now a  main result of this paper,  specialized to the above atomic  
model, reads as follows (see Remark \ref{remark:Dollard-time-depend-scatt} \ref{item:Dol3}).
\begin{thm}\label{thm:princ-example-atom} The operator-valued function
  $S_{\beta\alpha,{\rm atom}}(\cdot)$ is weakly continuous away
  from the set of thresholds (i.e. eigenvalues of the proper sub-Hamiltonians) and
  eigenvalues of $H$. Away  from a null set $S_{\beta\alpha,{\rm
      atom}}(\cdot)$ is strongly continuous. 

More generally $S_{{\rm atom}}(\cdot)$  is a weakly continuous
contraction away
  from the set of thresholds and eigenvalues of $H$, and away  from a null set $S_{{\rm
      atom}}(\cdot)$ is a strongly continuous unitary operator. 
  \end{thm} A similar result is valid for Coulomb systems of charged
  particles with static (i.e. infinite mass) nuclei. It is a minor
  doable 
  issue   to include non-threshold
   embedded eigenvalues in the above result (as mentioned above). The
  behaviour at thresholds is much more intriguing and completely outside
   the scope of  this paper. Strong continuity for all non-threshold
   energies (as conjectured before) is another difficult problem.

The second  main result, again for convenience specialized to the above 
model, reads as follows. Recall the standard notation for weighted spaces 
\begin{align*}
  L_s^2(\mathbf X)=\inp{x}^{-s}L^2(\mathbf X);
\quad s\in\R, \,\,\inp{x}=\parb{1+\abs{x}^2}^{1/2}.
\end{align*}

\begin{thm}\label{thm:chann-wave-matrD}  Let $\alpha$ be a  given channel  $\alpha
    =(a,\lambda^\alpha, u^\alpha)$,     $f:I^\alpha=(\lambda^\alpha,\infty)\to \C$ be
  continuous and compactly supported away from the set of 
thresholds and
  eigenvalues, and let $s>1/2$. For any $\varphi\in
  L^2(\mathbf  X_a)$   
  \begin{align}\label{eq:wavD1}
  W^\pm_{\alpha,{\rm atom}}
  f\paro{ k_\alpha}\varphi=\int_{I^\alpha } \,f(\lambda)
    W^\pm_{\alpha,{\rm atom}}(\lambda) \paro{F_\alpha \varphi)(\lambda,
  \cdot} 
    \,\d \lambda\in  L_{-s}^2(\bX), 
              \end{align} where the wave matrices
              $W^\pm_{\alpha,{\rm
                  atom}}(\lambda)\in\vL\parb{L^2(C_a),L_{-s}^2(\bX)}$
              with a strongly continuous dependence of $\lambda$.
 In particular for  $\varphi\in
  L_s^2(\mathbf  X_b)$   
the integrand is a continuous compactly supported 
$L_{-s}^2(\bX)$-valued function. In general the integral
has  the weak interpretation of a measurable 
$L_{-s}^2(\bX)$-valued function.
\end{thm}

Of particular interest is the case $f(\lambda)=f_t(\lambda)= \e^{-\i t\lambda}f_0(\lambda)$,
$t\in\R$,  in
which case the formulas \eqref{eq:wavD1}  represent exact 
Schr\"odinger wave packets of energy-localized states outgoing to  or incoming  from
the channel $\alpha$.  For a   state  $\varphi^-_\alpha=W^-_{\alpha,{\rm atom}}
  f_{0}\paro{ k_\alpha}\varphi$ the asymptotics of the 
  wave packet is `controlled' as $t\to -\infty$. To
  `understand' the behaviour  as $t\to +\infty$ (and possibly  compare
  with practical physics experiments) the right-hand side  of
  \eqref{eq:wavD1} in principle provides  a tool. The main contribution  comes from the
  asymptotics of   functions in the range of the outgoing resolvent
  $(H-(\lambda+\i 0))^{-1}$ (as confirmed by one of our  formulas for the
  generalized eigenfunctions in the range of $W^-_{\alpha,{\rm
      atom}}(\lambda)$). On the other hand this quantity is in general  poorly
  understood for the $N$-body case. In
  the paper we 
 derive several non-trivial estimates (more
  precisely 
   mostly 
for the delta-function of $H$ only), however they are all  of
  `weak type'.  Strong resolvent bounds is yet another  difficult 
  problem with  room for desirable improvements. So maybe the main virtue
  of Theorem \ref{thm:chann-wave-matrD} and its generalizations at
  this point is   conceptual: The  wave matrices
              $W^\pm_{\alpha}(\lambda)$ uniquely exist
              admitting representations like \eqref{eq:wavD1}.

From a practical point of view it is more convenient to substitute  the
factor $W^-_{\alpha,{\rm
      atom}}(\lambda)$ in  \eqref{eq:wavD1} by
  \begin{align}\label{eq:linkForm}
     W^-_{\alpha,{\rm
      atom}}(\lambda)=\sum_{\lambda^\beta<\lambda} W^+_{\beta,{\rm
      atom}}(\lambda)S_{\beta\alpha,{\rm
      atom}}(\lambda).
  \end{align} Now each term contributes by a term which is  `controlled'
  as  $t\to +\infty$, meaning that all relevant large time information
  of our ` incoming  $\alpha$-channel experiment' is encoded in the components
   of the scattering matrix. In turn,
  although 
  the resolvent is complicated,  it  appears  in  formulas  only as
  the  delta-function of $H$ sandwiched by some Kato smooth
  operators (as
   discussed previously). We do not in this paper offer  `finer analysis' of the
  derived formulas on interesting (but difficult) issues like regularity or asymptotics of associated
  distributional kernels. 

Finally let us remark that we indeed  show the
  validity of 
  \eqref{eq:linkForm}  as an  identity in
  $\vL\parb{L^2(C_a),L_{-s}^2(\bX)}$  (summing in the weak sense) for
  all `{stationary scattering energies}', see
  Definition~\ref{defn:scatEnergy}. This is for almost all energies,
  and it is also an immediate consequence of our theory   that  at
    any stationary scattering energy  the scattering matrix $S_{\rm
      atom}(\cdot)$ is a  strongly continuous unitary operator
    determined uniquely  by asymptotics of minimum
  generalized eigenfunctions.

\section{Preliminaries}\label{sec:preliminaries}
We introduce the standard abstract $N$-body setup and introduce a
radial limit for the one-body case (more precisely for a specific
one-body potential). Then we develop a stationary
scattering theory for the $N$-body case with  a decay assumption on
the channel bound states, Condition \ref{cond:decayBNstates}, and
prove under this condition the existence and weak continuity of
entries of the
scattering matrix, see Theorem \ref{thm:scat_matrices}. The
later sections of the paper do not involve Condition
\ref{cond:decayBNstates}. The reader may skip all results of the
section based on Condition
\ref{cond:decayBNstates} and go directly to Section \ref{sec:one-body
  matrices} where we have collected most of what is needed for the later
sections (except primarily for Subsection \ref{subsec:body Hamiltonians, limiting absorption
  principle  and notation}, some notation  and the discussion of the eikonal and
Hamilton-Jacobi equations). Section \ref{sec:one-body
  matrices}  concerns a specific  one-body problem  for which,
essentially  being viewed as a
special case, the results of  the present section apply.
\subsection{$N$-body Hamiltonians, limiting absorption principle  and
  notation}\label{subsec:body Hamiltonians, limiting absorption
  principle  and notation}

Let $\bX$ be                    
a finite dimensional real inner product space,
equipped with a
finite family $\{\bX_a\}_{a\in \vA}$ of subspaces closed under intersection:
For any $a,b\in\mathcal A$ there exists $c\in\mathcal A$ such that 
\begin{align}
\bX_a\cap\bX_b=\bX_c.
\label{171028}
\end{align}
For the  example of Subsection \ref{subsec:A principle example, atomic $N$-body
  Hamiltonians} the elements of $\vA$ are  {cluster
  decompositions}, but here  $\vA$ is  an abstract index set. We
 order $\vA$ by  writing $a\leq b$ (or equivalently as $b\geq a$) if
$\bX_a\supseteq \bX_b$. 
It is assumed that there exist
$a_{\min},a_{\max}\in \vA$ such that 
$$\bX_{a_{\min}}=\bX,\quad
\bX_{a_{\max}}=\{0\}.$$ 
\begin{comment}
For a chain of cluster decompositions $a_1\lneq  \cdots \lneq  a_k$
the number $k$ is the \textit{length} of the chain,
and such a chain is  \textit{connecting} $a=a_1$ and $b=a_k$. 
For any $a\in\vA$ we denote 
the maximal length of all the chains connecting $a$ and $a_{\max}$
by $\# a$: 
$$\# a=\max\{k\,|\, a=a_1\lneq\dots \lneq a_k=a_{\max}\};\quad 
\# a_{\max}=1.$$
We say that the family $\{\bX_a\}_{a\in\vA}$ is of {\it$N$-body type} 
if $\# a_{\min}=N+1$. 
To avoid confusion we remark that $(N+1)$ number of moving particles 
form an $N$-body system after separation of the center of mass,
cf. the atomic model of Subsection \ref{subsec:A principle example, atomic $N$-body
  Hamiltonians}.  
\end{comment}

Let $\bX^a\subseteq\bX$ be the orthogonal complement of $\bX_a\subseteq \bX$,
and denote the associated orthogonal decomposition of $x\in\bX$ by 
$$x=x^a\oplus x_a=\pi^ax\oplus \pi_ax\in\bX^a\oplus \bX_a.$$ 
The component $x_a$ may be  called the \emph{inter-cluster coordinates},
and $x^a$ the \emph{internal coordinates}. 
We note that the family $\{\mathbf X^a\}_{a\in \mathcal A}$ is closed under addition:
For any $a,b\in\mathcal A$ there exists $c\in\mathcal A$ such that 
$$\bX^a+\bX^b=\bX^c,$$
cf.\ \eqref{171028}.
 
A real-valued measurable function $V\colon\bX\to\mathbb R$ is 
a \textit{potential of many-body type} 
if there exist real-valued measurable functions
$V_a\colon\bX^a\to\mathbb R$ such that 
\begin{align}
V(x)=\sum_{a\in\mathcal A}V_a(x^a)\ \ \text{for }x\in\mathbf X.
\label{eq:170927}
\end{align} We take $V_{a_{\min}}=0$ (without loss of
generality). We impose throughout the paper the
following condition for    $a\neq a_{\min}$. By definition $\N_0=\N\cup\set{0}$.
\begin{cond}\label{cond:smooth2wea3n12}
    There exists $\mu\in (0,1)$  such that for all $a\in \vA\setminus\set{a_{\min}}$ the
    potential $V_a(x^a)=V^a_{\rm sr}(x^a)+V^a_{\rm lr}(x^a)$, where
    \begin{enumerate}
    \item $V^a_{\rm sr}(-\Delta_a+1)^{-1}$ is compact and  $\abs{x^a}^{1+\mu}
      V^a_{\rm sr}(-\Delta_a+1)^{-1}$ is bounded.
    \item $V^a_{\rm lr}\in C^\infty$ and for all $\gamma\in \N_0^{\dim \mathbf X^a}$
      \begin{align*}
        \partial^\gamma V^a_{\rm lr}(x^a)=\vO(\abs{x^a}^{-\mu-|\gamma|}).
      \end{align*}
\end{enumerate}
\end{cond}

The \emph{sub-Hamiltonian} $H^a$ associated with a cluster 
decomposition $a\in\vA$ is defined as follows. 
For $a=a_{\min}$  we define
$H^{a_{\min}}=0$ on $\mathcal H^{a_{\min}}=L^2(\set{0})=\mathbb C. $
For $a\neq a_{\min}$ 
we introduce 
\begin{equation*}
V^a(x^a)=\sum_{b\leq a} V_{b}(x^b)
,\quad
x^a\in \bX^a,
\end{equation*} 
and then 
\begin{align*}
 H^a=-
\Delta_{x^a} +V^a\ \ 
\text{on }\mathcal H^a=L^2(\bX^a).
\end{align*} 
We abbreviate 
\begin{align*}
V^{a_{\max}}=V,\quad
 H^{a_{\max}}=H,\quad 
 \mathcal H^{a_{\max}}=\mathcal H.
\end{align*}

The \textit{thresholds} of the \emph{full Hamiltonian}  $H$ are the
eigenvalues of the (proper) sub-Hamiltonians $H^a$; 
$a\in  \vA':=\vA\setminus
  \{a_{\max}\}$.
 Equivalently stated  the set of thresholds is 
\begin{align*}
% \label{eq:thres}
 \vT (H):= \bigcup_{a\in\vA'} \sigma_{\pupo}( H^a).
\end{align*} Note that for $a\in\vA\setminus
  \{a_{\min},a_{\max}\}$ the family $\{\mathbf X_b\cap\mathbf X^a\}_{b\leq a}$ forms 
a family of subspaces of many-body type in $\mathbf X^a$. This
self-similarity structure is useful for induction arguments involved
in the proofs of various  well-known properties: 
 We recall  (see for example \cite {AIIS}) that under 
Condition \ref{cond:smooth2wea3n12} 
 the set 
$\vT (H)$ is closed and   at most countable. Moreover the set of non-threshold eigenvalues
  is discrete in $\R\setminus \vT (H)$,  and it 
 can only  accumulate  at  points in
$\vT (H)$  from below.  
The essential spectrum of $H$ is given by the formula
$\sigma_{\ess}(H)= \bigl[\min \vT(H),\infty\bigr)$. We refer to the
mimimum example $\vA= 
  \{a_{\min},a_{\max}\}$ as the \emph{one-body model} in which case $\vT(H)=\set{0}$.

Define the \emph{Sobolev spaces $\mathcal H^s$ of order $s\in\mathbb R$ 
associated with $H$} as
\begin{align}
\mathcal H^s=(H-E)^{-s/2}\mathcal H;\quad
E=\min\sigma(H)-1.
\label{eq:17111317}
\end{align}
We note that the form domain of $H$ is 
$\vH^1=H^1(\mathbf X)$ and that the operator domain $
\mathcal D(H)=\mathcal H^2=H^2(\mathbf X)$ (i.e. given by standard
Sobolev spaces). It is standard to consider $\vH^{-1}$ and $\vH^{-2}$
as the corresponding dual spaces. Denote $R(z)=(H-z)^{-1}$ for
$z\notin \sigma(H)$.

Consider   and fix $\chi\in C^\infty(\mathbb{R})$ such that 
\begin{align}
\chi(t)
=\left\{\begin{array}{ll}
0 &\mbox{ for } t \le 4/3, \\
1 &\mbox{ for } t \ge 5/3,
\end{array}
\right.
\quad
\chi'\geq  0,
\label{eq:14.1.7.23.24}
\end{align} and such that the following properties are fulfilled:
\begin{align*}
  \sqrt{\chi}, \sqrt{\chi'}, (1-\chi^2)^{1/4} ,
  \sqrt{-\parb{(1-\chi^2)^{1/2} }'}\in C^\infty.
\end{align*} We define correspondingly $\chi_+=\chi$ and
$\chi_-=(1-\chi^2)^{1/2} $ and record that
\begin{align*}
  \chi_+^2+\chi_-^2=1\mand\sqrt{\chi_+}, \sqrt{\chi_+'}, \sqrt{\chi_-}, \sqrt{-\chi_-'}\in C^\infty.
\end{align*}

We shall use the standard notation 
$\inp{x}:=\parb{1+\abs{x}^2}^{1/2}$ for  $x\in \mathbf X$ (or more generally for
any $x$ in a normed space).  If $T$ is an operator on a
Hilbert space $\vG$ and $\varphi\in \vG$ then
$\inp{T}_\varphi:=\inp{\varphi,T\varphi}$. We denote the space of  bounded operators 
from one  (general) Banach space $X$ to another one $Y$ by $\vL(X,Y)$ 
and abbreviate $\mathcal L(X)=\mathcal L(X,X)$. The dual space of $X$
is denoted by $X^*$.

To define \emph{Besov spaces 
associated with the multiplication operator
$|x|$ on $\vH$}  
let
\begin{align*}
F_0&=F\bigl(\bigl\{ x\in \mathbf X\,\big|\,\abs{x}<1\bigr\}\bigr),\\
F_m&=F\bigl(\bigl\{ x\in \mathbf X\,\big|\,2^{m-1}\le \abs{x}<2^m\bigr\} \bigr)
\quad \text{for }m=1,2,\dots,
\end{align*}
where $F(U)=F_U$ is the sharp characteristic function of any given  subset
$U\subseteq {\mathbf X}$. 
The Besov spaces $\mathcal B =\mathcal B(\mathbf X)$, $\mathcal
B^*=\mathcal B^*(\mathbf X)$ and $\mathcal B^*_0=\mathcal
B^*_0(\mathbf X)$ are then given  as 
\begin{align*}
\mathcal B&=
\bigl\{\psi\in L^2_{\mathrm{loc}}(\mathbf X)\,\big|\,\|\psi\|_{\mathcal B}<\infty\bigr\},\quad 
\|\psi\|_{\mathcal B}=\sum_{m=0}^\infty 2^{m/2}
\|F_m\psi\|_{{\mathcal H}},\\
\mathcal B^*&=
\bigl\{\psi\in L^2_{\mathrm{loc}}(\mathbf X)\,\big|\, \|\psi\|_{\mathcal B^*}<\infty\bigr\},\quad 
\|\psi\|_{\mathcal B^*}=\sup_{m\ge 0}2^{-m/2}\|F_m\psi\|_{{\mathcal H}},
\\
\mathcal B^*_0
&=
\Bigl\{\psi\in \mathcal B^*\,\Big|\, \lim_{m\to\infty}2^{-m/2}\|F_m\psi\|_{{\mathcal H}}=0\Bigr\},
\end{align*}
respectively.
Denote the standard \emph{weighted $L^2$ spaces} by 
$$
L_s^2=L_s^2(\mathbf X)=\inp{x}^{-s}L^2(\mathbf X)\ \ \text{for }s\in\mathbb R ,\quad
L_{-\infty}^2=\bigcup_{s\in\R}L^2_s,\quad
L^2_\infty=\bigcap_{s\in\mathbb R}L_s^2
.
$$ 
Then for any $s>1/2$
\begin{align}\label{eq:nest}
 L^2_s\subsetneq \mathcal B\subsetneq L^2_{1/2}
\subsetneq \mathcal H
\subsetneq L^2_{-1/2}\subsetneq \mathcal B^*_0\subsetneq \mathcal B^*\subsetneq L^2_{-s}.
\end{align}

\begin{subequations}
Under a more general condition than Condition
\ref{cond:smooth2wea3n12} it is demonstrated in  \cite{AIIS} that the
following limits exist   locally
                                                              uniformly
                                                              in $\lambda\not\in \vT
(H)\cup\sigma_{{\rm pp}}(H)$:
\begin{align}\label{eq:LAPbnda}
  R(\lambda\pm \i
  0)=\lim _{\epsilon\to 0_+} \,R(\lambda\pm\i
    \epsilon)\in \vL\parb{L^2_s,L^2_{-s}}\text{ for any }s>1/2.
\end{align} Furthermore, 
   in the strong weak$^*$-topology, 
  \begin{align}\label{eq:BB^*a}
    \begin{split}
      R(\lambda\pm \i
  0)&=\swslim _{\epsilon\to 0_+} \,R(\lambda\pm\i
    \epsilon)\in \vL\parb{\vB,\vB^*}\\& \text{
                                                              with a
                                                              locally
                                                              uniform
                                                              norm bound in
                                                              }\lambda\not
                                                              \in \sigma_{{\rm pp}}(H)\cup\vT
(H).
    \end{split}
\end{align}
\end{subequations}

\subsection{A one-body effective potential and  one-body radial
  limits}\label{subsec: -body effective potential and a $1$-body radial limit}

We introduce   $I^{\rm
  sr}_a=\sum_{b\not\leq a}V_{\rm sr}^b$, $I^{\rm lr}_a=\sum_{b\not\leq
  a}V_{\rm lr}^b$, $I_a=I^{\rm sr}_a+I^{\rm
  lr}_a$  and
\begin{align*}
  \breve I_a=\breve I_a(x_a)=I^{\rm lr}_a(x_a)\prod_{b\not\leq a} \,\,\chi_+(|\pi^b
  x_a| \ln \inp{x_a}/\inp{x_a} );\quad a\neq a_{\max}.
\end{align*}  This `regularization' $\breve I_a$  appears  in \cite{HS} for the free channel,
i.e. for $a=a_{\min}$.  It is  a one-body potential fulfilling for any
$\breve\mu\in (0,\mu)$ the bounds 
\begin{align}\label{eq:brevePotential}
       \partial^\gamma \breve I_a(x_a)=\vO(\abs{x_a}^{-\breve\mu-|\gamma|}).
      \end{align} For notational
     convenience we take from this point and throughout the paper $\breve\mu=\mu$,
     i.e. more precisely we will assume \eqref{eq:brevePotential} with
     $\breve\mu$ replaced by $\mu$. We let 
      $K_a(\cdot,\lambda)$,  $\lambda>0$, denote the corresponding approximate solution
      to the eikonal equation $\abs{\nabla_{x_a} K_a}^2+ \breve
      I_a=\lambda$ as  taken 
       from \cite {Is,II}. 
More precisely writing  $K_a(x_a,\lambda)=\sqrt
      \lambda \abs{x_a}-k_a(x_a,\lambda)$  the following properties
      are fulfilled with $\R_+:=(0,\infty)$ and  $ n_a:=\dim \mathbf X_a$. The functions $k_a\in
      C^\infty(\mathbf X_a \times \R_+)$ and:
      \begin{enumerate}[1)]
      \item For any compact $\Lambda \subset \R_+$ there exists $R>1$
        such that for all $\lambda\in \Lambda$ and all $x_a\in \mathbf X_a$
        with $\abs{x_a}>R$
        \begin{align*}
          2\sqrt{\lambda}\,\tfrac{\partial} {\partial\abs{x_a}}k_a=\breve I_a(x_a)+\abs{\nabla_{x_a}k}^2.
        \end{align*}
      \item For all multiindices $\gamma\in \N_0^{n_a}$, $m\in \N_0$ and
         compact  $\Lambda \subset \R_+$
        \begin{align*}
          \abs[\big]{{\partial}_
          {{x_a}}^\gamma {\partial}_
          {\lambda}^mk_a}\leq C\inp {x_a}^{1-\abs{\gamma}-\mu}\text{
          uniformly in }\lambda\in \Lambda.
        \end{align*}
      \end{enumerate}

We drop for the moment the subscript $a$ and consider  the one-body
Hamiltonian
\begin{align*}
  \breve h=\breve h_a=-\Delta +\breve I\quad { on }\quad L^2(\R^n),
\end{align*}
  identifying
$\mathbf X_a=\R^n$ where   $n=n_a$. Let $\vB$,  $\vB^*$ and $\vB_0^*\subset \vB^*$ denote 
corresponding Besov spaces, cf. Subsection \ref{subsec:body Hamiltonians, limiting absorption
  principle  and notation}, and let  $C_a=\mathbf X_a\cap\S^{n_a-1}$. 
%For $U\subseteq \R^n$ the notation
%$F_U=F(U)$ stands    for  multiplication by $1_U$.
\begin{lemma}\label{lemma:basic}
  Let $U$ be an open subset of $\R^n$ such that $U'=U\cap
  \S^{n-1} \neq \emptyset$ and $U\cap\set{\abs{x}\geq 1}=\set{x=cx'|\,
    c\in [1,\infty), x'\in U'}$. Let for any $g\in C_\c^\infty (U')$,
  $\lambda> 0$,  
\begin{align*}
  v^{\pm}( x)=v_\lambda^{\pm} [g]( x)= \chi_+(\abs{x}) \abs{x}^{(1-n)/2}
                                 \e^{\pm \i K_a(x,\lambda)}g(\hat x); \quad \hat x=x/\abs{x}.
\end{align*} Suppose   $\breve u\in \vB^*\cap H^2_{\mathrm {loc}}$ and
  $F_U(\breve h-\lambda)\breve u\in \vB$. Then 
  \begin{align}
    \label{eq:lim}
    \lim_{\rho\to \infty}\,\rho^{-1}\int _{\abs{x}<\rho}\overline {v^{\pm} (
      x)}\breve u(x)\, \d x=\pm 2^{-1}\i \lambda^{-1/2}\parb{\inp{
        v^{\pm},(\breve h-\lambda)\breve u}-\inp{ (\breve
        h-\lambda)v^{\pm},\breve u}}.
  \end{align} 
\end{lemma}
\begin{proof} First
we  compute  for any $g\in C_\c^\infty (C_a) $ and $\lambda
>0$ 
\begin{align}\label{eq:basiC0}
  (\breve {h}-\lambda)v^{\pm}_{\lambda} [g]\in L^2_s(\R^n)=\inp{x_a}^{-s}L^2(\R^n)\text{
    for any  }s\in \parb{\tfrac 12,\tfrac 12 +\mu}.
\end{align} In particular 
$v^{\pm}\in \vB^*$ and $ (\breve h-\lambda)v^{\pm}\in \vB$, whence 
     the right-hand side of \eqref{eq:lim} is well-defined. For any
  $\epsilon\in(0,1/3)$ choose a decreasing function $\chi_\epsilon \in
  C^\infty(\mathbb{R})$ such that 
\begin{align*}
\chi_\epsilon(t)
=\left\{\begin{array}{ll}
1 &\mbox{ for } t \le \epsilon, \\
1+\tfrac 32 \epsilon-t &\mbox{ for } 3\epsilon \leq t \leq 1, \\
0 &\mbox{ for } t \ge 1+2\epsilon,
\end{array}
\right. 
\end{align*} and $\chi'_\epsilon\geq -1$. Letting
$\chi_{\epsilon,\rho}=\chi_\epsilon(\abs{x}/\rho)$, $\rho>1$,  we
compute the right-hand side of \eqref{eq:lim} as
\begin{align*}
   \mp 2^{-1} \lambda^{-1/2}\lim_{\rho\to \infty}\,\inp{
  v^{\pm},\i \comm{\breve h,\chi_{\epsilon,\rho}}\breve u}=-\lim_{\rho\to \infty}\,\rho^{-1}\inp{
  v^{\pm},\chi'_\epsilon(\abs{\cdot}/\rho)\breve u}.
\end{align*} By \cas we can write
\begin{align*}
  -\rho^{-1}\inp{
  v^{\pm},\chi'_\epsilon(\abs{\cdot}/\rho)\breve u}&=\rho^{-1}\int _{3\epsilon<\abs{x}/\rho<1}\overline {v^{\pm} (
                                              x)}\breve u(x)\, \d x +\vO(\sqrt \epsilon)\\
                                            &=\rho^{-1}\int _{\abs{x}/\rho<1}\overline {v^{\pm} (
                                              x)}\breve u(x)\, \d x +\vO(\sqrt \epsilon),
\end{align*} where the bounds are uniform in $\rho>1$. Whence
\begin{align*}
  -\lim_{\epsilon\to 0}\lim_{\rho\to \infty}\,\rho^{-1}\inp{
  v^{\pm},\chi'_\epsilon(\abs{\cdot}/\rho)\breve u}=\lim_{\rho\to \infty}\,\rho^{-1}\int _{\abs{x}/\rho<1}\overline {v^{\pm} (
                                              x)}\breve u(x)\, \d x.
\end{align*}
\end{proof}
\begin{remarks*} 
\begin{enumerate}[i)]
  \item \label{item:As1} The concept of radial limits introduced in Lemma
  \ref{lemma:basic} (and extended to  the $N$-body case in
  Corollary \ref{cor:radi-limits-chann}  below) is rather  weak. A
  stronger concept of  radial limits 
  (conforming with the weak one) 
  appears in several papers, see \cite{GY, HS, II, Is, IS2}. The
  latter requires  radiation condition bounds  and leads to stronger
  conclusions, see Section \ref{sec:one-body  matrices}. 
\item \label{item:As2} We learned
  about the  assertion Lemma
  \ref{lemma:basic}  from \cite{As}, where it appears   with a very
  different proof. Similarly the following Subsections \ref{Radial limits}--\ref{subsec:Time-dependent
  scattering theory} have  overlap with  \cite{As} although the goals
and 
presentation are  different. A  main result of these subsections is
the existence and weak continuity of entries of the
scattering matrix  under a suitable  decay condition on
the channel bound states.
\end{enumerate}
  \end{remarks*}

\subsection{Radial limits for channels}\label{Radial limits}
Consider any  channel
$\alpha =(a,\lambda^\alpha, u^\alpha)$, i.e.   $a\in\vA'$,  $u^\alpha\in
\mathcal H^a$ and 
$(H^a-\lambda^\alpha)u^\alpha=0$, obeying the following decay property.
  \begin{cond}\label{cond:decayBNstates} With $\alpha
    =(a,\lambda^\alpha, u^\alpha)$, the (sub) bound state
\begin{align}\label{eq:decay}
  u^\alpha\in \vD (\inp{x^a}^s)\text{ for some }s>\tfrac52-\mu.
\end{align}  
\end{cond}
 Alternatively stated,   Condition \ref{cond:decayBNstates} reads $u^\alpha\in L^2_s(\mathbf X^a)=\inp{x^a}^{-s}L^2(\mathbf X^a)$ for
some $s>\tfrac52-\mu$.  

Let 
 $ \vH_a=L^2(\mathbf X_a)$ and $\pi_\alpha\in \vL(\vH,\vH_a)$   be given by
 $(\pi_\alpha v)(x_a)=\inp{u^\alpha,v(\cdot,x_a)}$. It is a
 consequence of 
\eqref{eq:decay} that 
\begin{align}\label{eq:B2bnds}
  \pi_\alpha\in \vL(\vB^*(\mathbf X),\vB^*(\mathbf X_a))\cap \vL(\vB(\mathbf X),\vB(\mathbf X_a))
\end{align}
(only $ u^\alpha\in \vB(\mathbf X^a)$ is needed at this point). 

Recalling   $C_a=\mathbf X_a\cap\S^{n_a-1}$, $
                                    n_a=\dim \mathbf X_a$,  we let  $ C'_a=C_a\setminus \cup_{{b\not\leq a}}\mathbf X_b$. 

Suppose $u\in \vB^*(\mathbf X)$ obeys $(H-\lambda)u\in \vB(\mathbf
X)$; $\lambda>\lambda^\alpha$. We  shall then introduce
$Q^\pm_{\lambda,\alpha}(u)\in L^2(C_a)$ by the  weak radial limit
procedure of Subsection \ref{subsec: -body effective potential
  and a $1$-body radial limit}.
For this end we let, parallel to Lemma \ref{lemma:basic} and as a first step,   for any $g\in C_\c^\infty (C'_a) $
\begin{subequations}
\begin{align}\label{eq:Riesz}
  {Q^\pm_{\lambda,\alpha}(u)}[g]=\lim_{\rho\to \infty}\,\rho^{-1}\int _{\abs{x_a}<\rho}\overline {v_{\lambda,\alpha}^{\pm} [g](
      x_a)}(\pi_\alpha u)(x_a)\, \d x_a,
\end{align} where 
\begin{align*}
   v^{\pm}_{\lambda,\alpha} [g]( x_a)= \chi_+(\abs{x_a}) \abs{x_a}^{(1-n_a)/2}
                                      \e^{\pm \i
                                        K_a(x_a,\lambda_\alpha)}g(\hat x_a);\quad
                                    \lambda_\alpha=\lambda-\lambda^\alpha.
 \end{align*} We claim that indeed 
                                    Lemma \ref{lemma:basic}
                                    applies to $\breve u=\pi_\alpha
                                    u\in\vB^*(\mathbf X_a)$,  considering any open
                                    $U'\subseteq
                                    \overline{U'}\subseteq C_a'$ and
                                    any $g\in C_\c^\infty (U') $,
                                    yielding the well-definedness of
                                    \eqref{eq:Riesz} and  the formula
\begin{align}\label{eq:small0}
  {Q^\pm_{\lambda,\alpha}(u)}[g]= \pm 2^{-1}\i \lambda_\alpha^{-1/2}\parb{\inp{
        J_\alpha v^{\pm}_{\lambda,\alpha} [g],( H-\lambda)u}-\inp{ (
        H-\lambda)J_\alpha v^{\pm}_{\lambda,\alpha} [g], u}},
 \end{align} where the out- and incoming quasi-modes $J_\alpha
v^{\pm}_{\lambda,\alpha} [g]:=u^\alpha\otimes
v^{\pm}_{\lambda,\alpha}[g]$. 
\end{subequations} 
\begin{lemma}\label{lemma:radi-limits-chann00} With Condition
  \ref{cond:decayBNstates} imposed on  $\alpha=(a,\lambda^\alpha,
  u^\alpha)$,  for any  $\lambda>\lambda^\alpha$ and $g\in C_\c^\infty (C'_a) $
  \begin{align}\label{eq:besFull}
    (H-\lambda) J_\alpha
v^{\pm}_{\lambda,\alpha} [g]\in\vB(\mathbf X).
  \end{align} Furthermore, if $u\in \vB^*(\mathbf X)$ obeys  $(H-\lambda)u\in \vB(\mathbf
X)$, then \eqref{eq:small0} is fulfilled with the
  left-hand side given as  the limit \eqref{eq:Riesz}.
\end{lemma}
\begin{proof}
Note first that  with \eqref{eq:decay} at our disposal 
\begin{align}\label{eq:I-I}
  (I_a-\breve I_a) J_\alpha v^{\pm}_{\lambda,\alpha} [g]\in L^2_s(\mathbf X)\text{
    for some }s>1/2,
\end{align} which may be proved as follows using  the elementary bound 
\begin{align*}
  \sqrt 2 \inp {\pi^b x_a+y}\geq \inp {\pi^b x_a}\inp {y}^{-1};\quad y\in \mathbf X.
\end{align*}  For any $b\not\leq a$ and $\kappa\in [0,1]$ we 
estimate (for $\hat x_a\in \supp g$)
\begin{subequations}
 \begin{align}\label{eq:Ta1}
    \begin{split}
     \abs[\big]{V_{\rm lr}^b(x)- V_{\rm lr}^b(x_a)}^\kappa&=\abs[\Big]{\int_0^1 \pi^bx^a\cdot (\nabla
  V_{\rm lr}^b)\parb{\pi^b(x_a+tx^a)}\,\d
    t}^\kappa\\&=\inp{x^a}^{(\mu+2)\kappa}\vO\parb{\abs{x_a}^{-(\mu+1)\kappa}},
 \end{split}\\
\abs[\big]{V_{\rm lr}^b(x)-V_{\rm lr}^b(x_a)}^{1-\kappa}&=\inp{x^a}^{\mu(1-\kappa)}\vO\parb{\abs{x_a}^{-\mu(1-\kappa)}}.\label{eq:Ta2}
\end{align} We choose $\kappa> 1-\mu$, slightly bigger, and then
use \eqref{eq:decay} to assure that $u^\alpha\in
L_{\delta}^2(\mathbf X^a)$ for some
$\delta>(\mu+2)\kappa+\mu(1-\kappa)+1/2$, concluding  that
\begin{align*}
   \parb{V_{\rm lr}^b(x)- V_{\rm lr}^b(x_a)} J_\alpha v^{\pm}_{\lambda,\alpha} [g]\in L^2_s(\mathbf X)\text{
    for some }s>1/2.
\end{align*} 
 Similarly
\begin{align}\label{eq:Ta3}
   \inp {\pi^b x}^{-(1+\mu)\kappa}\leq C\inp {\pi^b x_a}^{-(1+\mu)\kappa}\inp{ x^a}^{(1+\mu)\kappa}
\end{align} leads to the requirement (by chosing $\kappa> (1+\mu)^{-1}$, slightly bigger)
that  $u^\alpha\in
L_{\delta}^2(\mathbf X^a)$ for some $\delta>3/2$ (which is consistent with \eqref{eq:decay}) for
concluding  that 
\begin{align*}
   V_{\rm sr}^bJ_\alpha v^{\pm}_{\lambda,\alpha} [g]\in L^2_s(\mathbf X)\text{
    for some }s>1/2.
\end{align*} 
 \end{subequations}

Next we write
\begin{align*}
  H=H^a\otimes I+I\otimes \brh_a+ (I_a-\breve I_a)
\end{align*} 
 and  conclude  from \eqref{eq:basiC0} and \eqref{eq:I-I} that
\begin{align}\label{eq:small}
  (H-\lambda) J_\alpha v^{\pm}_{\lambda,\alpha} [g]\in L^2_s(\mathbf X)\text{
    for some }s>1/2,
\end{align} 
 which is stronger than \eqref{eq:besFull}. Consequently  the second
 term on the
right-hand side of \eqref{eq:small0} is well-defined. The
well-definedness of the first term to the right follows from \eqref{eq:B2bnds} (implying   $J_\alpha=\pi_\alpha^*\in \vL(\vB^*(\mathbf
X_a),\vB^*(\mathbf X))$ by       Fubini's
theorem).

Finally we can  justify  \eqref{eq:Riesz} and \eqref{eq:small0} by
combining    \eqref{eq:basiC0}, \eqref{eq:B2bnds},  the above arguments for \eqref{eq:besFull}  and Lemma
\ref{lemma:basic}. We can write the  right-hand side
of \eqref{eq:small0}  exactly as the right-hand side
of \eqref{eq:lim} with $\breve u:=\pi_\alpha u\,(\in \vB^*(\mathbf
X_a))$ and $\lambda$ there replaced by
 $\lambda_\alpha\,(>0)$. With $F_U$ given as
in Lemma \ref{lemma:basic} we   can check the condition $F_U(\breve
h_a-\lambda_\alpha)\breve u\in \vB(\mathbf
X_a)$ of   Lemma
\ref{lemma:basic} as follows. By an approximation argument we see that
\begin{align*}
  F_U\pi_\alpha (H-\lambda) u-F_U (\breve
h_a-\lambda_\alpha)\breve u = F_U\pi_\alpha  (I_a-\breve I_a)u.
\end{align*} By the arguments for \eqref{eq:besFull} the right-hand
side is in $\vB(\mathbf X_a)$, and by assumption also the first term
to the left is in $\vB(\mathbf X_a)$. Consequently indeed $F_U(\breve
h_a-\lambda_\alpha)\breve u\in \vB(\mathbf
X_a)$. Note also that the right-hand
side when implemented in an expansion of the first  term of
\eqref{eq:small0} cancels with a term arising from  a similar expansion of the second term of \eqref{eq:small0}.
\end{proof}

\begin{corollary}\label{cor:radi-limits-chann}
  For any channel $\alpha=(a,\lambda^\alpha, u^\alpha)$ obeying Condition \ref{cond:decayBNstates}
  and any $u\in \vB^*(\mathbf X)$ obeying  $(H-\lambda)u\in \vB(\mathbf X)$ for some  $\lambda>\lambda^\alpha$
  there exist the weak limits
  \begin{align*}
    Q^\pm_{\lambda,\alpha}(u)=\wLlim_{\rho\to \infty}\,\rho^{-1}\int
    _0^\rho r^{(n_a-1)/2}
                                      \e^{\mp \i
                                        K_a(r\cdot,\lambda_\alpha)}\parb{\pi_\alpha u}(r\cdot)\, \d r
  \end{align*}
\end{corollary}
\begin{proof}
  By \eqref{eq:B2bnds} and the  Riesz theorem the right-hand side of \eqref{eq:Riesz} belongs
$L^2(C_a)$ as a bounded linear functional in 
 $g\in L^2(C_a')=L^2(C_a)$. The result then follows from  Fubini's theorem.
\end{proof}

A useful example is given by $u=R(\lambda\pm \i 0)f$, where $f\in
\vB(\mathbf X)$,
\begin{align}\label{eq:defEI}
  \lambda \in
\vE^\alpha:=I^\alpha \setminus \parb{\sigma_{\pp}(H)\cup
  \vT(H)}\quad\text{with}\quad I^\alpha:=(\lambda^\alpha,\infty),
\end{align}  and \eqref{eq:BB^*a} is used.  
  By \eqref{eq:small0} and Lemma
\ref{lemma:Sommerfeld} \ref{item:1s}  (stated below) it follows that
\begin{align}\label{eq:zero}
  Q^\pm_{\lambda,\alpha}(R(\lambda\mp \i 0)f)=0.
\end{align} On the other hand the functions
$Q^\pm_{\lambda,\alpha}(R(\lambda\pm \i 0)f)$ are  in general nonzero,
see Lemmas \ref{lemma:Sommerfeld} \ref{item:2s} and
\ref{lemma:poiss-oper-geom} below. For later use we record  that, thanks to  \eqref{eq:Riesz},
\begin{align}\label{eq:besBaGood}
  Q^\pm_{\lambda,\alpha}R(\lambda\pm \i
0)=Q^\pm_{\lambda,\alpha}\parb{R(\lambda\pm \i
0)(\cdot)}\in\vL\parb{\vB(\mathbf X),L^2(C_a)},
\end{align} and as functions of $\lambda\in \vE^\alpha$ these
operators are locally bounded.

\subsection{Poisson operators and  geometric scattering matrix}\label{Wave matrices}

Consider  a channel
$\alpha= (a,\lambda^\alpha, u^\alpha)$ obeying Condition \ref{cond:decayBNstates},  and
consider the
quasi-modes $J_\alpha
v^{\pm}_{\lambda,\alpha} [g]=u^\alpha\otimes
v^{\pm}_{\lambda,\alpha}[g]$, $g\in C_\c^\infty (C_a')$, cf.  Subsection \ref{Radial limits}.  The
                                    assertion \eqref{eq:small} allows
                                    us 
 to define for $\lambda \in
\vE^\alpha$ the exact
generalized eigenfunctions in $\vB^*(\mathbf X)$, 
\begin{align*}
  P^\mp_{\lambda,\alpha} [g]&= J_\alpha v^{\mp}_{\lambda,\alpha}
                              [g]-R(\lambda\pm \i 0)(H-\lambda) J_\alpha v^{\mp}_{\lambda,\alpha}
                              [g],\\
  \check P^\mp_{\lambda,\alpha} [g]&= J_\alpha v^{\mp}_{\lambda,\alpha}
                              [g]-R(\lambda\mp \i 0)(H-\lambda) J_\alpha v^{\mp}_{\lambda,\alpha}
                              [g].
\end{align*} 
\begin{lemma}\label{lemma:Sommerfeld}   For any channel
  $(a,\lambda^\alpha, u^\alpha)$ obeying Condition \ref{cond:decayBNstates}, $\lambda \in
\vE^\alpha$  and  $g\in
  C_\c^\infty (C_a')$: 
  \begin{enumerate}[1)]
  \item \label{item:1s} $\check P^\mp_{\lambda,\alpha} [g]=0$,
  \item  \label{item:2s} $g=Q^\mp_{\lambda,\alpha}\parb{P^\mp_{\lambda,\alpha}
  [g]}$,
\item  \label{item:3s} $g=Q^\mp_{\lambda,\alpha}\parb{R(\lambda\mp \i 0)(H-\lambda) J_\alpha v^{\mp}_{\lambda,\alpha}
                              [g]}$.
\end{enumerate}
  \end{lemma}
  \begin{proof}
    A computation using \eqref{eq:self-similar} and \eqref{eq:r0} (stated below)
     shows that the
    eigenfunctions $\check P^\mp_{\lambda,\alpha} [g]$ are purely
    incoming/outgoing in the sense of \cite[Corollary
    1.10]{AIIS}, see Appendix \ref{sec:AppendixO} for details. Whence \ref{item:1s} follows from the conclusion of
    this version of the Sommerfeld uniqueness result. 
    Clearly \ref{item:2s} and 
\ref{item:3s}  follow from \eqref{eq:zero} and  \ref{item:1s}, respectively.
\end{proof} The \emph{Poisson operators}  $P^\mp_{\lambda,\alpha} $
fulfill  the following properties.
\begin{lemma}\label{lemma:poiss-oper-geom} For any channel
  $(a,\lambda^\alpha, u^\alpha)$ obeying Condition \ref{cond:decayBNstates}, $g\in C_\c^\infty
  (C_a')$ and   any $f\in
\vB(\mathbf X)$
\begin{subequations}
 \begin{align}\label{eq:1Form}
   \pm 2^{-1}\i \lambda_\alpha^{-1/2}\inp{ P^\pm_{\lambda,\alpha}
  [g],f}&=\inp{g,Q^\pm_{\lambda,\alpha}(R(\lambda\pm \i 0)f)},\\
 P^\pm_{\lambda,\alpha}&\in\vL \parb{L^2(C_a),\vB^*(\mathbf X)}.\label{eq:2Form}
\end{align}  
\end{subequations} In particular $P^\pm_{\lambda,\alpha}
$ depends 
strongly weak$^*$-continuously of $\lambda \in
\vE^\alpha$. In
the weaker topology of $\vL(L^2(C_a),L^2_{-s}(\mathbf X))$, 
 $s>1/2$, $ P^\pm_{\lambda,\alpha}$ are strongly  continuous in $\lambda \in
\vE^\alpha$.
\begin{proof}
  The formulas \eqref{eq:1Form} follow  by applying \eqref{eq:small0} to $u=R(\lambda\pm
  \i 0)f$. The assertion \eqref{eq:2Form} follows
from \eqref{eq:1Form} and \eqref{eq:besBaGood}. The continuity  assertions follow  from checking the computation
  \eqref{eq:small} and combining with  regularity of
  the function $K_a$ and  continuity of the boundary values of the
  resolvent, cf.  \eqref{eq:LAPbnda}.
\end{proof}
  
\end{lemma}

For two channels  $\alpha=(a,\lambda^\alpha, u^\alpha)$ and  $\beta=(b,\lambda^\beta, u^\beta
)$ both fulfilling Condition \ref{cond:decayBNstates} the
\emph{geometric scattering matrix}, or rather the component  given by considering  $\alpha$ as incoming and
$\beta$ as outgoing, is given by
\begin{subequations}
 \begin{align}\label{eq:geoS0}
  \Sigma_{\beta\alpha}(\lambda)g=Q^+_{\lambda,\beta}\parb{P^-_{\lambda,\alpha}
  [g]};\quad \lambda \in
\vE^\alpha\cap \vE^\beta ,\, g\in C_\c^\infty (C_a').
\end{align} Alternatively this quantity is given by the radial limit
\begin{align}\label{eq:geoS}
  \Sigma_{\beta\alpha}(\lambda)g=-Q^+_{\lambda,\beta}\parb{R(\lambda+
    \i 0)(H-\lambda) J_\alpha v^{-}_{\lambda,\alpha} [g]}.
\end{align} 
\end{subequations}
 \begin{lemma}\label{lemma:geomS}
The component
$\Sigma_{\beta\alpha}(\lambda)$ of the
geometric scattering matrix extends to an operator in
$\vL \parb{L^2(C_a),L^2(C_b)}$ with 
weakly continuous dependence of $\lambda \in
  \vE^\alpha\cap\vE^\beta$.  
\end{lemma}
\begin{proof}
  It follows from  \eqref{eq:geoS0}, \eqref{eq:Riesz},
  \eqref{eq:B2bnds}  and \eqref{eq:2Form} that the extension $\Sigma_{\beta\alpha}(\lambda)\in
  \vL \parb{L^2(C_a),L^2(C_b)}$ exists. We record that it is locally uniformly
  bounded in $\lambda$.   The weak continuity assertion then follows
  from  \eqref{eq:geoS},  Lemma
\ref{lemma:poiss-oper-geom} and  \eqref{eq:small} (including 
 the underlying computation for  \eqref{eq:small} 
  exhibiting continuity). Indeed for any  $g_a\in C_\c^\infty
  (C_a')$ and $g_b\in C_\c^\infty
  (C_b')$ 
\begin{align*}
   \inp{g_b,Q^+_{\lambda,\beta}(R(\lambda+ \i 0)(H-\lambda) J_\alpha v^{-}_{\lambda,\alpha} [g_a])}=2^{-1}\i \lambda_\beta^{-1/2}\inp{ P^+_{\lambda,\beta}
  [g_b],(H-\lambda) J_\alpha v^{-}_{\lambda,\alpha} [g_a]}.
\end{align*}  Due to the last assertion of Lemma
\ref{lemma:poiss-oper-geom} the first entry of  the inner product to
the right is continuous as an  $L^2_{-s}(\mathbf X)$-valued
function. The second entry is continuous as an  $L^2_{s}(\mathbf X)$-valued
function, provided $s$ is taken slightly larger than $1/2$, so indeed
 the right-hand side  is continuous.
\end{proof}

\subsection{Time-dependent scattering theory and stationary
  representations}\label{subsec:Time-dependent
  scattering theory} For any channel
  $\alpha=(a,\lambda^\alpha, u^\alpha)$ obeying  Condition
  \ref{cond:decayBNstates}   one  can
  easily show the existence of
 \emph{channel wave operators} by the Cook criterion using
 \eqref{eq:Ta1}--\eqref{eq:Ta3} and one-body scattering theory \cite{H0}. Consequently, more precisely,  we introduce
\begin{align}\label{eq:wave_op}
  W_\alpha^{\pm}=\slim_{t\to \pm\infty}\e^{\i tH}J_\alpha\e^{-\i
  (S_a^\pm (p_a,t)+\lambda^\alpha t)},
 \end{align} where $p_a=-\i \nabla_{x_a}$ and $S_a^\pm (\xi_a,\pm|t|)=\pm S_a(\pm\xi_a,|t|)$ is
defined in term of the Legendre transform of the function $
K_a(\cdot, \lambda)$ of Subsection \ref{subsec: -body effective
  potential and a $1$-body radial limit}. This amounts to the
following construction,
citing \cite[Lemma 6.1]{II}:   

There exist an $\mathbf X_a$-valued function $x(\xi, t)$ and a
positive function $\lambda(\xi,t)$ both in $C^\infty\parb{(\mathbf
X_a\setminus\set{0})\times \R_+}$ and satisfying the following requirements.
 For any compact set $B\subseteq \mathbf
X_a\setminus\set{0}$,  there exist $ T, C > 0$ such  that for $\xi\in B$
and $t > T$ 
\begin{enumerate}[1)]
      \item\label{item:1ss} \quad $\xi=\partial_{x_a}K_a\parb{x(\xi, t),\lambda(\xi,t)},\quad t=\partial_{\lambda}K_a\parb{x(\xi, t),\lambda(\xi,t)},$
        
      \item \label{item:2ss}\quad $\abs{x(\xi, t)-2t\xi}\leq C\inp{t}^{1-\mu},\quad \abs[\big]{\lambda(\xi, t)-\abs{\xi}^2}\leq C\inp{t}^{-\mu}.$
\end{enumerate} Then we define 
\begin{align*}
  S_a(\xi,t)=x(\xi, t)\cdot \xi +\lambda(\xi, t)t - K_a(x(\xi, t),
  \lambda(\xi, t));\quad (\xi,t)\in\parb{\mathbf
X_a\setminus{0}}\times \R_+.
\end{align*} 

Note that this function solves the Hamilton-Jacobi equation
\begin{align*}
  \partial_t S_a(\xi,t)=
  \xi^2+\brI_a\parb{\partial_{\xi}S_a(\xi,t)};\quad t> t(\xi), \, \xi\neq 0.
\end{align*}

\begin{remarks}\label{remark:Dollard-time-depend-scatt}
  \begin{enumerate}[i)]
  \item \label{item:Dol1}
  For $\mu> 1/2$ in Condition \ref{cond:smooth2wea3n12} one can use
  the Dollard-approximation (cf. \cite{Do})  rather than the above exact solution $
  S_a$ to the Hamilton-Jacobi equation to define channel wave operators. It reads
  \begin{align*}
    S_{a,\dol}(\xi,t)=t\xi^2+\int_1^t\,\brI_a(2s\xi)\,\d s.
  \end{align*} Introducing   $S_{a,\dol}^\pm (\xi,\pm|t|)=\pm
  S_{a,\dol}(\pm\xi,|t|)$,  the existence of the limits 
\begin{align}\label{eq:Dolwave_op}
  W_{\alpha,\dol}^{\pm}=\slim_{t\to \pm\infty}\e^{\i tH}J_\alpha\e^{-\i
  (S_{a,\dol}^\pm (p_a,t)+\lambda^\alpha t)}
\end{align}  then  follows  by the chain rule  (assuming that
$W_\alpha^{\pm}$ exist) and  from the bounds
\begin{align*}
  \int_1^\infty\,\abs{\brI_a\parb{\partial_{\xi}S_a(\xi,s)}-\brI_a(2s\xi)}\,\d
  s<\infty;\quad \xi\neq 0.
\end{align*}  In turn the latter are  consequences of the bounds
\begin{align}\label{eq:xBNds}
  \abs{\partial_{\xi}S_a(\xi,t)-2t\xi}\leq C\inp{t}^{1-\mu};\quad
  C=C(\xi)<\infty.
\end{align} Note that \eqref{eq:xBNds}  follows from the above
property  \ref{item:2ss} and the fact that
\begin{align*}
  \partial_{\xi}S_a(\xi,t)-x(\xi, t)=0\text{ for }t> t(\xi).
\end{align*}
\item \label{item:Dol2} The construction of  \ref{item:Dol1} applies
  to the example of Subsection \ref{subsec:A principle example, atomic $N$-body
  Hamiltonians} and for that case one can alternatively use  the
function $S_{a,{\rm atom}}(\xi,t)$ defined there. Note the property
\begin{align*}
  \int_1^\infty\,\abs{I_a(2s\xi)-\brI_a(2s\xi)}\,\d
  s<\infty;\quad \xi\in {\mathbf X_a}\setminus\cup_{{c\not\leq a} }\mathbf X_c.
\end{align*} By the chain rule the existence of $ W_{\alpha,{\rm
    atom}}^{\pm}$ then follows from the existence of
$W_{\alpha,\dol}^{\pm}$.

\item \label{item:Dol3}
It is also an elementary fact that the existence, along
with continuity properties of the scattering matrix, to be stated in
Theorems \ref{thm:scat_matrices} and \ref{thm:mainstat-modif-n} (under
different hypotheses) imply similar assertions for the `Dollard' and
`atomic' counterparts discussed above. For example, under the conditions
of Theorem \ref{thm:mainstat-modif-n} 
\begin{align*}
  S_{\beta\alpha,{\rm atom}}(\lambda)=\e^{-\i\theta^+_\beta(\cdot,\lambda)}S_{\beta\alpha}(\lambda)\e^{\i\theta^-_\alpha(\cdot,\lambda)},
\end{align*} where on the right-hand side the multiplication operators to the left and to
the right are both strongly continuous factors (they are given by
explicit integrals). Whence Theorem \ref{thm:princ-example-atom} is an
immediate consequence of Theorem \ref{thm:mainstat-modif-n} and
Corollary \ref{cor:direct}.
\item \label{item:Dol4} For short-range systems, say given by
  Condition \ref{cond:smooth2wea3n12} with the additional requirement
  $V^a_{\rm lr}(x^a)=\vO(\abs{x^a}^{-1-\epsilon})$ for some
  $\epsilon>0$, we may use $ S_{a,{\rm sr}}(\xi,t):=t\xi^2$ rather
  than $ S_{a}(\xi,t)$. Arguing as  above we can then use the results of this
  paper (derived for the function $ S_a$) to deduce analogue results for stationary 
short-range scattering theory (based on the function $ S_{a,{\rm sr}}$), overlapping with \cite {Ya1, Ya4}.
\end{enumerate} 
\end{remarks}

Let
$k_\alpha=p_a^2+\lambda^\alpha$  and recall $I^\alpha=(\lambda^\alpha,\infty)$. By the intertwining property $H W_\alpha^{\pm}\supseteq
W_\alpha^{\pm}k_\alpha $  and the fact that $k_\alpha$ is diagonalized by
the unitary map  $F_\alpha:L^2(\mathbf X_a)\to L^2(I^\alpha ;L^2(C_a))$ given by
\begin{align}\label{eq:Four}
  \begin{split}
  (F_\alpha u)(\lambda,\omega)&=(2\pi)^{-n_a/2}2^{-1/2}
  \lambda_\alpha^{(n_a-2)/4}\int \e^{-\i  \lambda^{1/2}_\alpha \omega\cdot
  x_a}u(x_a)\,\d x_a,\\ &\quad \quad \text{ where }n_a=\dim \mathbf X_a\mand \lambda_\alpha=\lambda-\lambda^\alpha,  
  \end{split}
\end{align} we can write
\begin{align*}
  \hat
  S_{\beta\alpha}:=F_\beta(W_\beta^+)^*W_\alpha^-F_\alpha^{-1}=\int^\oplus_{
  I_{\beta\alpha} }
  S_{\beta\alpha}(\lambda)\,\d \lambda\quad
 \text{with }I_{\beta\alpha}=I^\beta\cap I^\alpha.
\end{align*} The fiber operator $ S_{\beta\alpha}(\lambda)\in
\vL(L^2(C_a),L^2(C_b))$ is a priori defined only for
a.e. $\lambda\in  I_{\beta\alpha}$. It is the
\emph{$\beta\alpha$-entry of the  scattering matrix}.

Similarly  the restrictions of the maps $F_\alpha (W^\pm_\alpha)^*$
have  
strong almost everywhere interpretations, meaning more precisely
\begin{align*}
  F_\alpha (W^\pm_\alpha  )^*f= \int^\oplus_{
  I^{\alpha} }\parb{F_\alpha (W^\pm_\alpha)^*f}(\lambda)\,\d
  \lambda;\quad f\in \vH.
\end{align*} When applied to $f\in \vB (\mathbf X)\subseteq
\vH $ the following assertion holds. 

Let $\vR_a$ denote the
reflection operator on $ L^2(C_a)$, $(\vR_ag)(\omega)=g(-\omega)$, and let 
\begin{align*}
 c_\alpha^\pm(\lambda)=\e^{\pm \i
  \pi (n_a-3)/4}\pi ^{-1/2}\lambda_\alpha^{1/4}, \quad\lambda>\lambda^\alpha.
\end{align*} Let for $\lambda\in \vE_\alpha$ 
\begin{align}\label{eq:genFou}
  \begin{split}
   \Gamma_\alpha^+(\lambda)=c_\alpha^+(\lambda)&Q^+_{\lambda,\alpha}R(\lambda+ \i
                               0)\in \vL(\vB(\mathbf X),L^2(C_a)),\\
\Gamma_\alpha^-(\lambda)=c_\alpha^-(\lambda)\vR_a&Q^-_{\lambda,\alpha}R(\lambda-
\i 0)\in \vL(\vB(\mathbf X),L^2(C_a)).
 \end{split}
\end{align} The  operator  $ \Gamma_\alpha^+(\lambda)^*$ is  the
\emph{$\alpha$-entry of the  future wave matrix}, while  $ \Gamma_\alpha^-(\lambda)^*$ is  the
\emph{$\alpha$-entry of the  past wave matrix}. This language is
motivated by the following result.

\begin{thm}\label{thm:wave_matrices} For any channel
  $(a,\lambda^\alpha, u^\alpha)$ obeying  Condition
  \ref{cond:decayBNstates}  and  any $f\in \vB (\mathbf X)\subseteq
\vH$ 
\begin{align}\label{eq:adjFORM}
  \parb{F_\alpha
  (W^\pm_\alpha)^*f}(\lambda)=\Gamma_\alpha^\pm(\lambda)f\quad\text{for a.e.
  } \lambda\in \vE^\alpha.
\end{align} In particular  for any  $f\in \vB (\mathbf X)$ the restrictions $\parb{F_\alpha (W^\pm_\alpha)^*f}(\cdot)\in
L^2(C_a)$ are
weakly continuous in $\vE^\alpha$.  
\end{thm}
\begin{proof} Due to Lemma \ref{lemma:poiss-oper-geom} it suffices to
  check \eqref{eq:adjFORM}, in fact only the formulas
\begin{align}\label{eq:ForINT}
  \inp{f,W^\pm_\alpha F_\alpha^*g}=\int
  \inp[\big]{\Gamma_\alpha^\pm(\lambda)f,g(\lambda, \cdot)} \,\d
  \lambda;\quad g\in C^\infty_\c(\vE^\alpha\times C_a').
\end{align} Note that the latter space is dense in $ L^2(I^\alpha ,
L^2(C_a))$. 
We can mimic  the  proof of \cite[Lemma 3.8]{IS3} using as input 
  \cite[Lemma 6.4]{II}, considering as in \cite{IS3} only  the plus
  case and (without loss of generality) only   $f\in
  L^2_1(\mathbf X)$. 
With reference  to the   proof of
  \cite[Lemma 3.8]{IS3}   we  first commute $H$ and a 
factor $\chi_m$ appearing there, here taken as
$\chi_m=\chi_-(\abs{x}/2^m)$,  and as a result  then obtain for
any $g$ as in \eqref{eq:ForINT}  the expression
$\chi_m(H-\lambda) J_\alpha v^{+}_{\lambda,\alpha}
[g(\lambda,\cdot)]$. Noting  that $(H-\lambda) J_\alpha v^{+}_{\lambda,\alpha}
[g(\lambda,\cdot)] \in\vB(\mathbf X)$, cf.  \eqref{eq:besFull},  we can take
 the parameter $\epsilon>0$  appearing in the   proof of
  \cite[Lemma 3.8]{IS3} to $0$  and $m\to \infty$ in the written
  order (parallel to \cite{IS3}). By using the resulting formula
  in combination with \eqref{eq:small0} we conclude  \eqref{eq:ForINT} (in the plus case).
\end{proof}

Under the condition of asymptotic completeness and with
 Condition
  \ref{cond:decayBNstates}  fulfilled for {all}  channels (or rather for
all  open channels) we can use Theorem \ref{thm:wave_matrices} and
write, arguing here  formally,
\begin{align*}
   \Gamma_\beta^+(\lambda)=\sum_{\alpha}\, S_{\beta\alpha}(\lambda)\Gamma_\alpha^-(\lambda).
 \end{align*} Applying this formula to $f=(H-\lambda) J_\alpha
 v^{-}_{\lambda,\alpha} [g]$  leads with Lemma \ref{lemma:Sommerfeld}
  and \eqref{eq:geoS}  to the identification of
 $\Sigma_{\beta\alpha}(\lambda)$ and $S_{\beta\alpha}(\lambda)$. 
 However we will do  the identification (stated precisely below) under
 weaker conditions.

\begin{lemma}\label{lemma:weakconV} Let $g\in
  C^\infty_\c(\vE^\alpha\times C_a')$, where 
  $\alpha=(a,\lambda^\alpha, u^\alpha)$ is  any channel obeying 
   Condition \ref{cond:decayBNstates}. Then the  map
   \begin{align*}
     \R\ni
                              \lambda\to f_{\lambda,g}:=(H-\lambda) J_\alpha v^{+}_{\lambda,\alpha}
                              [g(\lambda,\cdot)]
   \end{align*}
  is a
continuous $L^2_s(\mathbf X)$-valued function for some $s>1/2$
(depending on $\alpha$ and $\mu$), and 
  \begin{align}\label{eq:intRE}
   \wvHlim_{\epsilon \to 0_+}
    \int  \parb{R(\lambda+ \i \epsilon)-R(\lambda- \i
    \epsilon)}f^+_{\lambda,g}\,\d \lambda= \int P^+_{\lambda,\alpha} [g(\lambda,\cdot)]\,\d \lambda.
  \end{align} (Here the integrand  on the right-hand side is  an 
  $L^2_{-t}(\mathbf X)$-valued function  for any  $t>1/2$.)
\end{lemma}
\begin{proof} \Step{I} Thanks to \eqref{eq:small} and the known regularity of
  the function $K_a$ the map $\lambda\to f_{\lambda,g}$ is 
  checked to be a 
continuous $L^2_s(\mathbf X)$-valued function for some $s>1/2$. The
assertion  appears already in  the proof
of Lemma \ref{lemma:poiss-oper-geom}.
\Step{II}  It remains to prove
\eqref{eq:intRE}. 
Write $\delta_\epsilon(\lambda)=(2\pi\i)^{-1}\parb{R(\lambda+ \i \epsilon)-R(\lambda- \i
  \epsilon)}\geq 0$ and note the familiar
\begin{align*}
  \int^{\infty}_{-\infty}  \inp{\varphi,\delta_\epsilon(\lambda)\varphi}\,\d \lambda=
   \norm{\varphi}^2;\quad \phi\in \vH.
\end{align*} By the support properties of $g$, \eqref{eq:LAPbnda}, Step I and  \caS, it  follows
that the $\vH$-valued integrals to the left  in
\eqref{eq:intRE} are
uniformly bounded in $\epsilon\in (0,1)$.
\Step{III}  Taking any $f\in L^2_{t}(\mathbf X)$, $t>1/2$,  we compute, using in the second  step  Lemma
\ref{lemma:Sommerfeld} \ref{item:1s},
\begin{align*}
  &\lim_{\epsilon\to 0_+}\int \inp[\big]{f,\parb{R(\lambda+ \i \epsilon)-R(\lambda- \i \epsilon)}f^+_{\lambda,g}}\,\d \lambda\\
&=\int \inp[\big]{f,\parb{R(\lambda+ \i 0)-R(\lambda- \i 0)}f^+_{\lambda,g}}\,\d \lambda\\
&=\int \inp[\big]{f,J_\alpha v^{+}_{\lambda,\alpha}
                              [g(\lambda,\cdot)]-R(\lambda- \i 0)f^+_{\lambda,g}}\,\d \lambda\\
&=\int \inp[\big]{f,P^+_{\lambda,\alpha}[g(\lambda,\cdot)]}\,\d \lambda.
\end{align*} We conclude \eqref{eq:intRE} by this computation, Step II
and the  Riesz theorem.
\end{proof}

\begin{thm}\label{thm:scat_matrices} 
Let any two channels  $\alpha=(a,\lambda^\alpha, u^\alpha)$ and  $\beta=(b,\lambda^\beta, u^\beta)$ both fulfilling  Condition \ref{cond:decayBNstates} be
given (allowing the case  $\alpha=\beta$). Then
\begin{align}\label{eq:link}
  S_{\beta\alpha}(\lambda)=\e^{ \i
  \pi
  (n_a+n_b-2)/4}\lambda_\alpha^{1/4}\lambda_\beta^{-1/4}\Sigma_{\beta\alpha}(\lambda)\vR_a;\quad 
  \lambda\in 
  \vE^\alpha\cap\vE^\beta.
\end{align} In particular the map $\vE^\alpha\cap\vE^\beta\ni
\lambda\to S_{\beta\alpha}(\lambda)\in\vL \parb{L^2(C_a),L^2(C_b)}$ is
weakly continuous.
\end{thm}
\begin{proof} The second assertion
follows from \eqref{eq:link} and  Lemma  \ref{lemma:geomS}.
 To obtain  \eqref{eq:link} let $g_\alpha\in C^\infty_\c(\vE^\alpha\times C_a')$ and $g_\beta\in
  C^\infty_\c(\vE^\beta\times C_b')$ (taking $g_\alpha=g_\beta$ if $\alpha=\beta$). We compute using Lemma
  \ref{lemma:weakconV} and  \eqref{eq:small0}
  \begin{align*}
    &\inp*{\int P^+
_{\lambda,\beta} [g_\beta(\lambda,\cdot)]\,\d
    \lambda, \int P^-_{\lambda',\alpha} [g_\alpha(\lambda',\cdot)]\,\d
    \lambda'}\\
&=\lim_{\epsilon\to 0_+}\int \inp[\Big]{\parb{R(\lambda+ \i \epsilon)-R(\lambda- \i
  \epsilon)}f^+_{\lambda,g_\beta},      \int P^-
_{\lambda',\alpha} [g_\alpha(\lambda',\cdot)]\,\d
    \lambda'}\,\d \lambda\\
&=\lim_{\epsilon\to 0_+}\int \,\d
    \lambda\int \inp[\big]{   f^+_{\lambda,g_\beta},\tfrac{ -2\i\epsilon}{(\lambda'-\lambda)^2+\epsilon^2}P^-
_{\lambda',\alpha} [g_\alpha(\lambda',\cdot)]     }\,\d
  \lambda'\\
&=-2\pi\i \int \,\inp[\big]{ f^+_{\lambda,g_\beta}, P^-
_{\lambda,\alpha } [g_\alpha(\lambda,\cdot)]       }\,\d
  \lambda\\
&=4\pi\int \,\lambda_\alpha^{1/2} \inp[\big]{g_\beta(\lambda,\cdot),Q^+_{\lambda,\beta}\parb{P^-_{\lambda,\alpha}
  [g_\alpha(\lambda, \cdot)]}}\,\d \lambda\\
&=4\pi\int \,\lambda_\alpha^{1/2} \inp[\big]{g_\beta(\lambda,\cdot),\Sigma_{\beta\alpha}(\lambda)g_\alpha(\lambda, \cdot)}\,\d \lambda.
\end{align*} In the third step we used that  for any $t>1/2$ the map $\lambda'\to P^+
_{\lambda',\alpha} [g_\alpha(\lambda',\cdot)]\in L^2_{-t}(\mathbf X)$
is  continuous. The formula \eqref{eq:link} follows from this
computation, Lemma
\ref{lemma:poiss-oper-geom} and Theorem
\ref{thm:wave_matrices}. 
\end{proof}

\section{Wave and scattering matrices for the one-body
  problem}\label{sec:one-body  matrices}
In this section we collect various results from Section
\ref{sec:preliminaries} specialized to  the one-body case and add a few more
one-body results. Apart from this material and Subsection \ref{subsec:body Hamiltonians, limiting absorption
  principle  and notation} the other subsections of Section
\ref{sec:preliminaries} will not be used in the rest part of the paper.

For any $a\neq a_{\max}$ we consider with the one-body potential
$\brI_a=\brI_a(x_a)$ from 
\eqref{eq:brevePotential} (there with $\breve\mu$ replaced by $\mu$, for convenience) 
\begin{align}\label{eq:brevH}
  \brh_a=p^2_a+\brI_a,\quad\brH_a=H^a\otimes I+I\otimes \brh_a\,\mand \,H_a=H^a\otimes I+I\otimes p^2_a.
\end{align}
Using \eqref{eq:wave_op} with $H$ replaced by
$\brH_a$ we obtain 
\begin{align}\label{eq:wave_op2}
  \breve W_\alpha^{\pm}=\slim_{t\to \pm\infty}\e^{\i t\brH_a}J_\alpha\e^{-\i
  (S_a^\pm (p_a,t)+\lambda^\alpha t)}=J_\alpha \breve w_a^{\pm},
\end{align}  where 
\begin{align}\label{eq:wave_op3}
  \breve w_a^{\pm}=\slim_{t\to \pm\infty}\e^{\i t\brh_a}\e^{-\i
  S_a^\pm (p_a,t)}.
\end{align} Note that $\breve w_a^{\pm}$ is a one-body wave
operator and that the theory of  Section \ref{sec:preliminaries}  applies (in a
simplified form). 
 It is well-known (given the property \eqref{eq:brevePotential}) that the scattering operator $\breve
s_a=(\breve w_a^{+})^*\breve w_a^{-}$ is a unitary operator on
$L^2(\mathbf X_a)$. Whence in that case 
\begin{align}\label{eq:1-bodycont}
 (0,\infty)\ni
\lambda\to \brs_{a}(\lambda)\in\vL \parb{L^2(C_a)}\text{ is
  strongly continuous.}
\end{align} 

More importantly for this paper  we can introduce operators $
{\gamma}_a^\pm(\lambda)$ as in \eqref{eq:genFou}, using obviously
modified  notation
for the one-body case,
\begin{align}\label{eq:genFou2}
  \begin{split}
   \gamma_a^+(\lambda)=c_a^+(\lambda)&\brq^+_{\lambda,a}\brr_a(\lambda+ \i
                               0)\in \vL(\vB(\mathbf X_a),L^2(C_a)),\\
\gamma_a^-(\lambda)=c_a^-(\lambda)\vR_a&\brq^-_{\lambda,a}\brr_a(\lambda-
\i 0)\in \vL(\vB(\mathbf X_a),L^2(C_a)).
 \end{split}
\end{align} The operators $ \gamma_a^+(\lambda)^*$ and $
\gamma_a^-(\lambda)^*$, related to \eqref{eq:wave_op3} as in Theorem
\ref{thm:wave_matrices},  are   the
 \emph{future
 and past wave matrices} at $\lambda \in (0,\infty)$,
 respectively. They are strongly weak*-continuous as functions of $\lambda$. In
 the weaker topology of $\vL\parb{L^2(C_a),L^2_{-s}(\mathbf X_a)}$, 
 $s>1/2$, $ \gamma_a^\pm(\cdot)^*$ are strongly  continuous. The above 
 assertions follow from the one-body version of 
 Lemma \ref{lemma:poiss-oper-geom}.

We recall the familar formulas, see for example \cite[Lemma 3.5]{IS2}
(and its proof),
\begin{align}\label{eq:stone}
 {\gamma_a^\pm(\lambda)}^*
  \gamma_a^\pm(\lambda)=\delta\parb{\brh_a-\lambda}=(2\pi\i)^{-1}\parb{\brr_a(\lambda+ \i
                               0)-\brr_a(\lambda- \i
                               0)};\quad \lambda>0.
\end{align} In combination with the asserted weak*-continuity it
follows (as for  \cite[ Theorem 1.4]{IS2}) that  $ \gamma_a^\pm(\cdot)$ are strongly
continuous in $\vL\parb{\vB(\mathbf X_a),L^2(C_a)}$. We shall only use the latter property for the free case
 (i.e. with $\brI_a=0$, see Appendix \ref{sec:AppendixB}),  in which
 case we will use the notation
 $\gamma_{a,0}(\lambda)$ for $\gamma_a^\pm(\lambda)$ (the two operators
 corresponding to different  signs coincide in that case).  

\begin{subequations}
Introducing 
\begin{align}\label{eq:Ffield}
  G_a:=\vG_a\cdot
  p_a\mand \vG_a=\vG_a(x_a):=\chi_+(4\abs{x_a})\abs{x_a}^{-1/2}\parb{I-\abs{x_a}^{-2}\ket{ x_a}\bra{ x_a}},
\end{align} 
we have
\begin{align}\label{eq:perpOne1}
  G_a{\gamma_a^\pm(\lambda)}^* \in\vL\parb{L^2(C_a),L^2(\mathbf X_a)^{n_a}}\text{ with
  a locally uniform bound in }\lambda>0,
\end{align} and in fact
\begin{align}\label{eq:perpOne2}
   G_a{\gamma_a^\pm(\cdot)}^*\in
  \vL\parb{L^2(C_a),L^2(\mathbf X_a)^{n_a}}\text{ are
 strongly continuous}.
\end{align}  

As for \eqref{eq:perpOne1} it is thanks to \eqref{eq:stone} and `the
$TT^*$ argument' a consequence of  the  resolvent bound 
\cite[Theorem 3.9]{Ya1}.  Alternatively \eqref{eq:perpOne1}  can 
easily be derived by  `the
$T^*T$ argument' given in Step I of the proof of Lemma
\ref{lemma:strongCont}. As for
\eqref{eq:perpOne2} the result follows readily from  an application of \eqref{eq:perpOne1},
Lemma \ref{lemma:poiss-oper-geom}  and
\cite[Theorem 1.10]{IS1}. For comparison   we remark that
\eqref{eq:perpOne2}  for the free Laplacian appears in  \cite{Ya1} and 
the proof there applies 
a scaling argument.
\end{subequations}

\section{Stationary channel modifiers
   and a main result}\label{sec:Stationary modifier}
We assume from this point of the paper Condition
\ref{cond:smooth2wea3n12} with  $\mu\in (\sqrt 3-1,1)$. Asymptotic completeness is then a  known fact, 
\cite{De}. In particular the wave operators $ W_\alpha^{\pm} $ from
\eqref{eq:wave_op} exist. Henceforth Condition
\ref{cond:decayBNstates}  is not used.
 We will in the present  section explain our overall scheme for
 proving weak 
continuity of any entry of the scattering matrix. This  result is
stated as Theorem \ref{thm:mainstat-modif-n}. Most of the needed  notation for its
proof will be fixed in this section.

We shall consider a `sufficiently small' open interval $\Lambda$ containing any
arbitrarily fixed
$\lambda_0$  obeying
\begin{align*}
  \lambda_0\in \Lambda\subseteq \vE:=(\min \vT(H),\infty)\setminus \parb{\sigma_{\pp}(H)\cup \vT(H)}.
\end{align*} Only the scattering matrix in an open interval $ I_0\ni
\lambda_0$, $\bar I_0\subseteq \Lambda$, needs 
consideration (by partitioning of unity).

For  $f_1,f_2\in C^\infty_\c(\R)$ taking values in $[0,1]$ we write
$f_2\prec f_1$ if $f_1=1$ in a neighbourhood of $\supp f_2$.
We consider
real $f_1,f_2\in C^\infty_\c(\Lambda)$ with $f_2\prec f_1$  and such
that $f_2=1$ in $\bar I_0$. (The smallness of $\Lambda$ will depend on Mourre estimates at $\lambda_0$, see
Subsection \ref{subsec:Limiting absorption principles}.)

For a given channel $\alpha= (a,\lambda^\alpha, u^\alpha)$ let 
\begin{align*}
  \brJ^\pm_\alpha= N^a_\pm (I\otimes \breve w_a^{\pm})J_\alpha \mand
  \tilde{J}_\alpha^\pm =f_1(H)M_aN^a_\pm M_a(I\otimes \breve w_a^{\pm})J_\alpha,
\end{align*} where $\breve w_a^{\pm}$ are given by
\eqref{eq:wave_op3}, and  $N^a_\pm$ and $M_a$  are symmetric operators
specified below. The latter
operator $M_a$  is
given as in \cite[(5.1)]{Ya1} (with changed  notation $M^{(a)}\to M_a$, and
recalled in Section \ref{sec:Yafaev's construction}), in
particular $ M_{a}\in\vL(\vH^1,\vH)$, and
$N^a_\pm\in\vL(\vH)$ will be constructed such that
(recalling  $k_\alpha= p_a^2+\lambda^\alpha $)
\begin{subequations}
\begin{align}\label{eq:wave_opa}
  W_\alpha^{\pm} f_2(k_\alpha)=\slim_{t\to \pm\infty}\e^{\i tH}J_\alpha \breve w_a^{\pm}\e^{-\i
   t k_\alpha}f_2(k_\alpha)=\slim_{t\to \pm\infty}\e^{\i tH}\brJ^\pm_\alpha\e^{-\i
   tk_\alpha}f_2(k_\alpha),
\end{align} 
By the form of $\brJ^\pm_\alpha$ this  amounts to the property
\begin{align}\label{eq:lim00}
  \slim_{t\to \pm \infty}\parb{I-N^a_\pm }J_\alpha  \breve w_a^{\pm}\e^{-\i
   tk_\alpha}f_2(k_\alpha)=0.
\end{align}  
\end{subequations}

Recalling the operators in \eqref{eq:brevH} we introduce in addition
$\brh_\alpha=\brh_a+\lambda^\alpha$ and then (abbreviating $I\otimes
f_1(\brh_\alpha)=f_1(\brh_\alpha)$)
\begin{align*}
  \Phi_\alpha^\pm=f_1(H)M_a N^a_\pm M_a
    f_1(\brh_\alpha)f_1(\brH_a). 
\end{align*}  We  call  the operators  $\Phi_\alpha^\pm$
\emph{stationary channel  modifiers}. 
 Provided the two limits
\begin{align}\label{eq:wav2s}
  \Omega^\pm_\alpha=\slim_{t\to \pm\infty}\,\e^{\i tH}
  \Phi_\alpha^\pm  \e^{-\i t\brH_a}\text{ exist},
\end{align}  the limits 
\begin{subequations}
\begin{align}\label{eq:wave_opc}
  \begin{split}
   \slim_{t\to \pm\infty}\e^{\i tH}\tilde{J}_\alpha^\pm\e^{-\i
   tk_\alpha}f_2(k_\alpha)=\slim_{t\to \pm\infty}{\e^{\i tH} \Phi_\alpha^\pm
  \e^{-\i t\brH_a} J_\alpha \breve w_a^{\pm}}f_2(k_\alpha) \text{ exist},
  \end{split}
\end{align} and denoting them  by $\widetilde
W_\alpha^{\pm}
$, respectively, obviously 
\begin{align}\label{eq:formlW}
\widetilde W_\alpha^{\pm} = \Omega^\pm_\alpha J_\alpha \breve w_a^{\pm}f_2(k_\alpha).
\end{align}
  \end{subequations}

  For   $a\neq a_{\min}$  the operators $N^a_\pm$ are of
the form (one for  the upper sign plus  and one for the lower sign minus)
\begin{align}\label{eq:propgaObs}
  \begin{split}
  N^a_\pm&=A_1 A_2^aA_3^aA_2^aA_1;\\
A_1&=A_{1\pm}=\chi_+(\pm
  B/\epsilon_0),\\
A_2^a &=\chi_-\parb{
  r^{\rho_2-1}r_\delta^a},\\
A_3^a&=A^a_{3\pm}=\chi_-\parb{
  \pm r^{\rho_1/2}B_\delta^ar^{\rho_1/2}},
  \end{split}
\end{align} where  $\epsilon_0>0$ is sufficiently  small
 (determined by Mourre estimates at $\lambda_0$, see
 Subsection \ref{subsec:Limiting absorption principles}), 
\begin{align}\label{eq:parameters}
  1-\mu<\rho_2<\rho_1<1-\delta\text{ with } \delta\in [2/(2+\mu),\mu),
\end{align}
 and
$ r, B, r_\delta^a$ and $B_\delta^a$ are operators constructed by quantities from
\cite{De} ($r$ and $r_\delta^a$ are multiplication operators while
$B$ and $B_\delta^a$ are corresponding Graf vector field type constructions). 
For $a= a_{\min}$ we  take $
N^a_\pm=A_1^2=A_{1\pm}^2$.  

More precisely 
 $r  $ is  a function on $\mathbf X$, which apart from a trivial rescaling (to assure
 Mourre estimates for the  Graf
vector field $\nabla r^2/2$) is given by  the function $r$ from \cite{De}.  It partly plays the role  as a
`stationary time variable' compared to the usage of the real time
parameter in  \cite{De}. (It should not be mixed up with the function
$\abs{x}$.) The operator $B:
=2\Re(p\cdot\nabla r)$. Let $r^a$ be  the same
function now constructed on $\mathbf X^a$ rather than on $\mathbf X$. Let then $r_\delta^a=
r^\delta r^a(x^a/r^\delta)$ and 
$B_\delta^a=2\Re\parb{p^a\cdot(\nabla r^a)\parb{x^a/r^\delta}}$.
 We give further  elaboration of the above quantities in Section \ref{sec:Derezinski's
  construction}.

We  verify \eqref{eq:wav2s} and \eqref{eq:lim00} in  Subsection
\ref{subsec:Integrability} and Appendix \ref{sec:AppendixA}, respectively.  Moreover by
\eqref{eq:lim00}  and   stationary
phase analysis we may  rewrite \eqref{eq:wave_opc}  and
\eqref{eq:formlW}  as 
\begin{subequations}
\begin{align}\label{eq:wave_opc2}
  \widetilde W_\alpha^\pm=\slim_{t\to \pm\infty}{\e^{\i tH}  
  \e^{-\i t\brH_a} J_\alpha \breve
  w_a^{\pm}}f_2(k_\alpha)\parb{m_\alpha^\pm}^2= W_\alpha^\pm f_2(k_\alpha)\parb{m_\alpha^\pm}^2,
\end{align}
 where ${m_\alpha^\pm}$ (after conjugation by the Fourier
 transform $F_\alpha $, see \eqref{eq:Four})  act as multiplication operators on
 $L^2\parb{I^\alpha; L^2(C_a)}$ (with multiplicative
 fiber operators). More precisely 
 \begin{align}\label{eq:conjF}
   F_\alpha m_\alpha^\pm F^{-1}_\alpha =\int^\oplus_{I^\alpha}
   \pm 2\lambda_\alpha^{1/2 }m_a(\hat \xi_a) \, \d \lambda=\pm
   2\lambda_\alpha^{1/2 }m_a(\pm\hat \xi_a);\quad \hat \xi_a=\xi_a/\abs{\xi_a}\in C_a
 \end{align}  and  $m_a$ is the function used in \eqref{eq:M_a}
 specifying the operator $M_a$. The verification of
 \eqref{eq:wave_opc2} and \eqref{eq:conjF} is  given in Appendix \ref{sec:AppendixA}. 
\end{subequations}

We conclude  the following representation of any
 scattering matrix entry  $S_{\beta\alpha}(\lambda)$ (assuming 
 $\lambda^\alpha,\lambda^\beta<\lambda_0$),
 \begin{subequations}
 \begin{align}\label{eq:locScata}
   16\lambda_\beta\lambda_\alpha f^2_2(\lambda)m_b(\hat \xi_b)^2
   S_{\beta\alpha}(\lambda) m_a(-\hat \xi_a)^2=\widetilde S_{\beta\alpha}(\lambda),
 \end{align} where (recalling $I_{\beta\alpha}=I^\beta\cap
 I^\alpha=(\max\set{\lambda^\alpha,\lambda^\beta}, \infty)$)
\begin{align}\label{eq:locScatb}
  F_\beta\widetilde 
  S_{\beta\alpha}F_\alpha^{-1}=\int^\oplus_{
  I_{\beta\alpha} }
  \widetilde S_{\beta\alpha}(\lambda)\,\d \lambda\quad
  \text{with}\quad   \widetilde S_{\beta\alpha}=\parbb{\widetilde W_\beta^+}^*\widetilde W^-_\alpha.
\end{align}   
 \end{subequations}

We shall verify \eqref{eq:locScatb} for a weakly continuous
operator-valued function  $\widetilde
S_{\beta\alpha}(\cdot)$.  By \eqref{eq:locScata} the operator-valued function  
$S_{\beta\alpha}(\cdot)$ will then also be weakly
continuous  on $I_0\ni \lambda_0$ thanks to a freedom in 
choosing the functions $m_b$ and $m_a$. The reader may at this point
consult the beginning of Section \ref{sec:Exact channel wave-matrices}
where these functions are taken  explicitly  (of course conforming 
 with Section \ref{sec:Yafaev's construction}.)

Since $\lambda_0\in \vE$ is arbitrary a partition of unity allow us to
 conclude  a  main
result of the paper, stated as follows. Let $\vE_{\beta\alpha} :=I_{\beta\alpha}\setminus \parb{\sigma_{\pp}(H)\cup \vT(H)}$.
\begin{thm}\label{thm:mainstat-modif-n}
  For any  incoming  channel $\alpha$ and  any outgoing  channel $\beta$  the map
  \begin{align*}
\vE_{\beta\alpha} \ni
\lambda\to S_{\beta\alpha}(\lambda)\in\vL \parb{L^2(C_a),L^2(C_b)}
\end{align*} is a well-defined 
weakly continuous map. Within the class of such maps it is uniquely determined by the identity
\begin{align*}
  F_\beta 
  S_{\beta\alpha}F_\alpha^{-1}=F_\beta \parb{W_{\beta}^+}^*W_{\alpha}^-F_\alpha^{-1}=\int^\oplus_{\vE_{\beta\alpha} 
  }
   S_{\beta\alpha}(\lambda)\,\d \lambda.
\end{align*}  
\end{thm}

Motivated by  \eqref{eq:wav2s} and \eqref{eq:formlW} it
is convenient to introduce the notation
\begin{align*}
  T^\pm_\alpha=\i \parb{H\Phi_\alpha^\pm  -\Phi_\alpha^\pm\brH_a}=\i f_1(H)\parb{H  M_a N^a_\pm M_a - M_a N^a_\pm M_a\brH_a}f_1(\brh_\alpha)f_1(\brH_a).
\end{align*} In Section \ref{sec:Computation of T} we shall compute it
in detail, using various preliminary calculus considerations from  Section \ref{sec:Calculus considerations},  and then show the
existence of the operators $\Omega^\pm_\alpha$ of \eqref{eq:wav2s} 
using a variety of  `weak propagation estimates'. These  estimates will
then  be
applied to derive  a representation formula of $\widetilde
S_{\beta\alpha}(\lambda)$ in Section \ref{sec:Formula for widetilde},
and from this  formula we  finally deduce the weak continuity  of
Theorem \ref{thm:mainstat-modif-n}.

We devote Section \ref{sec:Exact channel wave-matrices} to further
results that come out as byproducts of our proof of Theorem
\ref{thm:mainstat-modif-n}. In particular we show that
$S_{\beta\alpha}(\cdot)$ is strongly continuous for almost
  all energies in $\vE_{\beta\alpha} $, and we  introduce and show the
strong  continuity of the channel wave matrices for  all energies in $\vE_{\beta\alpha} $.

\subsection{Extended lattice structure of subspaces}\label{sebsec:Lattice structure}

For the above constructions  we  assume that the family of subspaces $\{\mathbf X^c\}$, $c\in\vA$,  is stable under addition
(the standard assumption) and includes all sets of the form $\mathbf X_c$
(a non-standard assumption). This  can be done without loss of
generality by adding to  the collection  $\{\mathbf X^c\}\cup \{\mathbf X_c\}$  all sums of
subspaces from the collection, and we can then consider $H$ as an
$N$-body operator with this new lattice structure, say 
in this subsection indexed by
$\vA^{\rm{new}}$, simply by taking $V_d=0$   on  all 
added subspaces $\mathbf X^d$. The notions of eigenvalues and thresholds of
$H$ are the same as with  the original (i.e. `old')   lattice
structure, say in this subsection  indexed by
$\vA^{\rm{old}}$.

However  with $\vA^{\rm{new}}$  we can also consider the operators $\brH_c$
and $H_c$ from \eqref{eq:brevH} as `full' $N$-body operators.  
 More precisely for $\brH_c$, $c\in
\vA^{\rm{old}}$, we define
$V_c=\brI_c$  on $\mathbf X_c$ and for all 
added subspaces $\mathbf X^d\neq \mathbf X_c$ we take $V_d=0$.  If $\mathbf X^d\subseteq \mathbf X^c$
is an old subspace
the potential $V_d$  is the same as for the  old subspace, and for the
remaining 
 case we take $V_d=0$. 
 For $H_c$, $c\in \vA^{\rm{old}}$, we define
$V_c=0$  on $\mathbf X_c$.
 Otherwise we use
the same definition as for   $\brH_c$.

With these conventions it is easy to check that  the set of
thresholds for $H_c$, $c\in\vA^{\rm{old}}\setminus \{a_{\max}\}$, is
exactly given by 
$\sigma_{\pp}(H^c)\cup \vT(H^c)\subseteq \vT(H)$. We shall use this
feature in Subsection \ref{subsec:Limiting absorption principles}.

\section{Yafaev's constructions}\label{sec:Yafaev's construction}

Let for $\epsilon\in(0,1)$ (below  taken sufficiently small)  and for $a\in\vA'=\vA\setminus \{a_{\max}\}$
\begin{align*}
  \mathbf X_a(\epsilon)=\set{\abs{x_a}>(1-\epsilon)\abs{x}}.
\end{align*} Note that $\mathbf X_a \setminus\set{0}\subseteq \mathbf
X_a(\epsilon)$ and in fact that $\mathbf X_a \setminus\set{0}=\cap_{\epsilon\in(0,1)}\,\mathbf
X_a(\epsilon)$.

In \cite{Ya3} a real function $m_a\in C^\infty(\mathbf X)$,   $a\in\vA'$,  is
constructed, fulfilling the following
properties for any sufficiently small  $\epsilon>0$: 

\begin{enumerate}[1)]
\item\label{item:1a} $m_a(x)$ is homogeneous of degree $1$ for $\abs{x}\geq 1$ and
  $m_a(x)=0$ for $\abs{x}\leq 1/2$.
\item\label{item:2a} If $x\in \mathbf X_b(\epsilon)$ and $\abs{x}\geq 1$, then
  $m_a(x)=m_a(x_b)$ (i.e. $m_a(x)$ does not depend on $x^b$).
\item\label{item:3a} If ${b\not\leq a}$ and $x\in \mathbf X_b(\epsilon)$, then   $m_a(x)=0$.
\end{enumerate}

Let for  any such function $w_a=\mathop{\mathrm{grad}} m_a$ and 
\begin{align}\label{eq:M_a}
  M_a=2\Re(w_a\cdot p)=-\i\sum_{j\leq n}\parb{(w_a)_j\partial_{x_j}+\partial_{x_j}(w_a)_j};\quad  n:=\dim \mathbf X.
\end{align}

\begin{subequations}
Let for  $a\in\vA'$
\begin{align}\label{eq:primes}
  \begin{split}
    \mathbf X'_a&={\mathbf X_a}\setminus\cup_{{ c\gneq a} }\mathbf
                 X_c={\mathbf X_a}\setminus\cup_{{c\not\leq a}
                 }\mathbf X_c,\\
\mathbf X'_a(\delta)&= \mathbf X_a\setminus\cup_{ c\gneq a}
\overline{\mathbf X_c(\delta)};\quad \delta>0.
\end{split}
\end{align} Here and below the overline means topological closure.  Note that $\mathbf X'_a(\delta)\subseteq \mathbf X'_a$.

We need the following additional information  from \cite[Lemma 3.5]{Ya3}
 on the construction of the functions $m_a$, cf. a discussion in
 Section \ref{sec:Stationary modifier}:
We can for any  given point 
$z_a\in \mathbf X'_a$ 
 choose   $m_a$
with the additional property $m_a(\hat z_a)\neq 0$;  $\hat
z_a:=z_a/\abs{z_a}$. 
%More generally
%we can  assume   for any given pair $(z_a,z_b)\in \mathbf X'_a\times  \mathbf X'_b$ that
%$m_a(\hat z_a)\neq 0$ and $ m_b(\hat z_b)\neq 0$.
 More generally,   
for any given $\delta>0$ and for any sufficiently small $\epsilon>0$
  we can  take $m_a$    such that 
  $m^a( x) =\abs{x_a}$ for $x\in \mathbf X'_a(\delta)$ with
  $\abs{x}\geq 1$. We shall implement the latter property explicitly
  in the beginning of Section \ref{sec:Exact channel wave-matrices},
  however for our proof of weak continuity of the fiber operator
  $\widetilde S_{\beta\alpha}(\cdot)$ of \eqref{eq:locScatb} it is not
  relevant. It is needed only for extracting  the continuity of 
  $S_{\beta\alpha}(\cdot)$ in terms of  \eqref{eq:locScata}.

To control the Hessian  of the functions
$m_a$, $\mathop{\mathrm{Hess}}m_a=\nabla^2m_a$,  (or more precisely $
p\cdot\nabla^2m_a \,p$ arising in  commutator calculations with
\eqref{eq:M_a}) Yafaev used a family of  similar
functions  $m$ with an additional convexity property. Before recalling
 these functions
 we  introduce more conical subsets of  $\mathbf X$.

Let for  $a\in\vA'$ and  $\epsilon\in(0,1)$
\begin{align}\label{eq:Gamma}
  \mathbf \Gamma_a(\epsilon)=\parb{\mathbf X \setminus\set{0}}\setminus\cup_{{b\not\leq a}
}\mathbf X_b(\epsilon).
\end{align} Note,  as a motivation, that
$\set{\abs{x}\geq 1}\cap\supp m_a\subseteq \mathbf
\Gamma_a(\epsilon)$ for any small $\epsilon>0$. 

Next we note that  
\begin{align}\label{eq:Gamma2}
  \mathbf \Gamma_a(\epsilon)\subseteq\cup_{d\leq a}  \,\mathbf X'_d.
\end{align} In fact  we can for any $x\in \mathbf
\Gamma_a(\epsilon)$ introduce $d(x)\leq a$ by  $\mathbf X_{d(x)}=\cap_{d\leq
a,\,x\in \mathbf X_d} \,\mathbf X_d$ and then check that $x\in \mathbf
X'_{d(x)}$: If not, $x\in \mathbf X_c \cap \mathbf X_{d(x)}$ for some $c$ with $d(x)\lneq c\leq a$ and
  therefore  $\mathbf X_c=\mathbf X_c\cap \mathbf X_{d(x)}=\mathbf X_{d(x)}$,  contradicting that
  $c\neq d(x)$. 

Introduce for  $\delta>0$ the open cone 
\begin{align}\label{eq:Y}
  \mathbf Y_d(\delta)=\mathbf X_d(\delta)\setminus\cup_{c\gneq d}
\,\overline{\mathbf X_c(3\delta^{1/n})}.
\end{align}
 Thanks to   \eqref{eq:Gamma2} we can for any  $\delta_0>0$  write 
  \begin{align}\label{eq:deltaNBHa}
    \mathbf \Gamma_a(\epsilon)\subseteq\cup_{d\leq a} \cup_{\delta\in (0,\delta_0]} \,\mathbf Y_d(\delta).
  \end{align} By compactness this leads, for any fixed
  $\epsilon\in(0,1)$,  to the existence of $J\in\N$, $\delta_1,\dots, \delta_J\in
(0,\delta_0]$ and  $d_1,\dots, d_J\leq a$ such that 
\begin{align}\label{eq:deltaNBHb}
    \mathbf \Gamma_a(\epsilon)\subseteq\cup_{j\leq J} \,\,\mathbf Y_{d_j}(\delta_j).
  \end{align}

 \end{subequations}

 By \cite[p. 538]{Ya3} there exists $\delta_0'>0$ such that for all  
 $\delta\in(0,\delta_0']$ there exists a real function $m=m_{a_{\max}}\in C^\infty(\mathbf X)$ fulfilling  the following properties:

\begin{enumerate}[i)]
\item\label{item:1b} $m(x)$ is homogeneous of degree $1$ for $\abs{x}\geq 1$ and
  $m(x)=0$ for $\abs{x}\leq 1/2$.
\item\label{item:2b} If $x\in \mathbf X_b(\delta)$ and $\abs{x}\geq 1$, then
  $m(x)=m(x_b)$ (i.e. $m(x)$ does not depend on $x^b$).
\item\label{item:3b} $m(x)$ is convex in $\abs{x}\geq 1$.
\item\label{item:4b} For all $d\in \vA'$ there exists $\mu_d\geq 1$ such that 
\begin{align*}
  m( x) =\mu_d\abs{x_d} \text{ for all }x\in \mathbf Y_d(\delta)\text{ with } \abs{x}\geq 1.
\end{align*}
\end{enumerate}

Similar to \eqref{eq:M_a} we introduce for any such function  $w=\mathop{\mathrm{grad}} m$ and 
\begin{align}\label{eq:M}
  M=M_{a_{\max}}=2\Re(w\cdot p)=-\i\sum_{j\leq n}\parb{w_j\partial_{x_j}+\partial_{x_j}w_j}.
\end{align}

To relate and further specify the introduced functions $m_a$, $a\in\vA'$ and $a=a_{\max}$,
let us first fix the order of construction as follows: Fix $\epsilon>0$ conforming  with the
construction of the family of functions $m_a$ fulfilling the
properties 
\ref{item:1a}--\ref{item:3a} (as well as the non-vanishing condition
discussed above in a separate paragraph) and fix $\delta_0'>0$ conforming  with the
construction of the family of functions $m$ fulfilling  the
properties 
\ref{item:1b}--\ref{item:4b}. Obviously we can assume $\delta_0'\leq
\epsilon$. Take then $\delta_0=\delta_0'$ in
\eqref{eq:deltaNBHa}  and \eqref{eq:deltaNBHb}.   Next for
each function $m_a$, $a\in\vA'$, we use  \eqref{eq:deltaNBHb} and
construct functions $m_j$  fulfilling  \ref{item:1b}--\ref{item:4b}
with $\delta=\delta_j$; $j=1\dots J$.

 For each  $a\in\vA'$ we now choose 
   a quadratic partition $\xi_1,\dots, \xi_J\in
C^\infty(\S^{n-1})$ 
subordinate to the covering \eqref{eq:deltaNBHb}, and let
$\xi^+_j(x)=\xi_j(\hat x)\chi_+(4\abs{x})$; $\hat x:=x/\abs{x}$. Then
we can write, using  the support properties \ref{item:1a} and \ref{item:3a}
\begin{align*}
  m_a(x)=\Sigma_{j\leq J}\,\,m_{a,j}(x);\quad m_{a,j}(x)=\xi^+_j( x)^2 m_a(x),
\end{align*} and using  \ref{item:2a} and the properties 
$\mathbf Y_{d_j}(\delta_j)\subseteq \mathbf X_{d_j}(\delta_j)\subseteq \mathbf X_{d_j}(\epsilon)$
\begin{align*}
  m_{a,j}(x)\chi_+(\abs{x})=\xi^+_j( x)^2 m_a(x_{d_j})\chi_+(\abs{x}).
\end{align*} Similarly (using in addition the notation $G_{a}$ of \eqref{eq:Ffield})
\begin{subequations}
\begin{align}\label{eq:Hes0}
  \begin{split}
 &p\cdot\parb{\chi_+(\abs{x})\nabla^2m_a(x)} p=\Sigma_{j\leq J}\,\,
                                 p\cdot\parb{\chi_+(\abs{x}) \xi^+_j( x)^2 \nabla^2m_a(x_{d_j})}p,\\&
=\Sigma_{j\leq J}\, \,G^*_{d_j}\parb{\chi_+(\abs{x}) \xi^+_j( x)^2\vG_j}G_{d_j};\quad
                                                                  \vG_j=\vG_j(x_{d_j})\text{ 
                                                                  bounded}.   
  \end{split}
\end{align}
In turn due to \ref{item:3b} and \ref{item:4b}, the latter applied with
$d=d_j$ (and $\delta=\delta_j$), 
\begin{align}\label{eq:Hes_est}
   G^*_{d_j}\chi_+(\abs{x})\xi^+_j( x)^2
  G_{d_j}\leq  p\cdot \parb{\chi_+(\abs{x})\nabla^2m(x)} p.
\end{align} 
\end{subequations}

 Finally  we introduce functions $\xi^+_a$ and $\tilde\xi^+_a$ as follows.
First choose any $\xi_a\in
C^\infty(\S^{n-1})$ such that $\xi_a=1$ on
$\S^{n-1}\cap\mathbf \Gamma_a(\epsilon)$ and $\xi_a=0$ on
$\S^{n-1}\setminus \mathbf \Gamma_a(\epsilon/2)$. Choose then any $\tilde\xi_a\in
C^\infty(\S^{n-1})$ using this recipe   with $\epsilon$
replaced by  $\epsilon/2$. Finally  let  $\xi^+_a(x)=\xi_a(\hat
x)\chi_+(4\abs{x})$ and  $\tilde\xi^+_a(x)=\tilde\xi_a(\hat
x)\chi_+(8\abs{x})$,  and note  that
$\tilde\xi^+_a\xi^+_a=\xi^+_a$. Then by \ref{item:1a} and \ref{item:3a}
\begin{align}\label{eq:partM}
  M_a= M_a\xi^+_a= \xi^+_aM_a=M_a\tilde\xi^+_a=\tilde \xi^+_aM_a,
\end{align} which in applications will provide `free factors' of
$\xi^+_a$ and $\tilde\xi^+_a$
where  convenient.

\section{Derezi\'nski type constructions}\label{sec:Derezinski's
  construction}

There exists  a positive  function $r \in C^\infty(\mathbf X)$  fulfilling  the
following properties \cite{De, AIIS} (recall $n=\dim \mathbf X$):
 \begin{enumerate}[i)]
\item\label{item:1c} For all $\gamma\in \N^n_0$ and $k\in\N_0$ there
  exists $C>0$ such that 
\begin{align*}
  \abs[\big]{\partial ^\gamma (x\cdot \nabla)^k \parb{r(x)-\inp{x}}} \leq C\inp{x}^{-1}.
\end{align*} 
\item\label{item:2c} There exists $c>0$ such that for all $a\in\vA$:
  \begin{align*}
    \abs{x^a}\leq c\Rightarrow \nabla^a r(x)=0
  \end{align*}
(i.e.  if $\abs{x^a}\leq c$, then $r(x)$ does not depend on $x^a$).
\item\label{item:3c} $r(x)$ is convex.
\item\label{item:3d} For  any  given non-threshold  energy  of $H$ there
  is a Mourre
  estimate for $H$ with $A:= \Re\parb{\nabla r^2\cdot p}$, or more
  precisely with such $A$ for  a suitably rescaled version of the
  function $r$
  obeying \ref{item:1c}--\ref{item:3c}. (Here and later we slightly abusively  use the same
  notation $r$ for the rescaled 
  function $Rr(x/R)$ with  $R\geq 1$  taken sufficiently large). See
  Subsections \ref{subsec:Limiting absorption principles} and
  \ref{subsubsec:17111117} for  details.
\item\label{item:3e} The family of functions $\set{r^a|\, a\in \vA}$,
    with $r^a$ being a function on $\mathbf X^a$ obeying
  \ref{item:1c}--\ref{item:3c} (rather than on $\mathbf X$, and with
  $r^{a_{\min}}$ being a positive constant), fulfills the
following self-similar structure.
\begin{align}\label{eq:self-similar}
  \forall \epsilon>0 \,\exists R'\geq 1 \,\forall  x\in
  \mathbf \Gamma_a(\epsilon)\cap\set{\abs{x}\geq R'}: \quad r^2=\abs{x_a}^2+(r^a)^2;
\end{align} here  we use the notation of \eqref{eq:Gamma}, $\mathbf \Gamma_a(\epsilon)=\parb{\mathbf X \setminus\set{0}}\setminus\cup_{{b\not\leq a}
}\mathbf X_b(\epsilon)$.

\end{enumerate}

It follows from \ref{item:1c} that
\begin{subequations}
\begin{align}\label{eq:r0}
\abs[\big]{\partial ^\gamma \parb{r(x)-\inp{x}}} &\leq C\inp{x}^{-1},\\\label{eq:r1}
  \abs[\big]{\partial ^\gamma\parb{x\cdot \nabla r(x)-r(x)}} & \leq C
  \inp{x}^{-1},\\
\abs[\big]{\partial ^\gamma\parb{  x\cdot (\nabla^2r)(x)x} }& \leq C\inp{x}^{-1},\label{eq:r2}
\end{align} cf.  \cite{De}.
\end{subequations}

For completeness of presentation we  note that  the bounds of
\ref{item:1c}, proven only for $k\leq 2$ in  \cite{De}, are  equivalent to  
\begin{align*}
  \abs[\big]{\partial ^\gamma (x\cdot \nabla)^k \parb{r(x)^2-\inp{x}^2}} \leq C,
\end{align*}  and the latter estimates  can be checked using the construction of \cite{De}.
  (They were verified  explicitly in  \cite[Appendix
A]{Sk1} for  the seminal
 construction of Graf \cite{Gr}.) As in \cite{De} we shall need
 \eqref{eq:r0}-\eqref{eq:r2} rather than \ref{item:1c}.

We also remark that a  version of \eqref{eq:self-similar} appears as 
     \cite[(5.10)]{Sk2} (proven in detail in the paper using  the seminal
 construction of Graf). Again one can  check the stated assertion 
 using the construction of \cite{De}. Although it is conveniently  used  in
 Appendix \ref{sec:AppendixO} it is not essential  there.

\begin{subequations}
Let $\omega=\mathop{\mathrm{grad}}r$ and 
\begin{align}\label{eq:B}
  B=2\Re(\omega\cdot p)=-\i\sum_{j\leq n}\parb{\omega_j\partial_{x_j}+\partial_{x_j}\omega_j}.
\end{align} For the 
self-adjointness of $B$  we
refer to Subsection \ref{subsec:Operatat B}.  This operator appears in
the definition of the  factor $A_1$ of the operators $
N^a_\pm$  of Section \ref{sec:Stationary
  modifier}.

  With $\delta$ as in \eqref{eq:parameters} 
 we define for $a\in \vA\setminus\set {a_{\max},a_{\min}}$,  and with
 $r^a$ given by \ref{item:3e} and  $\omega^a=\mathop{\mathrm{grad}}  r^a$,
\begin{align}\label{eq:b^a1}
  r_\delta^a&=r_\delta^a(x)=
r^\delta
              r^a(x^a/r^\delta),\\\quad \omega^a_\delta&=\omega^a_\delta(x)=\omega^a\parb{x^a/r^\delta}
,\label{eq:b^a11}\\B_\delta^a&=2\Re\parb{p^a\cdot\omega^a_\delta}.\label{eq:b^a2}
\end{align}
\end{subequations}

The quantities of \eqref{eq:B}--\eqref{eq:b^a2} are
building blocks of 
the factors $A_2^a$ and $A_3^a$ of the operators $
N^a_\pm$ and fix them  for any choice of 
parameters  $\rho_2,\rho_1$ and $\delta$ in  \eqref{eq:parameters}. For the 
self-adjointness of $B_\delta^a$ and $B_{\delta,\rho_1}^a=r^{\rho_1/2}B_\delta^a
                         r^{\rho_1/2}$  we
refer to Subsection \ref{subsec:Commutation
  with A_3^a}. Recall  that $N^{a_{\min}}_\pm:=A_1^2$.

\subsection{Mourre estimates and limiting absorption principles}\label{subsec:Limiting absorption principles}

Parallel to \eqref{eq:LAPbnda}  we 
note that  (with the limits taken uniformly in $\lambda$ in the norm-topology)
\begin{align}\label{eq:LAPbndbr}
  \begin{split}
   &\forall  s>1/2\,\forall \lambda\in \R:\quad f_1(\brh_\alpha)f_1(\brH_a)  \brR_a(\lambda\pm \i
  0)f_1(\brH_a)f_1(\brh_\alpha)\\&=\lim _{\epsilon\to 0_+} f_1(\brh_\alpha)f_1(\brH_a)  \brR_a(\lambda\pm \i
  \epsilon)f_1(\brH_a)f_1(\brh_\alpha)\,\in
  \vL\parb{L^2_s,L^2_{-s}}.  
  \end{split}
\end{align} Considered as   operator-valued functions  of 
$\lambda$, these quantities    are  continuous. These facts follow
from similar (well-known) assertions for the one-body operator 
$\brh_a$ localized to positive energies.

We claim that for some $\epsilon_1>0$ (being independent of the
scaling parameter used below)
\begin{subequations}
\begin{align}\label{eq:mourre000}
  f_1(\brh_\alpha)f_1(\brH_a)
  \i[\brH_a,r^{1/2}Br^{1/2}]f_1(\brH_a)f_1(\brh_\alpha)\geq \epsilon_1f_1(\brh_\alpha)f_1(\brH_a)^2f_1(\brh_\alpha).
\end{align} Here the commutator $\i[\brH_a,r^{1/2}Br^{1/2}]$ is given
by  its formal expression, see
 Subsection \ref{subsubsec:17111117} or \cite[(2.15)]{AIIS}, which is a bounded form on $\vD(H)$.

We can verify \eqref{eq:mourre000} as follows. Consider 
$\tilde f_1,\tilde f_2\in C^\infty_\c(\Lambda)$  with $f_1\prec \tilde
f_2\prec \tilde
f_1$. Upon shrinking the support of $\tilde f_2$ (for fixed $\tilde f_1$)
\begin{align*}
  \norm{\tilde f_2(\brh_\alpha)\parb{\tilde{f_1}(\brH_a)-\tilde{f_1}(H_a)}}\text{
  is (arbitrarily) small},
\end{align*} and  it suffices  to show that 
\begin{align}
  \tilde{f_1}(H_a)
  \i[\brH_a,r^{1/2}Br^{1/2}]\tilde{f_1}(H_a)\geq
  \epsilon_2 \tilde{f_1}(H_a)^2.
\end{align} At this point we stress that the commutator in
\eqref{eq:mourre000} is given by its formal expression and note that it   can be estimated uniformly in $R\geq 1$. (Recall
that we are going to use the rescaled version  $r(x)\to
Rr(x/R)$, $R\geq 1$ sufficiently big.) The `size' of the support of
$\tilde f_2$ may depend on the scaling parameter  $R$ through an
$R$-dependence of $\tilde f_1$ (this  dependence consequently also
holds for the 'smaller' functions $f_2$ and $f_1$, however being an
irrelevant property). The stated smallness follows from  the
fact that $\lambda_0-\lambda^\alpha$ is not an eigenvalue of
$\brh_a$ (there are no positive eigenvalues) and a compactness argument.

In turn (possibly by taking the scaling parameter large enough, as seen
by the concrete expression of Subsection \ref{subsubsec:17111117} for the commutator) it suffices to show 
\begin{align}\label{eq:mourre}
  \tilde{f_1}(H_a)
  \i[H_a,r^{1/2}Br^{1/2}]\tilde{f_1}(H_a)\geq
  \epsilon_3 \tilde{f_1}(H_a)^2.
\end{align}
 
Now the Mourre estimate \eqref{eq:mourre} essentially  follows from
\cite{AIIS} (see the
paragraph below for a comment),
again   (if needed)  in terms of  the rescaled version  $r(x)\to
Rr(x/R)$, $R\geq 1$ sufficiently big. 
(As mentioned before we  suppress such modification and will use the
same notation $r$ for the rescaled version of the function.) Note that $\mathbf X_a$ belongs to the lattice of
subspaces $\{\mathbf X^c\}$, cf. Subsection \ref{sebsec:Lattice structure},
and that we use the fact that
$\lambda_0$ is not a threshold nor an eigenvalue of  the $N$-body
Hamiltonian $H_a$. 

Strictly speaking the conjugate operator appearing
in \eqref{eq:mourre} is defined in terms of  the vector field
$\tilde\omega=\tfrac12\mathop{\mathrm{grad}} r^2$ associated to  the
construction of \cite{De} while the Mourre
estimate of  \cite{AIIS}  (originating from \cite{Sk1})  is stated in
terms of  the  vector field from the
seminal paper \cite{Gr}. However the difference of the two vector
fields is minor, caused only  from the fact that different
regularization procedures are used. In particular one can indeed, based
on the construction of \cite{De}, easily check that the same proof as the one 
referred to in \cite{AIIS} works for 
the Mourre estimate \eqref{eq:mourre} (with the stated version of the
Graf vector). A  similar
remark  applies  for \eqref{eq:mourreFull} stated below.
\end{subequations}

Next, since  \eqref{eq:mourre000} is established, it is convenient
(and harmless) to change
notation by replacing $f_1$ in  \eqref{eq:mourre000} by some $\tilde
f_2\in C^\infty_\c(\Lambda)$ with ${f_1}\prec \tilde{f_2}$,
 amounting  to
\begin{align}\label{eq:mourre0}
  \tilde f_2(\brh_\alpha)\tilde f_2(\brH_a)
  \i[\brH_a,r^{1/2}Br^{1/2}]\tilde f_2(\brH_a)\tilde
  f_2(\brh_\alpha)\geq \epsilon_1\tilde
  f_2(\brh_\alpha)\tilde f_2(\brH_a)^2\tilde f_2(\brh_\alpha).
\end{align} 
%This strengthening of \eqref{eq:mourre000} is clearly  doable for   some 
%$\tilde f_1\in C^\infty_\c(\Lambda)$  with 
%$f_1\prec \tilde f_1$ (using different functions than above). 

 Note also  the Mourre estimate   used in   \cite{AIIS}, corresponding to
replacing $H_a$ in \eqref{eq:mourre} by $ H$,
\begin{align}\label{eq:mourreFull}
  \tilde{f_2}(H)
  \i[H,r^{1/2}Br^{1/2}]\tilde{f_2}(H)\geq
  \epsilon_4 \tilde{f_2}(H)^2;
\end{align}  again possibly for the rescaled version of $r$ only, and 
with $\tilde f_2$  given as in \eqref{eq:mourre0}. See 
Subsection \ref{subsubsec:17111117} for a convenient form of
\eqref{eq:mourreFull} needed  in Section \ref{sec:Computation of T}.

\begin{subequations}
We recall 
 from \cite{AIIS} that \eqref{eq:mourreFull} implies
the following version of \eqref{eq:LAPbndbr} in a neighbourhood $\vN$
of $\supp
f_1$,
cf. \eqref{eq:LAPbnda} and \eqref{eq:BB^*a}:
\begin{align}\label{eq:LAPbnd}
  \forall \, s>1/2:\quad R(\lambda\pm \i
  0)=\lim _{\epsilon\to 0_+}  R(\lambda\pm \i
  \epsilon)\in \vL\parb{L^2_s,L^2_{-s}}\text{ for }\lambda\in \vN.
\end{align} In  any such  norm-topology the limits are taken
uniformly on $\vN$, and for  the (stronger) Besov space  topology the
limiting operators obey
  \begin{align}\label{eq:BB^*}
    R(\lambda\pm \i
  0)\in \vL\parb{\vB,\vB^*}\text{  with a uniform bound in }\lambda\in \vN.
  \end{align}
\end{subequations} 

With \eqref{eq:mourre0}
  and \eqref{eq:mourreFull}  in place we can henceforth  consider the scaling
  parameter $R\geq 1$ as fixed, see  Subsection
  \ref{subsec:Operatat B} for a more general discussion.
In particular  the above parameters $\epsilon_1$ and $\epsilon_4$  can be chosen as a
function of the distance to the biggest threshold of  $H$ below
$\lambda$, see Lemma \ref{lemma:Mourre1_hard}.
 We may at this point fix the parameter $\epsilon_0>0$ in the factor $A_1$  of the operators $
N^a_\pm$  in Section \ref{sec:Stationary
  modifier} by the single requirement $ 4\epsilon^2_0<\min\set{\epsilon_1,
  \epsilon_4}$, however its size will not play any role. Henceforth we
just consider  $\epsilon_0$  as a small positive parameter. (The 
freedom of possibly choosing $\epsilon_0$ smaller will come in
conveniently, although not crucially, in the proof of Lemma
\ref{lemma:strongCont}.)  In combination
with the previous constructions \eqref{eq:B}--\eqref{eq:b^a2} we have
by now fixed  $
N^a_\pm$ (for any given 
parameters  $\rho_2,\rho_1$ and $\delta$).

\section{Calculus considerations}\label{sec:Calculus considerations}

 To facilitate  our treatment of  
$T^\pm_\alpha$ in Section \ref{sec:Computation of T} we recall from
\cite[Section 2]{AIIS}  a calculus in which the Mourre estimate
 \eqref{eq:mourreFull} can be implemented. This calculus fits well for
 computing $T^\pm_\alpha$ by  the usual product rule for commutation.

In Subsection~\ref{subsubsec:Smooth sign function}
we introduce notation frequently used in the later arguments.
Subsection~\ref{subsec:Functional calculus} 
 concerns  the Helffer--Sj\"ostrand formula, and two applications on
 computing  commutators are given.
In Section~\ref{subsec:Operatat B} we provide  the self-adjoint realization of 
the  operator $B$ of \eqref{eq:B}. 
We investigate the first commutator $\mathrm i[H,B]$ in Subsection~\ref{subsubsec:17111117}, 
and the second commutator $\mathrm i[\mathrm i[H,B],B]$ in
Subsection~\ref{subsubsec:17111119}. This will apply for commutation
with the factor $A_1$.   Commutation
with the factor $A_3^a$ is similar. This will be treated in Subsection \ref{subsec:Commutation
  with A_3^a}.   Commutation
with the factors $M_a$ and $A_2^a$ are easier, to be treated in Subsections \ref{subsec:Commutation
  with M} and \ref{subsec:A2}, respectively.

\subsection{Notation}\label{subsubsec:Smooth sign function}

%This is a short subsection devoted to  some notation only.

Let $T$ be a linear operator on $\mathcal H=L^2(\mathbf X)$ such that
$T,T^*:L^2_\infty\to L^2_\infty$, and let $t\in\mathbb R$.  Then we
say that   $T$ is an {\emph{operator of order $t$}}, if 
 for each $s\in\mathbb R$  the restriction  $T_{|L^2_\infty}$ extends to
 an operator $T_s\in\vL(L^2_{s}, L^2_{s-t})$. Alternatively stated,
\begin{align}\label{eq:defOrder}
\|r^{s-t}Tr^{-s}\psi\|\le C_s\|\psi\| \text{ for all }\psi\in L^2_\infty.
\end{align} We  can (and will) use this with  $r$ replaced by
$r_R=r_R(x):=Rr(x/R)$ for  a fixed sufficiently large   $R\geq 1$,
cf. 
Subsection \ref{subsec:Limiting absorption principles}. Slightly abusively  we dont change the
notation and use  $r$ rather than $r_R$ throughout the section,  although in the context of Subsection
\ref{subsubsec:17111117} the quantity $r_R$ is  needed in our
presentation (and meant).  Note  (for consistency) that  $T_s$ extends the restriction $T_{|
\vD(T)\cap L^2_{s}}$. 
 If   $T$ is of {order $t$}, we write 
\begin{align}
T=\vO(r^t).
\label{eq:1712022}
\end{align} 
Note also that, if $T=\vO(r^t)$ and $S=\vO(r^s)$, then $T^*=\vO(r^t)$
and $TS=\vO(r^{t+s})$.  
If
$T=\vO(r^t)$  for all $t\in\R$, then we write $T=\vO(r^{-\infty})$.
 If $T=\vO(r^s)$ for some $s<t$, we write
$T=\vO(r^{t_-})$. (This will in Section \ref{sec:Computation of T} be
a desirable  property  with  $t=-1$.)

\subsection{Functional calculus}\label{subsec:Functional calculus}

Here we present the Helffer--Sj\"ostrand formula to represent 
functions of self-adjoint operators, 
and its application to commutators.

For  $t\in\mathbb R$ we denote by $\mathcal F^t$  the set of real $f\in
C^\infty(\mathbb R)$ obeying
\begin{align*}
|f^{(k)}(x)|\le C_k\langle x\rangle^{t-k}\text{ for any }k\in\mathbb N_0\text{ and }x\in\mathbb R
.
\end{align*}
It is known that for any $f\in\mathcal F^t$, $t\in\mathbb R$,
there always exists a  \textit{good  almost analytic extension}
$\tilde f\in C^\infty (\C)$, meaning more precisely that such function
$\tilde f$ obeys
\begin{equation*} 
%\label{87a} 
\tilde f_{|\mathbb R}=f, 
\quad 
|\tilde f(z)|\le C\langle z\rangle^t,\quad 
\bigl|(\bar{\partial}\tilde{f})(z)\bigr| 
\leq C_k|\Im z|^{k}\left\langle z\right\rangle^{t-k-1} 
\ \ \text{for any }k\in\mathbb N_0.
\end{equation*} 
Here one can choose $\tilde f\in C^\infty_{\mathrm c}(\mathbb C)$
if $f\in C^\infty_{\mathrm c}(\mathbb R)$.

\begin{lemma}\label{lem:A1} 
Let $T$ be a self-adjoint operator on $\vH$, and let $ f\in \mathcal F^t$ with $t\in\mathbb R$.
Take an almost analytic extension $\tilde f\in C^{\infty }(\C)$ of $f$
as above. 
Then for any $k\in\mathbb N_0$ with $k>t$ 
the operator $f^{(k)}(T)\in\mathcal L(\mathcal H)$ is expressed as 
\begin{align}\label{82a0}
  \begin{split}
 f^{(k)}(T) 
&=
(-1)^kk!\int _{\C}(T -z)^{-k-1}\,\mathrm d\mu_f(z)\, \text{ with}\\
&\mathrm d\mu_f(z)=\pi^{-1}(\bar\partial\tilde f)(z)\,\mathrm du\mathrm dv;\quad 
z=u+\i v.   
  \end{split}
\end{align}
\end{lemma}

The expression \eqref{82a0} for $k=0$ is the well-known Helffer--Sj\"ostrand
formula used extensively in the literature to compute and  bound commutators. 
In general there are  several variations of  the  definition of a commutator,
and in this paper we do not fix a particular one. 
It will be  clear from the context 
in what sense we will be considering a commutator.
Typically, for symmetric operators $T$ and $S$,  we first define $\mathrm
i[T,S]$ as the   quadratic form 
\begin{align*}
\langle \mathrm i[T,S]\rangle_\psi
=2\langle \mathop{\mathrm{Im}}(ST)\rangle_\psi
=\mathrm i\langle T\psi,S\psi\rangle-\mathrm i\langle S\psi,T\psi\rangle
\ \ \text{for }\psi\in \mathcal D(T)\cap\mathcal D(S),
\end{align*}
and then extend it to a larger space. 

Let us provide an example of a commutator formula derived this way using \eqref{82a0}.

\begin{corollary}\label{cor:A2} 
Let $T$ be a self-adjoint operator on $\vH$, 
 $S$ be a  symmetric  relatively $T$-bounded operator,
and assume that there exists a bounded extension
$$(|T|+1)^{-\epsilon/2}\bigl(\mathrm i[T,S]\bigr) (|T|+1)^{-\epsilon/2} \in\mathcal L(\mathcal H)
\ \ \text{for some }\epsilon\in [0,2].$$
Let a    $f\in \mathcal F^t$  with $t<1-\epsilon$ be given, and let $\mathrm d\mu_f$
be given as in \eqref{82a0}.
Then, as a quadratic form on $\mathcal D(f(T))\cap \mathcal D(S)$,
\begin{equation*}
\mathrm i[f(T),S] 
=-
\int _{\C}
(T-z)^{-1}\bigl(\mathrm i[T,S]\bigr) (T-z)^{-1}\,\mathrm d\mu_f(z),
\end{equation*}
and it extends to  a bounded self-adjoint operator on $\mathcal H$.
\end{corollary}

Another example is the commutator $[f(H),r^s]$ treated below, which is
an example of commutators of functions of entries from the triple of operators $(H,r,B)$.
Other such  examples  will be discussed  in 
Lemmas~\ref{lem:171113b} and \ref{lem:fHB}.
In Section \ref{sec:Computation of T} we will  repeatedly use
Lemmas~\ref{lem:171113}, \ref{lem:171113b} and 
\ref{lem:fHB}--\ref{lem:17111515}.

\begin{lemma}\label{lem:171113}
\begin{enumerate}[1)]
\item \label{item:180128b}
For any  $f\in \mathcal F^0$ the operator $f(H)$ is of order $0$.
\item\label{item:180128}
Let any  $f\in \mathcal F^t$ with $t<1/2$  and $s\in\mathbb
R$ be given.
Then $\mathrm i[f(H),r^s]$ 
has an expression, 
as a  quadratic  form on $\mathcal D(f(H))\cap L^2_{\max\{0,s\}}$,
\begin{align}\label{eq:fComR}
\mathrm i[f(H),r^s] 
=
-2s
\int _{\C}
(H-z)^{-1}\mathop{\mathrm{Re}}(r^{s-1}\omega\cdot p)\,(H-z)^{-1}\,\mathrm d\mu_f(z).
\end{align} 
In particular  $\mathrm i[f(H),r^s]$ is of order $s-1$.
\end{enumerate}
\end{lemma}

\subsection{The operator $B$, the Mourre estimate  and commutation with $A_1$}\label{subsec:Operatat B}

We show that the operator $B$ is self-adjoint and examine  commutators
 with functions of $B$, including the  prime  example $\i[f(H),A_1]$,
 $f\in
 C^\infty_\c(\R)$  real  and  $A_1$ given by \eqref{eq:propgaObs}.

\subsubsection{Self-adjoint realization}\label{subsubsec: self-ad B}

We recall  the self-adjoint realization of the  operator $B$ from
\cite{AIIS} and accompanying properties related to 
 the spaces  \eqref{eq:17111317}.

\begin{lemma}\label{lemma:Flow} 
The operator $B$ defined as \eqref{eq:B}
is essentially self-adjoint on $C^\infty_{\mathrm c}(\mathbf X)$,
and the self-adjoint extension, denoted by $B$ again,
satisfies that for some $C>0$
\begin{subequations}
\begin{align}
\mathcal D(B)\supseteq \mathcal H^1,\quad 
\|B\psi\|_{\mathcal H}\le C\|\psi\|_{\mathcal H^1}\ \ \text{for  }\psi\in\mathcal H^1.
\label{eq:171005}
\end{align}
In addition, $\mathrm e^{\mathrm itB}$ for each $t\in\mathbb R$
naturally restricts/extends
as bounded operators $\mathrm e^{\mathrm itB}\colon \mathcal H^{\pm
  k}\to\mathcal H^{\pm k}$, $k=1,2$, and 
they satisfy
\begin{align}
\sup_{t\in [-1,1]}\|\mathrm e^{\mathrm itB}\|_{\mathcal L(\mathcal H^{\pm k})}<\infty,
\label{eq:17100722}
\end{align}  
\end{subequations}
respectively. Moreover, the restriction $\mathrm e^{\mathrm itB}\in
\mathcal L(\mathcal H^k)$, $k=1,2$,  is 
strongly continuous in $t\in\mathbb R$.
\end{lemma}

\begin{lemma}\label{lem:171113b}
\begin{enumerate}[1)]
\item\label{item:181115}
For any  $F\in \mathcal F^0$  the operator   $F(B)$ is of order $0$.
\item\label{item:BrForm}
 Let  $F\in \mathcal F^t$ with $t<1$ and 
  $s\in\mathbb R$ be given.
Then $\mathrm i[F(B),r^s]$ 
is represented,  
 as a  quadratic form on $\mathcal D(F(B))\cap L^2_{\max\{0,s\}}$,
by 
\begin{align*}
\mathrm i[F(B),r^s] 
= -2s
\int _{\C}
(B-z)^{-1}\bigl(\omega^2r^{s-1}\bigr) (B-z)^{-1}\,\mathrm d\mu_F(z).
\end{align*}  
 In particular  $\mathrm i[F(B),r^s]$ is of order $s-1$.
\end{enumerate}
\end{lemma}

\subsubsection{First commutator with $B$ and the Mourre estimate}\label{subsubsec:17111117}
Here we are going to compute the commutator $\mathrm i[H,B]$,
and bound it from below.
We define 
$\mathrm i[H,B]$ first as a (bounded) quadratic form on $\mathcal H^2$:
\begin{align}
\langle \mathrm i[H,B]\rangle_\psi
=2\langle \mathop{\mathrm{Im}}(BH)\rangle_\psi
=\mathrm i\langle H\psi,B\psi\rangle-\mathrm i\langle B\psi,H\psi\rangle
\ \ \text{for }\psi\in \mathcal H^2.
\label{eq:17100614}
\end{align}
We recall that $\omega:=\mathop{\mathrm{grad}} r$ and let 
$$
\tilde\omega=\tfrac12\mathop{\mathrm{grad}} r^2,\quad 
\tilde h=\tfrac12\mathop{\mathrm{Hess}}r^2,\quad
h=\mathop{\mathrm{Hess}}r.
$$
Then \textit{formal} computations  suggest that 
\begin{align*}
 B=&\mathrm i[H,r],\\
 A:=&\tfrac{\mathrm i}2[H,r^2]=r^{1/2}Br^{1/2},\\
%\label{eq:formulasyb}
\mathbf DA:=\i [H,A]=&4p\cdot \tilde h p-\tfrac12\parb{\Delta^2r^2}-2\tilde \omega \cdot (\nabla V),\\
\mathbf DB:=\mathrm i[H,B]=&r^{-1/2}\bigr(\mathrm i[H,A] -B^2\bigr)r^{-1/2} +r^{-2}\omega\cdot h\omega.
%\label{eq:formulasybb}
\end{align*} 
Thus we {could expect that $\mathrm i[H,B]$  extends  continuously to larger spaces,
and this is justified in the following lemma, which in turn partly
justifies the above computations.

\begin{lemma}\label{lem:17100620}
Denoting  the extension of the quadratic form ${\bf D}B=\mathrm
i[H,B]$ of \eqref{eq:17100614}  by the same notation, 
it is expressed as 
\begin{subequations}
\begin{align}
{\bf D}B
&=\slim_{t\to 0} t^{-1}\parb{H\e^{\i tB}-\e^{\i tB}H}\ \ 
\text{in }
{\mathcal L(\mathcal H^2,\mathcal H^{-1})
\cap \mathcal L(\mathcal H^1,\mathcal H^{-2})}
\label{eq:limit00b}
\end{align}
and, more explicitly,  as 
\begin{align}\label{eq:formulasy}
\begin{split}
{\bf D}B
=r^{-1/2}\parb{L -B^2}r^{-1/2}
, 
\end{split}
\end{align} 
where $L\in \mathcal L(\mathcal H^2,\mathcal H^{-1})
\cap \mathcal L(\mathcal H^1,\mathcal H^{-2})$ 
is given by
\begin{align}\label{eq:mourre comm}
\begin{split}
L&=4p\cdot \tilde h p
-\tfrac12(\Delta^2 r^2)
+r^{-1} \omega\cdot h\omega
\\&\phantom{{}={}}{}
+2\sum_{a\in\mathcal A}
\Bigl(-\tilde\omega^a\cdot \bigl(\nabla^a V^a_{\rm lr}\bigr)
+\bigl(V^a_{\rm sr}\tilde\omega^a\bigr)\cdot \nabla^a
-\nabla^a\cdot
\bigl(V^a_{\rm sr}\tilde\omega^a\bigr)
+V^a_{\rm sr}\mathop{\mathrm{div}} \tilde\omega^a
\Bigr).
\end{split}
\end{align}  
\end{subequations}
Here $\tilde\omega^a$ and $\nabla^a$  for any $a\in\mathcal A$ denote  the projection 
onto the internal components of $\tilde\omega$ and $\nabla$,
respectively.
\end{lemma}

Note that formally 
\begin{align*}
L=\mathbf DA+ r^{-1} \omega\cdot h\omega=\mathbf DA+\vO(r^{-1}).
\end{align*}
 Whence we may use the following  version of the  {Mourre estimate}, cf. \cite{AIIS}. For
$\lambda\in \R\setminus \vT(H)$,  $\lambda>\lambda_\c:=\min\sigma_{\c}(H)=\min\sigma_{\ess}(H)$,  the notation $d(\lambda)$ is used for
the distance to the biggest threshold of $H$ below $\lambda$.
\begin{lemma}\label{lemma:Mourre1_hard} 
For any $\lambda\in
(\lambda_\c, \infty)\setminus \vT(H)$   and $\epsilon>0$ 
there exist a
neighbourhood $\vU$ of $\lambda$ and a compact operator $K$ on
$\vH$, such that for all  real $f\in C^\infty_{\c}(\vU)$
\begin{equation*}%\label{eq:mour7d}
f(H)Lf(H)\geq f(H)\bigl(4d(\lambda)-\epsilon
-K\bigr)f(H).
\end{equation*} Here, strictly speaking,  $r$ is meant as the rescaled
version $r_R=r_R(x):=Rr(x/R)$ with $R=R(\lambda,\epsilon)\ge 1$ taken
sufficiently large, and also   $\vU$ and  $K$ may depend on $R$.
\end{lemma}

We will in Section \ref{sec:Computation of T} implement Lemma~\ref{lemma:Mourre1_hard}
in combination with Lemma~\ref{lem:17100620} in the following form,
see also \eqref{eq:basicFBB}.

\begin{corollary}\label{cor:171007}  For any $\lambda\in
(\lambda_\c, \infty)\setminus \vT(H)$ 
there exist $\sigma>0$ and a
neighbourhood $\vU$ of $\lambda$:  For any real 
$f\in C^\infty_{\mathrm c}(\mathcal U)$  there exists $C>0$ 
such that (as quadratic forms on $\mathcal H$)
 \begin{align}
\label{eq:limit}
 \begin{split}
 f(H)(\mathbf DB)f(H)
\ge 
f(H)r^{-1/2}\bigl(\sigma^2-B^2\bigr)r^{-1/2}f(H)
-Cr^{-2}. 
 \end{split}
 \end{align} Here the  strict meaning is the bound with  $r$ replaced with  the rescaled
version $r_R$ with $R\ge 1$ taken
sufficiently large (as in Lemma \ref{lemma:Mourre1_hard}, and  $\vU$ and  $C$ may depend on $R$). 
\end{corollary}

Finally we compute and bound commutators of functions of $H$ and $B$.

\begin{lemma}\label{lem:fHB} 
Let    $f\in \mathcal F^t$ and   $F\in \mathcal F^{t'}$ with
$t<-1/2$ and $t'<1$, respectively. 
Then the  commutators $\mathrm i[f(H),B]$ and  $\mathrm i[f(H),F(B)]$
extend  from $\mathcal D(B)$ and  $\mathcal D(F(B))$ to
bounded quadratic forms 
on $\mathcal H$,  
 represented as
\begin{align*}
\mathrm i[f(H),B]
&=
-\int _{\C}(H-z)^{-1}\bigl(\mathbf DB\bigr) (H-z)^{-1}\,\mathrm d\mu_f(z)
\quad{\mand}\\
\mathrm i[f(H),F(B)]
&=
-\int _{\mathbb C}(B-z)^{-1}\bigl(\mathrm i[f(H),B]\bigr)(B-z)^{-1}\,\mathrm d\mu_F(z),
\end{align*}
respectively. 
Moreover, with  the notation   \eqref{eq:1712022}
\begin{align}\label{eq:firstHB}
\mathrm i[f(H),B]=\vO(r^{-1})
\quad \mand \quad
\mathrm i[f(H),F(B)]=\vO(r^{-1})
.
\end{align}
\end{lemma}

\subsubsection{Second commutator $\ad^2_{\i B}(\cdot)$}\label{subsubsec:17111119}
Here we provide a realization of the 
second commutator $\mathrm i[\mathbf DB,B]$, 
and bound it in some operator space. 
Although 
one can explicitly compute  this second commutator,
this will not be needed. 
 
Note that by Lemmas~\ref{lemma:Flow} and \ref{lem:17100620} we may consider 
$${(\mathbf DB)\e^{\i tB}-\e^{\i tB}(\mathbf DB)\in 
\mathcal L(\mathcal H^2,\mathcal H^{-2}).}$$

\begin{lemma}\label{lemma:sing} 
There exists the strong limit 
\begin{align*} 
\ad^2_{\i B}(H)=\mathrm i[\mathbf DB,B]&:=\slim_{t\to 0}t^{-1}
\bigl((\mathbf DB)\mathrm e^{\mathrm i tB}-\e^{\mathrm i tB}(\mathbf DB)\bigr)
\ \ \text{in }\mathcal L(\mathcal H^2,\mathcal H^{-2}).
\end{align*} 
Moreover 
\begin{align}
(H-\mathrm i)^{-1}\ad^2_{\i B}(H)(H+\mathrm i)^{-1}
=\vO(r^{-1-\mu}).
\label{eq:1710081345}
\end{align}
\end{lemma}

Finally, as a continuation of Lemmas \ref{lem:fHB} and
\ref{lemma:sing},   we consider  a  second commutator of $B$ and  a function of $H$.

\begin{lemma}\label{lem:17111515}
   For any  $f\in \mathcal F^t$ with $t<-1$
the second commutator
\begin{align*}
  \ad^2_{\i B}\parb{f(H)}=\mathrm i\bigl[\mathrm i[f(H),B],B\bigr]
\end{align*}
extends  from $\mathcal D(B)$ to  a bounded  quadratic  form  on  $\mathcal H$, 
 represented as 
\begin{align*}
\ad^2_{\i B}\parb{f(H)} 
&
=-\int_{\C} (H-z)^{-1}\ad^2_{\i B}(H)(H-z)^{-1}\d \mu_f(z)
\\&\phantom{{}={}}{}
+2\int_{\C} (H-z)^{-1}(\bD B)(H-z)^{-1} (\bD B)(H-z)^{-1}\d \mu_f(z).
\end{align*} 
In particular
\begin{align}\label{eq:order}
  \ad^2_{\i B}\parb{f(H)}  =\vO(r^{-1-\mu}).
\end{align}
\end{lemma}

\subsubsection{Commutation with $A_1$}\label{subsubsec:A_1}

By combining Lemmas \ref{lem:fHB}--\ref{lem:17111515} we obtain the
following result.

\begin{lemma}\label{lem:Bpost} Let    $f\in \mathcal F^t$ and   $F\in \mathcal F^{t'}$ with
$t<-1$ and $t'<1$, respectively. Suppose also that $ F'\geq 0$  with
$\sqrt{F'}\in \mathcal F^0$. Then 
\begin{align*}
  \mathrm i[f(H),F(B)]
= \sqrt{F'}(B)\parb{\i[f(H),B]}\sqrt{F'}(B) +\vO(r^{(-1)_-}).
\end{align*}  In particular for any real functions  $f,g$ and ${F}$ with $ F'\geq 0$,
$f,g,\sqrt{F'}\in C^\infty_\c(\R)$ and given such that   $g(\lambda)f(\lambda)=\lambda f(\lambda)$,
\begin{align}\label{eq:basicFB}
  \begin{split}
  &f(H)\mathrm i[g(H),F(B)]f(H)\\
&\quad= \sqrt{F'}(B)f(H)(\mathbf DB) f(H)\sqrt{F'}(B) +\vO(r^{(-1)_-}).
  \end{split}
\end{align} 
\end{lemma}

We shall in Section 
\ref{sec:Computation of T}  use \eqref{eq:basicFB} in combination with
Corollary \ref{cor:171007}. Note that with \eqref{eq:limit} and with
\eqref{eq:basicFB} for ${F'}\in C^\infty_\c(\R)$ supported close to
$0$ we can infer the bound
\begin{align}\label{eq:basicFBB}
  \begin{split}
  \exists \epsilon>0:\quad &f(H)\mathrm i[g(H),F(B)]f(H)\\
&\quad\geq  \epsilon f(H)\sqrt{F'}(B)r^{-1}\sqrt{F'}(B)f(H) +\vO(r^{(-1)_-}).
  \end{split}
\end{align} 

\subsection{Commutation with $M_a$}\label{subsec:Commutation
  with M}
 The operators $B$, $M=M_{a_{\max}}$  and  $M_a$, $a\in\vA'$, are all of first order given in
 terms of  smooth bounded globally Lipschitz vector
 fields. Consequently the conclusions of Lemma \ref{lemma:Flow} also
 hold  
 for $M_a$,  $a\in\vA$. However we shall actually not need the
 self-adjointness of the latter  (and use  their propagators for
 example); symmetry will suffice. However we will need to pass
 functions of $H$ through $M_a$. Recalling 
 $w_a:=\mathop{\mathrm{grad}} m_a$ we let  $w^b_a=(w_a)^b$ for any  $a,b\in\vA$.

\begin{lemma}\label{lem:fHB2} 
Let    $f\in \mathcal F^t$  with
$t<-1/2$ and $a\in\vA$. 
Then the  product $M_af(H)=\vO(r^{-0})$, and the commutator $\mathrm
i[f(H),M_a]$  extends  from $\vH^1$ to
a bounded quadratic form  
on $\mathcal H$,  
 represented    as 
\begin{align*}
\mathrm i[f(H),M_a]
&=
-\int _{\C}(H-z)^{-1}\i[H,M_a] (H-z)^{-1}\,\mathrm d\mu_f(z);\\
  \i[H,M_a] &= 4p\cdot \parb{\mathop{\mathrm{Hess}} m_a}  p
-(\Delta^2 m_a)
\\&\phantom{{}={}}{}
+2\sum_{b\in\mathcal A}
\Bigl(-w_a^b\cdot \bigl(\nabla^b V^b_{\rm lr}\bigr)
+\bigl(V^b_{\rm sr}w_a^b\bigr)\cdot \nabla^b
-\nabla^b\cdot
\bigl(V^b_{\rm sr}w_a^b\bigr)
+V^b_{\rm sr}\mathop{\mathrm{div}} w_a^b
\Bigr).
\end{align*}

\begin{subequations}
In particular
 (with  the notation   \eqref{eq:1712022})
\begin{align}\label{eq:firstHB2}
\mathrm i[f(H),M_a
]=\vO(r^{-1}),
\end{align} and
\begin{align}\label{eq:firstHB22}\begin{split}
&\mathrm i[f(H),M_a
]-4f'(H)p\cdot \parb{\mathop{\mathrm{Hess}} m_a}  p\\=
\mathrm i[f(H),M_a
]&+4\int _{\C}(H-z)^{-1}p\cdot \parb{\mathop{\mathrm{Hess}} m_a}  p
  (H-z)^{-1}\,\mathrm d\mu_f(z)\\&+\vO(r^{-1-\mu})=\vO(r^{-1-\mu}). 
\end{split}
 \end{align}
  \end{subequations}
\end{lemma}

In applications the second term to the left in \eqref{eq:firstHB22}
for  $a\in\vA'$
will be controlled by \eqref{eq:Hes0} and \eqref{eq:Hes_est}, i.e. by
such term for $a=a_{\max}$.

\subsection{Commutation with $A_3^a$, $a\neq a_{\min}$}\label{subsec:Commutation
  with A_3^a}

Here we examine the operator $A_3^a$ appearing as a factor in $ N^a_\pm
$. For convenience we only look at the `plus' case, i.e. we study the
factor $A_3^a$ appearing as a factor in $ N^a_+
$. The `minus' case may be treated  similarly.  

Clearly the operators
\begin{align}\label{eq:notS}
  \begin{split}
  B_1&:=B_\delta^a,\\
B_2=B_{\delta,\rho_1}^a&:=r^{\rho_1/2}B_\delta^a
                         r^{\rho_1/2}=2\Re\parb{p^a\cdot
                           r^{\rho_1}\omega^a_\delta}
  \end{split}
\end{align}
 are both   of first order given in
 terms of  smooth  globally Lipschitz vector
 fields. As for  $B_1$ the vector field is bounded as well, and
 consequently all of the conclusions of Lemma \ref{lemma:Flow} hold
 upon replacing $B$ by  $B_1$. The operators $B_1$ and $B_2$
 are  essentially self-adjoint on $C^\infty_{\mathrm c}(\mathbf X)$ and
 the analogue of \eqref{eq:17100722} is  fulfilled for both of them. 

Now 
$A_3^a=\chi_-(B_2)$ is well-defined and we can record that
$A_3^a\in\vO(r^0)$, cf. Lemma \ref{lem:171113b}. From a formula similar
to the one of Lemma  \ref{lem:171113b} \ref{item:BrForm} we see that
$[A_3^a,r^s]$ is of order $s-1+\rho_1<s$, exemplifying good
commutation properties when  commuting  with functions of $r$. 

As for commutation with
functions of $H$ we  revisit Subsections
\ref{subsubsec:17111117}--\ref{subsubsec:A_1} except for the part  concerning the Mourre
estimate. This is technically the most demanding part of Section \ref{sec:Calculus considerations}.

\begin{lemma}\label{lem:2version}
Introducing  the quadratic form ${\bf D}_aB_2=\mathrm
i[H_a,B_2]$ as 
\begin{subequations}
\begin{align}
{\bf D}_aB_2
&=\slim_{t\to 0} t^{-1}\parb{H_a\e^{\i tB_2}-\e^{\i tB_2}H_a}\ \ 
\text{in }
{\mathcal L(\mathcal H^2,\mathcal H^{-1})
\cap \mathcal L(\mathcal H^1,\mathcal H^{-2})},
\label{eq:limit00b2}
\end{align}
 it is given by  
\begin{align}\label{eq:formulasy2}
  &{{\bf D}_aB_2} =r^{\rho_1/2}\parb{T_1+\cdots +T_5}r^{\rho_1/2};\\\nonumber
  T_1&=4p^a\cdot r^{-\delta}\parb{\mathop{\mathrm{Hess}  }
    r^a}\parb{x^a/r^\delta}  p^a,\\\nonumber
T_2&=-2\delta\Re\parb{p^a\cdot
  r^{-\delta}\parb{\mathop{\mathrm{Hess}  }  r^a}\parb{x^a/r^\delta}
  x^a r^{-1}B},\\\nonumber
T_3&=\rho_1\Re\parb{Br^{-1}B_\delta^a},\\\nonumber
T_4 &=2\sum_{b\leq a}
\Bigl(-(\omega^a_\delta)^b\cdot \bigl(\nabla^b V^b_{\rm lr}\bigr)
+\bigl(V^b_{\rm sr}(\omega^a_\delta)^b\bigr)\cdot \nabla^b
-\nabla^b\cdot
\bigl(V^b_{\rm sr}(\omega^a_\delta)^b\bigr)
+V^b_{\rm sr}\mathop{\mathrm{div}} (\omega^a_\delta)^b
\Bigr),\\\nonumber
T_5&=T_5(x)=\vO(r^{-3\delta})\text{ explicitly given by (summing over
     repeated indices):}\nonumber\\
-\partial^a_i&\partial^a_j\,\set{r^{-\delta}\parb{\mathop{\mathrm{Hess}  }
  r^a}_{ij}\parb{x^a/r^\delta}}-\delta (\Delta r)\mathop{\mathrm{div}^a}\set{r^{-1}
  F^a} -\tfrac 12
  \omega_j^a\parb{x^a/r^\delta} \partial^a_j\set{\abs{\omega}^2r^{-2}}\nonumber\\&
\quad \quad +\partial_i \partial^a_j\,M_{ij}-\delta r^{-1}\set{\partial^a_j (\Delta r)}F_j^a;
\nonumber\\
&F^a=F^a(x)=\set{\parb{{y^a\cdot \nabla_{y^a} \omega^a}
  }(y^a)}_{|y^a=x^a/r^\delta},\quad M_{ij}=M_{ij}(x)=\delta
  r^{-1}\omega_i F^a_j. \nonumber
\end{align}   
   \end{subequations}
\end{lemma}
\begin{remarks*} Although the stated exact  expression  for $T_5$  will
  not be relevant for us, we remark that 
   the  fifth term of the expression is on the form
  $\vO(r^{-2})$ (and no better), however it cancels with a term
  from expanding the fourth  term using the Leibniz rule for
  differentiation. The first term has the stated order
  $\vO(r^{-3\delta})$, but  the sum of all the others has a better
  order. The stated order $\vO(r^{-3\delta})$ will  by far be sufficient  for our treatment of  $T_5$.
\end{remarks*}

\begin{lemma}\label{lem:fHB3} 
Let    $f\in \mathcal F^t$ and  $F\in \mathcal F^{t'}$ with
$t<-1/2$ and $t'<1$, respectively. 
Then the  commutators $\mathrm i[f(H_a),B_2]$ and  $\mathrm
i[f(H_a),F(B_2)]$ extend  from $\mathcal D(B_2)$ and  $\mathcal
D(F(B_2))$ to
bounded quadratic forms 
on $\mathcal H$,  
 represented as
 \begin{subequations}
  \begin{align}\label{eq:2_1}
\mathrm i[f(H_a),B_2]
&=
-\int _{\C}(H_a-z)^{-1}\bigl(\mathbf D_aB_2\bigr) (H_a-z)^{-1}\,\mathrm d\mu_f(z)
\quad{\mand}\\
\mathrm i[f(H_a),F(B_2)]
&=
-\int _{\mathbb C}(B_2-z)^{-1}\bigl(\mathrm i[f(H_a),B_2]\bigr)(B_2-z)^{-1}\,\mathrm d\mu_F(z),\label{eq:2_2}
\end{align} 
 \end{subequations}
respectively. 
Moreover
\begin{align}\label{eq:firstHB3}
\mathrm i[f(H_a),B_2]=\vO(r^{\rho_1-\delta}) 
\quad \mand \quad
\mathrm i[f(H_a),F(B_2)]=\vO(r^{\rho_1-\delta})
.
\end{align}
\end{lemma}

Ideally we would like the right-hand sides of 
\eqref{eq:firstHB3} to be on the form $\vO(r^{(-1)_-})$, which is not
  doable. However there is the following partial result of this
  type: Let us  split
  \begin{align*}
    {{\bf D}_aB_2} &=D_1+D_2;\\
D_1&=
r^{\rho_1/2}\parb{T_1+T_2+T_3}r^{\rho_1/2}\\
D_2&=
r^{\rho_1/2}\parb{T_4+T_5}r^{\rho_1/2}.
  \end{align*}   From the the property \ref{item:2c} of Section \ref{sec:Derezinski's
  construction} (with $r$ replaced by  $r^a$) and the relations 
$(1+\mu)\delta-\rho_1>1$ 
 and $3\delta-\rho_1>1$  it  follows that  
\begin{align*}
  -\int _{\C}(H_a-z)^{-1}D_2 (H_a-z)^{-1}\,\mathrm d\mu_f(z)=\vO(r^{(-1)_-}),
\end{align*} and therefore we conclude by  plugging into 
\eqref{eq:2_1} and \eqref{eq:2_2} that the contribution  from $D_2$ to
$\mathrm i[f(H_a),F(B_2)]$  is on the desired form $\vO(r^{(-1)_-})$. 
As for  the contribution  from $D_1$ to
$\mathrm i[f(H_a),F(B_2)]$ we write 
\begin{align*}
  &\mathrm i[f(H_a),F(B_2)]-\vO(r^{(-1)_-})
\\
   &=\Re\parb{ T^{(1)} F'(B_2)}+\tfrac12 \int _{\mathbb C}(B_2-z)^{-2}\mathrm
  \ad^2_{\i B_2}(T^{(1)})(B_2-z)^{-2}\,\mathrm d\mu_F(z);\\
T^{(1)}&:=-\int _{\C}(H_a-z)^{-1}D_1 (H_a-z)^{-1}\,\mathrm d\mu_f(z)=\vO(r^{\rho_1-\delta}).
\end{align*}  
The second  term to the right is $\vO(r^{(-1)_-})$,  which follows from  the
following more complicated  version of Lemma \ref{lem:17111515} 
involving  
\begin{align*}
 T^{(2)}&:=
-\int_{\C} (H_a-z)^{-1}\ad_{\i B_2}(D_1)(H_a-z)^{-1}\d \mu_f(z)
\\&\quad \quad 
+2\int_{\C} (H_a-z)^{-1}D_1(H_a-z)^{-1} D_1(H_a-z)^{-1}\d
    \mu_f(z)= \vO(r^{2(\rho_1-\delta)}).
\end{align*}  For convenience we write in the remaining part of
the section
  $T\simeq T'$ for any
given  operators
$T,T'\in \vL(\vH)$, if $T-T'=\vO(r^{(-1)_-})$.

\begin{lemma}\label{lemma:sing2} 
Let  $T^{(1)}$ and  $T^{(2)}$ be given as above for   $f\in \mathcal F^t$ with
$t<-1/2$. Then there exists the strong limits 
\begin{align*} 
\ad_{\i B_2}(T^{(j)})&=\mathrm i[T^{(j)},B_2]:=\slim_{t\to 0}t^{-1}
\bigl(T^{(j)}\mathrm e^{\mathrm i tB_2}-\e^{\mathrm i
                                                 tB_2}T^{(j)}\bigr)
\ \ \text{in }\mathcal L(\mathcal H);\,\quad j=1,2.
\end{align*}  Moreover 
\begin{align*}
\ad_{\i B_2}\parb{T^{(1)}} 
&
=T^{(2)}+\vO(r^{(-1)_-}),
\end{align*} and  for   the above second term we can substitute
\begin{align*}  
\ad^2_{\i B_2}(T^{(1)})&=\ad_{\i
  B_2}(T^{(2)})+\i\parb{\ad_{\i
  B_2}(T^{(1)})-T^{(2)}}B_2-\i B_2\parb{\ad_{\i
  B_2}(T^{(1)})-T^{(2)}}\\
&=\ad_{\i
  B_2}(T^{(2)})+\vO(r^{(-1)_-})B_2+ B_2\vO(r^{(-1)_-})\\&\quad \quad  \simeq\vO(r^{(-1)_-})B_2+ B_2\vO(r^{(-1)_-}).
\end{align*} 
\end{lemma}
\begin{proof} We compute (for the very last assertion)
\begin{align*}
&\ad_{\i
  B_2}(T^{(2)}) 
=-\int_{\C} (H_a-z)^{-1}\ad^2_{\i B_2}(D_1)(H_a-z)^{-1}\d \mu_f(z)
\\ &\phantom{{}={}}{}
+3\int_{\C} (H_a-z)^{-1}\ad_{\i B_2}(D_1)(H_a-z)^{-1} D_1(H_a-z)^{-1}\d
    \mu_f(z)\\
&+3\int_{\C} (H_a-z)^{-1} D_1(H_a-z)^{-1}\ad_{\i B_2}(D_1)(H_a-z)^{-1}\d
    \mu_f(z)\\
&-6\int_{\C} (H_a-z)^{-1}D_1(H_a-z)^{-1} D_1(H_a-z)^{-1} D_1(H_a-z)^{-1}\d
    \mu_f(z)+\vO(r^{(-1)_-})\\
&\simeq \vO(r^{3(\rho_1-\delta)})\simeq 0.
\end{align*} 
  \end{proof}

By  commutation  --  of the same sort,  continuing the above
calculations --  we conclude  the
following result.

\begin{subequations}

\begin{lemma}\label{lem:171115152} Let    $f\in \mathcal F^t$ and   $F\in \mathcal F^{t'}$ with
$t<-1/2$ and $t'<1$, respectively. Suppose also that $ F'\leq 0$  with
$\sqrt{-F'}\in \mathcal F^0$. Then 
\begin{align*}
  \mathrm i[f(H_a),F(B_2)]
\simeq  -\sqrt{-F'}(B_2)\parb{\i[f(H_a),B_2]}\sqrt{-F'}(B_2).
\end{align*}  In particular for any real functions  $f,g$ and ${F}$ with $ F'\leq 0$,
$f,g,\sqrt{-F'}\in C^\infty_\c(\R)$ and given such that   $g(\lambda)f(\lambda)=\lambda f(\lambda)$,
\begin{align}\label{eq:basicF}
  \begin{split}
  &f(H_a)\mathrm i[g(H_a),F(B_2)]f(H_a)\\
&\quad\simeq  -\sqrt{-F'}(B_2)f(H_a)\parb{\i[g(H_a),B_2]}f(H_a)\sqrt{-F'}(B_2)
\\
&\quad\simeq  -\sqrt{-F'}(B_2)f(H_a)D_1f(H_a)\sqrt{-F'}(B_2).  
  \end{split}
\end{align} 
\end{lemma}

Effectively the   right-hand side  of \eqref{eq:basicF} will
contribute by a negative term. 
To see this we first write
\begin{align*}
  T_1+T_2&=4\parb{p^a-\tfrac \delta 4B\tfrac {x^a}r}\cdot r^{-\delta}\parb{\mathop{\mathrm{Hess}  }
    r^a}\parb{x^a/r^\delta}  \parb{p^a-\tfrac \delta 4\tfrac
  {x^a}r B}\\
&\quad -\tfrac{\delta^2}4B\tfrac {x^a}r\cdot r^{-\delta}\parb{\mathop{\mathrm{Hess}  }
    r^a}\parb{x^a/r^\delta} \tfrac
  {x^a}r B.
\end{align*} Using \eqref{eq:r2} this leads to
\begin{align}
  &f(H_a)r^{\rho_1/2}\parb{T_1+T_2}r^{\rho_1/2}f(H_a) \label{eq:pos1}\\
&=f(H_a)r^{\rho_1/2}\parb{4\parb{p^a-\tfrac \delta 4B\tfrac {x^a}r}\cdot r^{-\delta}\parb{\mathop{\mathrm{Hess}  }
    r^a}\parb{x^a/r^\delta}  \parb{p^a-\tfrac \delta 4\tfrac
  {x^a}r B}}r^{\rho_1/2}f(H_a)\nonumber\\& 
 \quad +\vO(r^{\rho_1+\delta-2}).\nonumber
\end{align} 
 Now the first term to the right is positive, and the last term
 $\vO(r^{\rho_1+\delta-2})=\vO(r^{(-1)_-})$.

We are only interested in $f\in  C^\infty_\c(\R)$  taking values in $[0,1]$. For convenience let us
henceforth impose this condition.
 The contribution from $T_3$ is of order $\vO(r^{-1})$ (here we use
that $F'$ is compactly supported), but with
a localization  to $B>0$ and provided $\sqrt{-F'}\in C^\infty_\c(\R_+)$
 we can again find a better lower bound. More
precisely we have the following result for which it remains to be
noted that the  
right-hand side   of \eqref{eq:Pos2} is positive.
\begin{lemma}\label{lemma:comm-with-a_3x}
Suppose  $\zeta:=\sqrt{-F'}\in C^\infty_\c(\R_+)$. Then 
for any $\epsilon>0$ and with $\eta(b):=\sqrt{b}\chi_+(
  b/\epsilon)$ and  $\zeta_1(b):=\sqrt{b}\zeta(b)$,
\begin{align}\label{eq:Pos2}
  \begin{split}
   &\chi_+(
  B/\epsilon) \zeta(B_2)f(H_a)r^{\rho_1/2}T_3r^{\rho_1/2}f(H_a)\zeta(B_2)\chi_+(
  B/\epsilon)\\& \simeq \rho_1f(H_a)\eta(B) \zeta_1(B_2)r^{
-1}\zeta_1(B_2)\eta(B)f(H_a). 
  \end{split}
\end{align}   
\end{lemma}
  \end{subequations}
\begin{proof}
  Pick $\tilde{f}\in C^\infty_\c(\R)$ such that
  $f\prec\tilde{f}$, and let $\widetilde{B}_a=\tilde{f}(H_a)B\tilde{f}(H_a)$. Then by commutation
\begin{align*}
T:=&\chi_+(
  B/\epsilon) \zeta(B_2)f(H_a)r^{\rho_1/2}   \Re\parb{Br^{-1}B_\delta^a}  r^{\rho_1/2}f(H_a) \zeta(B_2)\chi_+(
  B/\epsilon)
\\ &\simeq f(H_a) \chi_+(
  B/\epsilon) \Re\parb{\zeta(B_2)\widetilde{B}_ar^{-1}B_2\zeta(B_2)}  \chi_+(
  B/\epsilon)f(H_a).
\end{align*}  It follows from  the first bound of
\eqref{eq:firstHB3} and   the bound $\tilde f{(H_a)}
[B_2,{B}]\tilde f{(H_a)} =\vO(r^{\rho_1-\delta})$ that
\begin{subequations}
 \begin{align}\label{eq:comtild}
[\zeta(B_2),\widetilde{B}_a]
=-
\int _{\C}
(B_2-z)^{-1}\mathrm [B_2,\widetilde{B}_a](B_2-z)^{-1}\,\mathrm d\mu_\zeta(z)=\vO(r^{\rho_1-\delta}).
\end{align}
Whence 
\begin{align*}
  &\Re\parb{\zeta(B_2)\widetilde{B}_ar^{-1}B_2\zeta(B_2)} \\
&=\Re\parb{\widetilde{B}_a\zeta(B_2)r^{-1}B_2\zeta(B_2)}
  +\Re\parb{[\zeta(B_2),\widetilde{B}_a]r^{-1}B_2\zeta(B_2)}\\
&=\Re\parb{\widetilde{B}_ar^{-1}\zeta^2_1(B_2)} +\vO(r^{\rho_1-\delta-1}),
\end{align*} and we infer that 
\begin{align*}
T&\simeq f(H_a) \chi_+(
  B/\epsilon) \Re\parb{\widetilde{B}_ar^{-1}\zeta^2_1(B_2)} \chi_+(
  B/\epsilon)f(H_a)\\
&\simeq f(H_a) \Re\parb{\chi_+(
  B/\epsilon) {B}r^{-1}\zeta^2_1(B_2)\chi_+(
  B/\epsilon)} f(H_a).
\end{align*}  
Next we write
\begin{align*}
  &\Re\parb{\chi_+(
  B/\epsilon) {B}r^{-1}\zeta^2_1(B_2)\chi_+(
  B/\epsilon)}\\& =\Re\parb{\eta^2(B)r^{-1}\zeta^2_1(B_2)} +\Re\parb{\chi_+(
  B/\epsilon) {B}[r^{-1}\zeta^2_1(B_2),\chi_+(
  B/\epsilon)]}, 
\end{align*} and (for the second term)
\begin{align*}
  [r^{-1}\zeta^2_1(B_2),&\chi_+(
  B/\epsilon)]\\&=[r^{-1},\chi_+(
  B/\epsilon)]\zeta^2_1(B_2)+r^{-1}[\zeta^2_1(B_2),\chi_+(
  B/\epsilon)]\\
&=\vO(r^{-2})\zeta^2_1(B_2)+r^{-1}[\zeta^2_1(B_2),\chi_+(
  B/\epsilon)].
\end{align*} Here the first term  contributes by a term on the form $\vO(r^{(-1)_-})$, so
we are left with
\begin{align*}
  T\simeq f(H_a) \Re\parb{\eta^2(B)&r^{-1}\zeta^2_1(B_2)}f(H_a) \\&-f(H_a) \Re\parbb{\chi_+(
  B/\epsilon) {B}r^{-1}[\chi_+(
  B/\epsilon),\zeta^2_1(B_2)]}f(H_a) .
\end{align*} The first term
\begin{align*}
  f(H_a) \Re\parb{\eta^2(B)&r^{-1}\zeta^2_1(B_2)}f(H_a) 
\\\simeq f(H_a)&
    \Re\parb{\eta(B)\zeta_1(B_2)r^{-1}\zeta_1(B_2)\eta(B)}f(H_a)\\&
    +f(H_a) \Re\parb{\eta(B)r^{-1}[\eta(B),\zeta^2_1(B_2)]}f(H_a),
\end{align*}  and here in turn the first term is what we need for
\eqref{eq:Pos2}.

It remains to show that
\begin{align*}
  -f(H_a) \Re\parbb{\chi_+(
  B/\epsilon) &{B}r^{-1}[\chi_+(
  B/\epsilon),\zeta^2_1(B_2)]}f(H_a) \\&+f(H_a) \Re\parb{\eta(B)r^{-1}[\eta(B),\zeta^2_1(B_2)]}f(H_a)\simeq 0.
\end{align*}
 To deal with these commutators we  let 
 $\widetilde{Z}_a=\tilde{f}(H_a)Z\tilde{f}(H_a)$, $Z=\zeta^2_1(B_2)$. By
 \eqref{eq:firstHB} it suffices to show that
 \begin{align}
\label{eq:comtilde}
   [\chi_+(
  B/\epsilon),\widetilde{Z}_a], \,[\eta(B),\widetilde{Z}_a]=\vO(r^{0_-}).
 \end{align}  
\end{subequations}
Let us only do the second bound. We use
 \eqref{eq:firstHB} and 
\eqref{eq:firstHB3} computing  in the last step exactly as in \eqref{eq:comtild}
\begin{align*}
[\eta(B),\widetilde{Z}_a]
&=-
\int _{\C}
(B-z)^{-1}\mathrm [B,\widetilde{Z}_a](B-z)^{-1}\,\mathrm
d\mu_\eta(z)
\\
&=-
\int _{\C}
(B-z)^{-1}\mathrm [\widetilde B_a,\zeta^2_1(B_2)](B-z)^{-1}\,\mathrm
  d\mu_\eta(z)+\vO(r^{\rho_1-\delta})
\\
&=\vO(r^{\rho_1-\delta})=\vO(r^{0_-}).
\end{align*}
\end{proof}

\begin{remark}\label{remark:comm-with-a_3yy} The  procedure in the
  proof of Lemma \ref{lemma:comm-with-a_3x} of regularizing by
  replacing an operator $T$ by
  $\widetilde{T}_{a}=\tilde{f}(H_a)T\tilde{f}(H_a)$ before computing a
  commutator with $T$ by the Helffer--Sj\"ostrand
formula will be used
  frequently in the remaining part of the paper. We are dealing  with various
  examples of first order operators,  most typically   $B$, $B_{\delta,\rho_1}^a$ and $\xi^+_j( x)
  G_{d_j}$, that along with functions of $B$ and
  $B_{\delta,\rho_1}^a$  at
  various points need to be commuted. The technical problem one should
  avoid (and possibly would encounter upon `blindly' applying  the Helffer--Sj\"ostrand
formula)  is  the appearance of certain (presumably unbounded)
commutators  without a proper  weight. Thus for
example the formal commutator  $[B,B_{\delta,\rho_1}^a]$ is of first
order but  does not seem  relatively  bounded  neither to $B$ nor
 to $B_{\delta,\rho_1}^a$. In the
  proof of Lemma \ref{lemma:comm-with-a_3x} the $\widetilde{T}_{a}$-construction  cured
 this problem providing  $H_a$-weights. We used `free' factors of
$\tilde{f}(H_a)$ at our disposal and the feature that commutation with those was
relatively harmless. We will proceed similarly later (mostly with  cases
where $H_a$ is replaced by  $H$), although typically 
without giving 
 the details.
  
\end{remark}

We have computed the commutator \eqref{eq:basicF} as a sum of two negative operators plus an
operator on the form $\vO(r^{(-1)_-})$. This is  with
a localization  to $B>0$ and for 
$\sqrt{-F'}\in C^\infty_\c(\R_+)$, which will be   conditions  met in 
our treatment of $A^a_3$ in the `plus' case.

We shall in applications need
commutation with  $A^a_3$, however  with $H_a$ replaced by $H$. The presence of
factors of $M_a$   enables us to reduce to the case studied
above. Note  that thanks to \eqref{eq:partM} for  any real ${f}\in C^\infty_\c(\R)$
\begin{align}\label{eq:comMloc}
  \parb{f(H_a)-f(H)}\tilde\xi_a^+=\vO(r^{-\mu}),
\end{align}  and note that
$\vO(r^{-\mu})\vO(r^{\rho_1-\delta})=\vO(r^{(-1)_-})$ (the second
factor appears in  \eqref{eq:firstHB3}).
Alternatively, one can essentially just mimic the above arguments  with
$H_a$ replaced by $H$; the additional complication is minor thanks to \eqref{eq:partM}.

\subsection{Commutation with $A_2^a$,  $a\neq a_{\min}$}\label{subsec:A2}
Recalling 
\begin{align*}
  A_2^a =\chi_-\parb{
  r^{\rho_2-1}r_\delta^a}\text{ with } r_\delta^a=
r^\delta r^a(x^a/r^\delta),
\end{align*} we compute for  real ${f}\in C^\infty_\c(\R)$
\begin{align*}
  \i [f(H), A_2^a]&=-2\int _{\C}
(H-z)^{-1}\mathop{\mathrm{Re}}\parb{\chi_-'(\cdot)\nabla(
  r^{\rho_2-1}r_\delta^a)\cdot p } (H-z)^{-1}\,\mathrm d\mu_f(z)\\
&=-\int _{\C}
(H-z)^{-1}\parb{S_1+S_2+S_3}(H-z)^{-1}\,\mathrm d\mu_f(z);\\
&\quad S_1=\mathop{\mathrm{Re}}\parb{\chi_-'(\cdot)r^{(\rho_2-1)}B^a_\delta},\\
&\quad S_2=(\rho_2-1)\mathop{\mathrm{Re}}\parb{\chi_-'(\cdot)
  r^{\rho_2-2}r_\delta^a B}, \\
&\quad S_3=\delta
  \mathop{\mathrm{Re}}\parb{\chi_-'(\cdot)r^{\rho_2+\delta-2}\parb{
  r^{-\delta}r_\delta^a -\omega_\delta^a\cdot \tfrac{x^a}{r^\delta}} B}.
\end{align*}

Due to \eqref{eq:r1}   we can write
\begin{align*}
 S_3= \mathop{\mathrm{Re}}\parb{\vO(r^{\rho_2+\delta-2})B}= \mathop{\mathrm{Re}}\parb{\vO(r^{(-1)_-})B},
\end{align*} and consequently the contribution from $S_3$ to the
commutator is of order $\vO(r^{(-1)_-})$. On the other hand  the
contributions from $S_1$ and $S_2$  are of order $\vO(r^{\rho_2-1})$
and $\vO(r^{-1})$, respectively, and no better.

Applied to real $f,g\in C^\infty_\c(\R)$ with 
$g(\lambda)f(\lambda)=\lambda f(\lambda)$ (as we used for
\eqref{eq:basicFB} and \eqref{eq:basicF}) we estimate  with $\eta(b)=\sqrt{b}\chi_+(
  b/\epsilon)$   and for any $\epsilon>0$, $\zeta(b)=-b \chi_+(-b)$,
   $\chi_1(s)=-(\chi^2_-)'(s)=2\chi_+(s)\chi_+'(s)$, $\chi_2(s)=-s(\chi^2_-)'(s)=2s\chi_+(s)\chi_+'(s)$
\begin{align}\label{eq:comlong}
  &2\chi_+(
  B/\epsilon) \Re\parbb{f(H) \i [g(H), A_2^a]f(H)\chi_-(B_2)A_2^a}\chi_+(
  B/\epsilon)\nonumber\\&\simeq 2\chi_+(
  B/\epsilon) \Re\parbb{f(H) \parb{S_1+S_2}f(H)\chi_-(B_2)A_2^a}\chi_+(
  B/\epsilon)
 \\&\simeq  -\chi_+(B/\epsilon) f(H) r^{(\rho_2-\rho_1-1)/2}\sqrt{\chi_1}(\cdot)B_2\chi_-(B_2) \sqrt{\chi_1}(\cdot)r^{(\rho_2-\rho_1-1)/2}f(H) \chi_+(
  B/\epsilon)\nonumber\\
&\quad +(1-\rho_2)\eta(
  B) f(H)\chi^{1/2}_-(B_2) \parb{\chi_1(\cdot)
  r^{\rho_2-2}r_\delta^a }\chi^{1/2}_-(B_2)f(H)\eta(B)\nonumber
\\&\simeq  \chi_+(B/\epsilon) f(H) r^{(\rho_2-\rho_1-1)/2}\sqrt{\zeta}(B_2)\chi_1\parb{r^{\rho_2-1}r_\delta^a }\sqrt{\zeta}(B_2) r^{(\rho_2-\rho_1-1)/2}f(H)\chi_+(
  B/\epsilon) \nonumber\\
&\quad +(1-\rho_2) \eta(
  B) f(H)\chi^{1/2}_-(B_2){\chi_2\parb{r^{\rho_2-1}r_\delta^a }
  }r^{-1}\chi^{1/2}_-(B_2)f(H)\eta(B). \nonumber
  \end{align}  We commuted repeatedly functions of $B_2$ and $B$ with multiplication
operators (for example $A_2^a$). Using the fact that $\rho_1<1/3$ this
is rather harmless (we  do not give  the details).  Note also that the above treatment of $S_2$ used commutation of
functions of $B$ and $\chi_-(B_2)$. This is implemented as explained in  Remark
\ref{remark:comm-with-a_3yy} (with $H_a$  replaced by
$H$).

To the
right of the  computation  \eqref{eq:comlong} both terms are positive.

\section{Computation of $T^\pm_\alpha$ and existence of
   $\Omega^\pm_\alpha$}\label{sec:Computation of T}

We pick  $\tilde{f_1},\tilde{f_2} \in C^\infty_\c(\Lambda)$ such that   $f_1\prec \tilde
f_2\prec\tilde{f_1}$ (cf.  Subsection \ref{subsec:Limiting absorption
  principles}) and  let    $g(\lambda)=\lambda\tilde f_1(\lambda)$.
This allows us to  write 
\begin{align*}
  T^\pm_\alpha=\i f_1(H) \parbb{g(H)  M_a N^a_\pm M_a - M_a N^a_\pm M_ag(\brH_a)}f_1(\brh_\alpha)f_1(\brH_a),
\end{align*} and with this representation  we can show  that $
T^\pm_\alpha=\vO(r^{\rho_1-\delta})$ in the sense of Subsection
\ref{subsubsec:Smooth sign function}. Furthermore we can show that all
terms in an expansion that are not on the form $\vO(r^{-1-2\epsilon})$
for some $\epsilon>0$ tend to have a sign allowing us to prove
non-trivial smoothness estimates for concrete expressions  that fail
to be on the favourite  form 
$\vO(r^{-1/2-\epsilon})$. Note that for any $\epsilon>0$ the
operators 
$r^{-1/2-\epsilon}f_1(\brh_\alpha)f_1(\brH_a)$ and $r^{-1/2-\epsilon}f_1(H)$
are $\brH_a$- and $H$-smooth in the sense of Kato (see  \cite{Ka}), cf.
\eqref{eq:LAPbndbr} and \eqref{eq:LAPbnd}. 

\subsection{Propagation estimates}\label{subsec:Propagation estimates}
For any $\psi\in \vH$ and any self-adjoint operator $T$ on $\vH$ we
let $\psi(t)=\e^{-\i tH } \psi$  and denote
$\inp{T}_t=\inp{\psi(t),T \psi(t)}$, $t\in
\R$. In the proof of the following lemma we calculate modulo  terms
on the form $f_1(H)\vO(r^{(-1)_-})f_1(H)$.  We will in this subsection
write $T_1\simeq_1 T_2$ if $T_1-T_2$ has this form (which is slightly
different from the relation used in Subsections \ref{subsec:Commutation
  with A_3^a}--\ref{subsec:A2}). 
Recall the smoothness bound
\begin{align}\label{eq:smoothBnd}
  \int^\infty_{-\infty} \,\norm{Q_0f_1 (H)\psi(t)}^2\,\d t\leq
  C_\epsilon \norm{\psi}^2;\quad  Q_0=r^{-1/2-\epsilon},\,\epsilon>0.
\end{align} We will below extensively use  the notation
$B_{\delta,\rho_1}^a=r^{\rho_1/2}B_\delta^ar^{\rho_1/2}$ for
$a\in\vA'\setminus \set{a_{\min}}$, cf. \eqref{eq:notS}. Recall that
the factors $A^a_3=A^a_{3\pm}$ of $N^a_\pm $ are  given in terms of
$B_{\delta,\rho_1}^a$ with parameters given as in
 \eqref{eq:parameters}. In the list below \ref{item:H1}--\ref{item:H8}
there appear respective operators $Q_1\dots Q_8$ (equipped with
additional indices). For $a={a_{\min}}$ only  $Q_1$ and $Q_2$ enter.
They all have the feature of being $H$-bounded
and (if wanted) we can make them bounded by looking instead at
$Q_1\tilde f_1 (H)\dots Q_8\tilde f_1 (H)$, respectively. None of
these operators appears to be on the form $\vO(r^{(-1/2)_-})$, so the
corresponding bounds do not conform with \eqref{eq:smoothBnd}. On the
other hand some are on the form $\vO(r^{-1/2})$, just missing the
applicability of \eqref{eq:smoothBnd}, and the `worse' appearing order in
the list is  $\vO(r^{(\rho_1-\delta)/2})$.
\begin{lemma}\label{lemma:Hbnds} 
  Let $a\in\vA'$. The following bounds hold for a constant $C>0$
    being independent of $\psi\in \vH$. 
\begin{enumerate}[1)] 
  \item \label{item:H1} For $j\leq J=J(a)$ (as in  \eqref{eq:Hes0})
    \begin{align*}
  \int^\infty_{-\infty} \,\norm[\big]{Q_1f_1 (H)\psi(t)}^2\,\d t\leq
  C \norm{\psi}^2;\quad  Q_1=Q_1(a,j)=\xi^+_j( x) G_{d_j} .
\end{align*}
\item \label{item:H2} 
    \begin{align*}
  \int^\infty_{-\infty} \,\norm[\big]{Q_2^\pm f_1 (H)\psi(t)}^2\,\d t\leq
  C \norm{\psi}^2;\quad Q_2^\pm=r^{-1/2}\sqrt{(\chi^2_+)'}\parbb{\pm B/\epsilon_0}.
\end{align*}
\item \label{item:H3} Let $ T_1^a= \sqrt{-\chi'_-}\parb{\pm B_{\delta,\rho_1}^a}{\chi_+}(\pm
B/\epsilon_0)$ and 
% for $\kappa\in(\rho_1-\delta+1,1)$
$\vH^a={\parb{\mathop{\mathrm{Hess}  }r^a}\parb{x^a/r^\delta}} $.
Then 
    \begin{align*}
  \int^\infty_{-\infty} \,&\norm[\big]{ Q_3^\pm f_1 (H)\psi(t)}^2\,\d t\leq
  C \norm{\psi}^2;\\&\quad Q_3^\pm=Q_3^\pm(a)=2\parb{\vH^a}^{1/2}\parb{p^a\mp\tfrac \delta 4\tfrac
  {x^a}r
  B}\tilde{f}_2(H_a) r^{(\rho_1-\delta)/2}T_1^aM_a.
\end{align*}
\item \label{item:H4} Let $\eta(b)=\sqrt{b}\chi_+(
  b/\epsilon_0)$, $\zeta_1(b)=\sqrt{-b\chi'_-(
  b)}$  and $T_{2}^a={\zeta_1}\parb{\pm B_{\delta,\rho_1}^a}{\eta}(\pm
      B)$. Then 
    \begin{align*}
  \int^\infty_{-\infty}
      \,&\norm[\big]{Q_4^\pm f_1 (H)\psi(t)}^2\,\d t\leq
  C \norm{\psi}^2;\\&\quad Q_4^\pm=Q_4^\pm(a)=\rho_1^{1/2} r^{-1/2}T^a_2M_a.
\end{align*}
\item \label{item:H5} Let $T_{1+}^a= \sqrt{-\chi'_-}\parb{\pm B_{\delta,\rho_1}^a}\chi_+\parb{r^{\rho_2-1}r_\delta^a
}{\chi_+}(\pm B/\epsilon_0)$ and  $\vH^a$ be given as in \ref{item:H3}. Then 
    \begin{align*}
  \int^\infty_{-\infty} \,&\norm[\big]{Q_5^\pm f_1 (H)\psi(t)}^2\,\d t\leq
  C \norm{\psi}^2;\\&\quad Q_5^\pm=Q_5^\pm(a)=2\parb{\vH^a}^{1/2}\parb{p^a\mp\tfrac \delta 4\tfrac
  {x^a}r
  B}\tilde{f}_2(H_a)r^{(\rho_1-\delta)/2}T_{1+}^aM_a.
\end{align*}
\item \label{item:H6} Let $\eta(b)=\sqrt{b}\chi_+(
  b/\epsilon_0)$,   $\zeta_1(b)=\sqrt{-b\chi'_-(
  b)}$ (as in \ref{item:H4}) and \newline $T_{2+}^a= {\zeta_1}\parb{\pm B_{\delta,\rho_1}^a}\chi_+\parb{r^{\rho_2-1}r_\delta^a
      }{\eta}(\pm B)$. Then 
    \begin{align*}
  \int^\infty_{-\infty}
      \,&\norm[\big]{Q_6^\pm f_1 (H)\psi(t)}^2\,\d t\leq
  C \norm{\psi}^2;\\&\quad Q_6^\pm=Q_6^\pm(a)=\rho_1^{1/2} r^{-1/2}T_{2+}^aM_a.
\end{align*}
\item \label{item:H7} Let $\zeta_2(b)=\sqrt{-b \chi_+(-b)}$, 
  $T_3^a =\zeta_2\parb{\pm B_{\delta,\rho_1}^a}\tilde{f}_2(H_a)\chi_+(
  \pm B/\epsilon_0)$ and \newline $\chi_{1+}(s)=2\chi_+(s)\chi_+'(s)$. Then 
    \begin{align*}
  \int^\infty_{-\infty} \,&\norm[\big]{Q_7^\pm f_1 (H)\psi(t)}^2\,\d t\leq
  C \norm{\psi}^2;\\&\quad Q_7^\pm=Q_7^\pm(a)={r^{(\rho_2-\rho_1-1)/2}\sqrt{\chi_{1+}}\parb{r^{\rho_2-1}r_\delta^a }
  }T_3^aM_a.
\end{align*}
\item \label{item:H8} Let $\eta(b)=\sqrt{b}\chi_+(
  b/\epsilon_0)$ (as in \ref{item:H4}),
 $T_4^a=\chi^{1/2}_-\parb{\pm B_{\delta,\rho_1}^a}\eta(\pm
      B)$ and  $\chi_{2+}(s)=2s\chi_+(s)\chi_+'(s)$. Then 
    \begin{align*}
  \int^\infty_{-\infty} \,&\norm[\big]{Q_8^\pm f_1 (H)\psi(t)}^2\,\d t\leq
  C \norm{\psi}^2;\\&\quad Q_8^\pm=Q_8^\pm(a)=(1-\rho_2)^{1/2} {r^{-1/2}\sqrt{\chi_{2+}}\parb{r^{\rho_2-1}r_\delta^a }
  }T_4^aM_a.
\end{align*}
\end{enumerate}
\end{lemma}
\begin{proof} 
We only consider the upper sign case. The lower sign case assertions can be
demonstrated similarly.

{\bf \ref{item:H1}:}
 We choose  $m=m_j$ conforming with \eqref{eq:Hes_est} and let $M=M_j$ be the
 corresponding operator  defined in  \eqref{eq:M}. Applying
 \eqref{eq:firstHB22} and \eqref{eq:Hes_est} we can then estimate
\begin{align*}
  &{\i f_1(H) [g(H),  M] f_1(H)}\\& \simeq_1
                                                          4{
                                                          f_1(H)
                                                          p\cdot \parb{\chi_+(\abs{x})\mathop{\mathrm{Hess}}
                                                          m}p
                                                     f_1(H)}
\\& \geq 4f_1(H) G^*_{d_j}\chi_+(\abs{x})\xi^+_j( x)^2
  G_{d_j} f_1(H).
\end{align*} Let \begin{align*}
  \Psi_{1,j}=  f_1(H) M_j f_1(H)
\end{align*} Since  $\Psi_1=\Psi_{1,j}$ is   bounded, \ref{item:H1} follows from
integration of $\tfrac{\d}{\d t}\inp{\Psi_1}_t=\inp{\i[g(H),\Psi_1]}_t$ using the above
estimation   and 
\eqref{eq:smoothBnd}.

 {\bf \ref{item:H2}:}  Let \begin{align*}
  \Psi_2=  f_1(H) {\chi^2_+}(B/\epsilon_0) f_1(H).
\end{align*} We obtain \ref{item:H2} by integrating  $\tfrac{\d}{\d
  t}\inp{\Psi_2}_t=\inp{\i[g(H),\Psi_2]}_t$ using \eqref{eq:basicFBB}, \eqref{eq:smoothBnd}
and the boundedness of  $\Psi_2$.  Here 
\eqref{eq:basicFBB} is used with  $f=\tilde f_2$. Upon    multiplying   both
sides of that  estimate from the left and from the right by a factor
$f_1(H)$  the second term to the right gets the good form $\simeq_1 0$.

{\bf \ref{item:H3}, \ref{item:H4}:} 
 Let  $A_1=\chi_+(
  B/\epsilon_0)$, $A_3^a=\chi_-\parb{B_{\delta,\rho_1}^a}$ (as in
\eqref{eq:propgaObs} for the `plus case') and 
\begin{align*}
  \Psi_3=  f_1(H) M_aA_1A_3^aA_1M_af_1(H).
\end{align*} We obtain \ref{item:H3} and  \ref{item:H4} by integrating  $\tfrac{\d}{\d
  t}\inp{\Psi_3}_t$ using   \ref{item:H1}, \ref{item:H2},
\eqref{eq:smoothBnd}, 
 the boundedness of  $\Psi_3$ and by using \eqref{eq:basicF}--\eqref{eq:comMloc}
(as well as  the arguments for the latter  assertions and the
accompanying   comments). 
By the  product rule for commutation  there are several
terms. Let us first treat the contribution
from $\i[g(H), M_a]$: We write using 
\eqref{eq:firstHB22}, \eqref{eq:Hes0}  and  Remark \ref{remark:comm-with-a_3yy}
\begin{align*}
  &2f_1(H) \Re\parb{\i [g(H),  M_a] A_1A_3^aA_1M_a}f_1(H)\\
& \simeq_1
                                                          8f_1(H)\Re\parb{
                                                          p\cdot \parb{\mathop{\mathrm{Hess}}
                                                          m_a}p
                                                          A_1A_3^aA_1 M_a}
                                                          f_1(H)
\\& \simeq_1 8\Sigma_{j\leq J}f_1(H) \Re\parbb{ \,
    G^*_{d_j} \xi^+_j\vG_j \parb{\tilde{f}_1(H)\xi^+_jG_{d_j}\tilde{f}_1(H)}  A_1A_3^aA_1
    M_a} f_1(H)
\\& \simeq_1 \Sigma_{j\leq J}\,T_j;\quad T_j={ f_1(H) \, G^*_{d_j}\xi^+_j{\vG'_j}\xi^+_j G_{d_j} f_1(H)}\text{ with } {\vG'_j}\text{ bounded}.
\end{align*}  Indeed the last expression arises by commuting the
factor $\tilde{f}_1(H)\xi^+_jG_{d_j}\tilde{f}_1(H)$ to the right producing an error on the
form $\vO(r^{\rho_1-\delta-1})=\vO(r^{(-1)_-})$ along the pattern of  Remark \ref{remark:comm-with-a_3yy}. 
 Now we   apply \cas and \ref{item:H1} to each term $T_j$ on the
right-hand side.

As for the contribution
from $\i[g(H), A_1]$ we use \eqref{eq:basicFB},  writing
\begin{align*}
  &2f_1(H) \Re\parb{\i \tilde f_2(H) M_a[g(H),  A_1] A_3^aA_1M_a\tilde
    f_2(H) }f_1(H) \\& \simeq_1
2f_1(H) \Re\parb{\i M_a\tilde f_2(H) [g(H),  A_1]\tilde f_2(H)  A_3^aA_1M_a}f_1(H) \\& \simeq_1
                                                          f_1(H)\sqrt{(\chi^2_+)'}\parbb{\pm B/\epsilon_0} \vO(r^{-1})\sqrt{(\chi^2_+)'}\parbb{\pm B/\epsilon_0}f_1(H).
             \end{align*} 
 The right-hand side is treated by  \ref{item:H2}.

Next let us elaborate on how
\eqref{eq:basicF}--\eqref{eq:comMloc} are  used for  treating  the contribution
from $\i[g(H), A_3^a]$. We need to compute
\begin{align*}
  f_1(H) M_aA_1\i
[g(H),A_3^a]A_1M_af_1(H).
\end{align*}
  Modulo a term in $\vO(r^{(-1)_-})$ this
quantity is given by  
\begin{align*}
  &f_1(H) \tilde\xi_a^+M_aA_1\i
[\tilde\xi_a^+g(H)\tilde\xi_a^+,A_3^a]A_1M_a\tilde\xi_a^+f_1(H)\\
&\simeq_1 f_1 (H)M_aA_1 \tilde f_2(H_a)\i
  [g(H_a),A_3^a] \tilde f_2(H_a)A_1M_af_1(H). 
\end{align*} Note that any derivative of $\tilde\xi_a^+$ will
contribute by an error on the form  $\vO(r^{-\infty})$ due to
  \eqref{eq:partM}. Then we apply the computations of Subsection \ref{subsec:Commutation
  with A_3^a}, and by the same argument as above
(i.e.  using that $f_1\prec \tilde f_2$) we
can then remove  appearing  factors of $\tilde f_2(H_a)$.
%Obviously the term $r^{-\kappa}$ in the assertion  \ref{item:H3} is
%harmless producing when added a term of the form $\vO(r^{(-1)_-})$.

{\bf \ref{item:H5}--\ref{item:H8}:} With $A_1$ and $A_3^a$ as above we
consider
\begin{align*}
  \Psi_4=  f_1(H) M_aA_1\chi_+\parb{r^{\rho_2-1}r_\delta^a }A_3^a\chi_+\parb{r^{\rho_2-1}r_\delta^a }A_1M_af_1(H).
\end{align*}  We   compute and integrate  $\tfrac{\d}{\d
  t}\inp{\Psi_4}_t$, using again 
  \ref{item:H1}, \ref{item:H2} and   \eqref{eq:smoothBnd} as well as the
  boundedness of  $\Psi_4$. We got from the above  computation of  $\i[g(H),A_3^a]$ effectively negative
main terms. There is an additional  contribution, more
  precisely from the two appearances of 
$\i[g(H),\chi_+\parb{r^{\rho_2-1}r_\delta^a }]$  and this is also
effectively negative (contributing by two  additional negative terms), cf. the treatment of
  $\i[g(H),A_2^a]$ in  Subsection
  \ref{subsec:A2}. Since the  four non-trivial  contributing terms
  come out with the same sign we conclude
   \ref{item:H5}--\ref{item:H8}.
\end{proof}

\begin{remarks}\label{remark:breveBNDs}
  \begin{enumerate}[i)] \item \label{item:P0} Due to
    \eqref{eq:comMloc} the appearing factor of
    $ \tilde{f}_2(H_a) $ in \ref{item:H3},  \ref{item:H5}  and \ref{item:H7} can  be
    replaced by  $ \tilde{f}_2(H) $ (or by  $ \tilde{f}_2(\breve H_a)
    $, cf. \ref{item:P1} below). (This was actually implicitly used in
    the above proof.) The factor  serves as a regularization indeed
    making $Q^\pm_3(a)$, $Q^\pm_5(a)$  and  $Q^\pm_7(a)$ bounded relatively to  $H$
    and can  (presumably) not  be removed. By the same argument one
    could  make all of the `$Q$-operators' bounded, if preferred. Note
    for example that  in all cases $Q \tilde f_2 (H)=\vO(r^{(\rho_1-\delta)/2})$
 in the precise sense of  \eqref{eq:1712022}.
  \item \label{item:P1}
There are very  similar
  bounds  for $\brH_a$. So denote  $\varphi(t)=\e^{-\i t \brH_a } \varphi$ for
  $\varphi\in \vH$.  Using  \eqref{eq:LAPbndbr} and 
  \eqref{eq:mourre0} we may replace $f_1 (H)\psi(t)$ in
  \ref{item:H1}--\ref{item:H8} by
  $f_1(\brH_a)f_1(\brh_\alpha)\varphi(t)$ and the expression
  $\norm{\psi}^2$ to right of the estimates by $\norm{\varphi}^2$.
  The proof of these 
  modifications of 
  \ref{item:H1}--\ref{item:H8}  are essentially the same. Note that the analogue of \eqref{eq:smoothBnd} reads
\begin{align}\label{eq:smoothBnd22}
  \int^\infty_{-\infty} \,\norm{Q_0f_1(\brH_a)f_1(\brh_\alpha)\varphi(t)}^2\,\d t\leq
  C_\epsilon \norm{\varphi}^2;\quad Q_0=r^{-1/2-\epsilon},\,\epsilon>0.
\end{align}
When modified by replacing the factors of $f_1 (H)$ by factors of
$f_1(\brH_a)f_1(\brh_\alpha)$ the  propagation observables
$\Psi_{1,j},j\leq J,\Psi_2,\Psi_3$ and $ \Psi_4$ may be denoted
 by
$\Phi_{1,j},$ $j\leq J, \Phi_2,\Phi_3$ and $\Phi_4$, respectively (to
be  useful  in the proof of Lemma 
\ref{lemma:strongCont}).
\item \label{item:P2} The structure of the bounds 
  \ref{item:H1}--\ref{item:H8}  and the analogous ones  mentioned in \ref{item:P1} is given as 
\begin{align*}
  &\int^\infty_{-\infty} \,\norm[\big]{Q^\pm f_1 (H)\psi(t)}^2\,\d t\leq
  C \norm{\psi}^2,\\
&\int^\infty_{-\infty} \,\norm[\big]{Q^\pm f_1(\brH_a)f_1(\brh_\alpha)\varphi(t)}^2\,\d t\leq
  C \norm{\varphi}^2,
\end{align*} respectively, where $Q^\pm $ may (or may not)  depend on
$a$ and $j\leq J=J(a)$. In addition we also have such bounds with
$Q=\vO(r^{(-1/2)_-})$,
cf. \eqref{eq:smoothBnd} and \eqref{eq:smoothBnd22}. We  give the
mentioned 
$Q$'s  (including $Q=\vO(r^{(-1/2)_-})$) some uniform index to distinguish them, say
$Q^\pm =Q^\pm(a,k)$ (not specifying $k$). By the Kato smoothness
theory \cite{Ka} a consequence of the above bounds are the following resolvent
bounds, valid for any combination of the indices,
\begin{subequations}
 \begin{align}\label{eq:kato1}
  \quad \sup _{\Im z>0}\norm{(2\i\pi)^{-1}Q^+(b,l)f_1 (H)\parb{R(z)-R(\bar z)}f_1
     (H)Q^-(a,k)^*}< \infty.
\end{align} In particular, recalling \eqref{eq:LAPbnda},
\eqref{eq:LAPbnd} and the surrounding discussion, there exist weak limits as $z\to \lambda\in\R$ from above and these limits define a   weakly continuous
function in
$\lambda$. Whence we can record that the  limiting operators 
\begin{align}\label{eq:kato2}
  \begin{split}
   & \quad Q^+(b,l))f_1 (H)\delta(H-\lambda)f_1 (H)Q^-(a,k)^* \in
   \vL(\vH)\\&  \quad\quad \text{ exist with a weakly continuous
     dependence of }\lambda\in\R.
  \end{split}
\end{align}
\item \label{item:P3} We do not show the existence of
  \begin{align*}
    \wlim _{\epsilon\to 0_+} \,Q^+(b,l))f_1 (H)R(\lambda+ \i \epsilon)f_1 (H)Q^-(a,k)^*, 
  \end{align*} although it is likely  doable by mimicking
  \cite[Section 3]{Ya1}.  Up to Subsection \ref{subsubsec:Besov space setting at stationary scattering
    energies} it is more
  convenient 
  to deal with the delta function only. However in
  Subsection \ref{subsubsec:Besov space setting at stationary scattering
    energies}  we need the resolvent bounds
\begin{align}\label{eq:kato10}
  \quad \sup _{\Im z\neq 0}\norm{Q^\pm(b,l)f_1 (H){R(z)}}_{\vL(\vB,\vH)}< \infty,
\end{align} which follow by combining our commutator
bounds, the Besov space bound (cf. \eqref{eq:BB^*a})
  and the proof
of \cite[Lemma 3.3]{Ya1}.

 \end{subequations}  

\end{enumerate}
  \end{remarks}

 \subsection{Integrability of  $T^\pm_\alpha$ }\label{subsec:Integrability}
We shall show the assertion \eqref{eq:wav2s} in a standard fashion  verifying `integrability of the
   time-derivative'. 

   \begin{lemma}\label{lemma:exist_integr-tpm_-}
     The operators $\Omega^\pm_\alpha$ are well-defined.
   \end{lemma}
   \begin{proof}
   The existence of the time-limits  will be verified by means of  Section
   \ref{sec:Calculus considerations}, \eqref{eq:smoothBnd}, Lemma \ref{lemma:Hbnds} and
   Remarks \ref{remark:breveBNDs} \ref{item:P1},  computing  and estimating the quantity
\begin{align*}
  T^\pm_\alpha&=\i f_1(H) g(H)  M_a N^a_\pm M_a f_1(\brH_a)f_1(\brh_\alpha)-\i f_1(H)  M_a N^a_\pm M_a g(\brH_a) f_1(\brH_a)f_1(\brh_\alpha)\\& =
{f_1(H) \i [g(H),  M_a N^a_\pm M_a] f_1(\brH_a)f_1(\brh_\alpha)}\\&\quad\quad
                                                                                                                                                                   +\i
                                                                                                                                                                   f_1(H)
                                                                                                                                                                   M_a
                                                                                                                                                                   N^a_\pm
                                                                                                                                                                   M_a\tilde{\xi}_a^+\parb{g(H)-g(\brH_a)}f_1(\brH_a)f_1(\brh_\alpha)=: T_1+T_2.
\end{align*} For convenience we treat only the `plus case'.
 The calculation is  modulo  terms on the form
$T=f_1(H)\vO(r^{(-1)_-})f_1(\brH_a)f_1(\brh_\alpha)$. Let us abbreviate this
property as $T\simeq_2 0$ (which is a different relation than the one
used in Subsection \ref{subsec:Propagation estimates}). Any such `error'
can be treated by  
 the smoothness bounds 
\eqref{eq:smoothBnd} and \eqref{eq:smoothBnd22}. 

 We aim at verifying integrability of $T^+_\alpha$  in the precise
sense  stated below  as
\eqref{eq:intCau} ($T^-_\alpha$ can be
treated similarly).
Most of the work is
essentially  already done in the proof of  Lemma
\ref{lemma:Hbnds}. However the 
computation and estimation of the  term $T_2$ to the right is different. We claim
that this term $\simeq_2 0$.  To see this we
note that on the support of the factors $A_2^a$ in the definition of
$N^a_+$
\begin{align}\label{eq:abndindre}
  \inp{x^a}\leq C {r}^{1-\rho_2}.
\end{align} Next we
use a Taylor formula like the one used in \eqref{eq:Ta1}, however used
 differently.
We write and estimate  for any $b\not\leq a$, on the support of
$\tilde{\xi}_a^+$, with \eqref{eq:abndindre}  and for $r$ large
\begin{align*}
  \abs[\big]{V_{\rm lr}^b(x)- V_{\rm lr}^b(x_a)}&=\abs[\Big]{\int_0^1 \pi^bx^a\cdot (\nabla
  V_{\rm lr}^b)\parb{\pi^b(x-tx^a)}\,\d
    t}\\&\leq C_1\inp{x^a}\parb {\abs{x^b}-C_2\abs{x^a}}^{-(1+\mu)}
\\&\leq C_3{r}^{1-\rho_2}\parb {\abs{x}-C_4{r}^{1-\rho_2}}^{-(1+\mu)}
\\&\leq C_5{r}^{-\rho_2-\mu}.
\end{align*} %Of course the final bound is valid for all
% $x\in \supp \tilde{\xi}_a^+$ obeying \eqref{eq:abndindre}. 
 Similarly,  under the same
conditions as above,
\begin{align*}
  \abs{\pi^bx_a}\geq \abs{\pi^bx}-\abs{\pi^bx^a}\geq c_1 r\geq
  c_2\abs{x_a};\quad c_1>c_2>0.
\end{align*}

For the case where $I^{\rm sr}_a=0$ we can use the above bounds and
the fact that $\rho_2+\mu>1$ to
argue, abbreviating $\chi_-\parb{\inp{x^a}
       {r}^{\rho_2-1}/C}=\chi_-(\cdot)$,
\begin{align*}
  T_2&\simeq_2 f_1(H)\vO(r^0)\chi_-(\cdot)\tilde{\xi}_a^+\parb{g(H)-g(\brH_a)}f_1(\brH_a)f_1(\brh_\alpha)\\
&\simeq_2 \int_{\C}
  f_1(H)\vO(r^0)(H-z)^{-1}\chi_-(\cdot)\tilde{\xi}_a^+\parb{I_a-\breve
  I_a}(\brH_a-z)^{-1}f_1(\brH_a)f_1(\brh_\alpha)\,\d \mu_g(z)\\
&=f_1(H)\vO(r^{-\rho_2-\mu})f_1(\brH_a)f_1(\brh_\alpha)\,\simeq_2 0.
\end{align*} If $I^{\rm
  sr}_a\neq 0$ there is an extra term on the form $\vO(r^{-1-\mu})$ in
the above integral, which obviously yields the same conclusion.

As for the term $T_1$ we use  as in Section \ref{subsec:Propagation
  estimates} the functions $ f_1\prec \tilde f_2\prec\tilde{f_1}$. We may write $T_1$ as
\begin{align}\label{eq:comuse}
  T_1\simeq_2 f_1(H) \tilde f_2(H)\i [g(H),  M_a N^a_\pm M_a] \tilde f_2(H)f_1(\brH_a)f_1(\brh_\alpha).
\end{align}

Applying   the  product rule for commutation  several
terms arise, 
essentially all treated in Section
\ref{sec:Calculus considerations} (see also the proof of Lemma \ref{lemma:Hbnds}). Let us show how  to treat the contribution from the factors
of $\i [g(H), A_1]$. (The contribution from the factors
of $\i [g(H), M_a]$ may be  treated similarly, cf. the proof of  Lemma
\ref{lemma:Hbnds}.)
  
 By \eqref{eq:basicFB} and Remark \ref{remark:comm-with-a_3yy}  we can write it as 
\begin{align*}
 &2f_1(H) \tilde f_2(H) M_a\Re\parbb{\i [g(H),
  A_1]A_2^aA_3^aA_2^aA_1}M_a\tilde f_2(\brH_a)f_1(\brH_a)f_1(\brh_\alpha)\\
&\simeq_2 f_1(H)
   (Q^+_2)^*\vO(r^{-0})Q^+_2f_1(\brH_a)f_1(\brh_\alpha);\quad Q^+_2=r^{-1/2}\sqrt{(\chi^2_+)'}\parbb{\pm B/\epsilon_0}.
\end{align*} 
Next we apply \cas estimating
\begin{align*}
  &\int^\infty_{-\infty} \abs[\big]{\inp{f_1(H)\psi(t),
    (Q_2^+)^*\vO(r^{-0})Q_2^+f_1(\brH_a)f_1(\brh_\alpha)\varphi(t)}}\,\d t\\
&\leq \int^\infty_{-\infty}
  \,\norm{Q_2^+{f_1(H)\psi(t)}}\,\norm{Q_2^+f_1(\brH_a)f_1(\brh_\alpha)\varphi(t)}\,\d t\\
&\leq \parbb{\int^\infty_{-\infty} \,\norm{Q_2^+{f_1(H)\psi(t)}}^2\,\d t}^{1/2}\,\parbb{\int^\infty_{-\infty} \,\norm{Q_2^+f_1(\brH_a)f_1(\brh_\alpha)\varphi(t)}^2\,\d t}^{1/2}.
\end{align*} The first factor is finite due to Lemma \ref{lemma:Hbnds}
\ref{item:H2}, in turn  with a bound proportional to 
$\norm{\psi}$. Similarly, thanks to   Remark \ref{remark:breveBNDs}
\ref{item:P1}, the
second factor is finite too (although not needed with a bound proportional with
$\norm{\varphi}$).

We may argue similarly for the other commutators arising from expanding
the commutator in  \eqref{eq:comuse}. More precisely we  first compute
\begin{align*}
  &f_1(H) M_aA_1\i
    [g(H),A_2^aA_3^aA_2^a]A_1M_af_1(\brH_a)f_1(\brh_\alpha)\\&\quad\simeq_2 f_1(H) M_aA_1\i
    \Big [g(H),\parbb {A_3^a-\chi_+\parb{r^{\rho_2-1}r_\delta^a
  }A_3^a\chi_+\parb{r^{\rho_2-1}r_\delta^a }}\Big ]A_1M_af_1(\brH_a)f_1(\brh_\alpha)\\&\quad
\simeq_2 f_1(H) \parbb{-(Q_{3}^+)^2-Q_{4}^+)^2+(Q_{5}^+)^2+Q_{6}^+)^2 +(Q_{7}^+)^2+Q_{8}^+)^2} f_1(\brH_a)f_1(\brh_\alpha).
\end{align*} We integrate, estimate by  \cas and  invoke  Remarks \ref{remark:breveBNDs}
\ref{item:P1} and \ref{item:P2}.

These estimates  lead 
 to the conclusion that 
\begin{align*}  \int^\infty_{-\infty} \abs[\big]{\inp{\psi(t),
   T^+_\alpha\varphi(t)}}\,\d t\leq C_\varphi\norm{\psi},
\end{align*} and by the same arguments,  that 
\begin{align}\label{eq:intCau}
  \forall \epsilon>0\,\,\exists t_\epsilon=t_\epsilon(\varphi)>0:\quad \int^\infty_{t_\epsilon} \abs[\big]{\inp{\psi(t),
   T^+_\alpha\varphi(t)}}\,\d t\leq \epsilon \norm{\psi}.
\end{align} 

Clearly the existence of
   $\Omega^+_\alpha$ follows from \eqref{eq:intCau}. The existence of
   $\Omega^-_\alpha$ can be shown similarly.   
   \end{proof}

\section{Formula for $\widetilde
  S_{\beta\alpha}(\lambda)$}\label{sec:Formula for widetilde}

Recall from Section \ref{sec:Stationary modifier},
\begin{align}\label{eq:formwOp}
  \widetilde W^\pm_\alpha= \Omega^\pm_\alpha J_\alpha \breve
  w_a^{\pm}f_2(k_\alpha)\mand \widetilde S_{\beta\alpha}=\parb{\Omega^+_\beta J_\beta \breve
  w_b^{+}f_2(k_\beta}^*\Omega^-_\alpha J_\alpha \breve
  w_a^{-}f_2(k_\alpha).
\end{align}

We show  in Appendix \ref{sec:AppendixB} the following representation
of  $\widetilde
  S_{\beta\alpha}(\lambda)$, formally given as
\begin{align}\label{eq:SreP}
  \widetilde
  S_{\beta\alpha}(\lambda)=(2\pi \i)^2f^2_2(\lambda){\gamma_b^+(\lambda_\beta)}J^*_\beta
  \parb{T^+_\beta}^* \delta(H-\lambda)T^-_\alpha
    J_\alpha{\gamma_a^-(\lambda_\alpha)}^*;\quad \lambda_\alpha:=
  \lambda-\lambda^\alpha.  
\end{align}  

Parallel to Section \ref{sec:one-body  matrices} we call the appearing
 $\lambda$-depending operators
$J_\beta{\gamma_b^+(\lambda_\beta)}^*$ and $J_\alpha{\gamma_a^-(\lambda_\alpha)}^*$ \emph{future
 and past auxiliary 
  channel wave matrices}, respectively. In Subsections \ref{subsec:Trace estimates} and \ref{subsec:Conclusion of argument, the
              weak continuity} we examine the well-definedness  of
            the right-hand side of \eqref{eq:SreP} upon substituting
            expressions for $T^+_\beta$ and $T^-_\alpha$ 
            derived in the proof of Lemma
            \ref{lemma:exist_integr-tpm_-}. With the resulting  final
            representation we then derive the weak
            continuity of $\widetilde S_{\beta\alpha}(\cdot)$  from
            which Theorem \ref{thm:mainstat-modif-n} follows.

\subsection{Phase space estimates of  auxiliary channel wave matrices}\label{subsec:Trace estimates}
 We aim at   bounding  the operators (recall the index convention of
 Remark \ref{remark:breveBNDs} \ref{item:P2})
\begin{align*}
    Q^\pm(a,k)f_1(\brH_a)f_1(\brh_\alpha)J_\alpha{\gamma_a^\pm(\lambda_\alpha)}^*=
  f^2_1(\lambda)Q^\pm(a,k)J_\alpha{\gamma_a^\pm(\lambda_\alpha)}^*\in\vL\parb{L^2(C_a),
  \vH}
\end{align*} as well as showing  a continuity property. This  will
along with \eqref{eq:kato2} be
crucial  for our usage of \eqref{eq:SreP}.

\begin{lemma}\label{lemma:strongCont}
   For all  $a\in \vA'$ and indices $k$ 
\begin{align}\label{eq:Tbnd}
    \sup_{\lambda\in\Lambda}\,\,\norm{Q^\pm(a,k)f_1(\brH_a)f_1(\brh_\alpha)J_\alpha{\gamma_a^\pm(\lambda_\alpha)}^*}<\infty,
\end{align}
    and  the operator-valued functions  
\begin{align*}
 Q^\pm(a,k)f_1(\brH_a)f_1(\brh_\alpha)J_\alpha{\gamma_a^\pm(\lambda_\alpha)}^*\in\vL\parb{L^2(C_a),
  \vH}
\end{align*} are strongly continuous in $\lambda\in
\Lambda$. 
\end{lemma}
\begin{proof} It suffices     to consider   the cases $Q^\pm(a,k)\neq Q_0$ since
  for 
  $ Q_0$ (as given by \eqref{eq:smoothBnd}) the assertions  reduce to  properties of
  ${\gamma_a^\pm(\lambda_\alpha)}^*$ stated in Section
  \ref{sec:one-body  matrices}. Whence from this point we assume
  $Q^\pm(a,k)\neq Q_0$. 
The  uniform
  bounds \eqref{eq:Tbnd} will be proved in Step I below. The stated
  strong continuity assertions will then be shown in  Steps
  II--V. Note that  weak continuity is an immediate consequence of
  \eqref{eq:Tbnd} and the strong continuity for  $Q^\pm(a,k)= Q_0$ (to
  be used in Steps IV and V).

{\bf I.} By `the $T^*T$ argument' is suffices for  \eqref{eq:Tbnd}  to show that 
\begin{align*}
    \sup_{\lambda\in\Lambda}\,\,\norm{{\gamma_a^\pm(\lambda_\alpha)}J_\alpha^*f_1(\brh_\alpha)f_1(\brH_a)Q^\pm(a,k)^*Q^\pm(a,k)f_1(\brH_a)f_1(\brh_\alpha)J_\alpha{\gamma_a^\pm(\lambda_\alpha)}^*}<\infty.
\end{align*} We pick, using quantities from Remark
\ref{remark:breveBNDs} \ref{item:P1}, 
\begin{align}\label{eq:probndbbs}
  \begin{split}
  &\Phi\in\spann_{\R}
\set{\Phi_{1,j},\Phi_2,\Phi_3,\Phi_4\,|\,j\leq J}:\\&\quad
  \i[\brH_a,\Phi]\geq f_1(\brh_\alpha)f_1(\brH_a)Q^\pm(a,k)^*Q^\pm(a,k)f_1(\brH_a)f_1(\brh_\alpha)\\&\quad\quad+f_1(\brh_\alpha)f_1(\brH_a)\vO(r^{(-1)_-})f_1(\brH_a)f_1(\brh_\alpha).  
  \end{split}
\end{align}
For $\rho>1$ we let $\chi_\rho=\chi_-(r/\rho)$ and compute
using   the notation $\inp{T}_\phi=\inp{\phi,T\phi}$ with 
$\phi:=J_\alpha{\gamma_a^\pm(\lambda_\alpha)}^*g$, 
 $
  g\in L^2(C_a)$,
\begin{align*}
  0&=\inp{\i [\brH_a-\lambda,\chi_\rho\Phi\chi_\rho]}_\phi\\
&=\inp{\chi_\rho\i [\brH_a,\Phi]\chi_\rho}_\phi+2\Re{\inp{\i
  [\brH_a,\chi_\rho]\Phi \chi_\rho}_\phi}.
\end{align*}
For the last term to the right we compute and insert the expression
\begin{align*}
  \i
  [\brH_a,\chi_\rho]=-\rho^{-1}{\sqrt{-\chi'_-}(r/\rho)B\sqrt{-\chi'_-}(r/\rho)}, 
\end{align*} commute and  invoke the dual versions of
\eqref{eq:genFou2} allowing us to  conclude that this term is
bounded by $C\norm{g}^2$ with  a constant $C>0$ being  independent of
$\rho>1$. For the first term we use \eqref{eq:probndbbs},
\eqref{eq:genFou2} and  a commutation,  yielding the bound
\begin{align}\label{eq:Tbndmm}
    \sup_{\lambda\in\Lambda, \rho>1}\,\,\norm{\chi_\rho Q^\pm(a,k)f_1(\brH_a)f_1(\brh_\alpha)J_\alpha{\gamma_a^\pm(\lambda_\alpha)}^*}<\infty.
\end{align} 

Clearly 
\eqref{eq:Tbnd} follows from  \eqref{eq:Tbndmm}  and 
 Lebesgue's monotone convergence theorem (by taking $\rho\to \infty$).

{\bf II.}  For
  $Q_1(a,j)=\xi^+_jG_{d_j}$ (as  in   Lemma
  \ref{lemma:Hbnds} \ref{item:H1}) the strong continuity assertion follows from
  \eqref{eq:perpOne2} if  $ d_j= a$ due to the representation 
  \begin{align*}
    \xi^+_jG_{d_j}f_1(\brH_a)f_1(\brh_\alpha)J_\alpha{\gamma_a^\pm(\lambda_\alpha)}^*= f^2_1(\lambda)\xi^+_jJ_\alpha
    G_{d_j}
   {\gamma_a^\pm(\lambda_\alpha)}^*.
  \end{align*}

 For   $ d_j\neq a$ we have $ d_j \lneq a$ (since $ d_j\leq a$). The fact
that the factor
$\xi^+_j$ is supported in $Y_{d_j}(\delta_j)$ then yields  the estimate
$\abs{x^a}\geq \epsilon\abs{x}$ on $\supp \xi^+_j$ for some
$\epsilon>0$.
We decompose  for
 (large) $\rho>1$
\begin{align*}
  \xi^+_jG_{d_j}&=\chi_1\xi^+_jG_{d_j}+\chi_2\xi^+_jG_{d_j},\\&\chi_1=\chi^2_+\parb{\abs{x^a}/\rho},\quad \chi_2=\chi^2_-\parb{\abs{x^a}/\rho},
\end{align*} and write correspondingly
\begin{align*}
  \xi^+_jG_{d_j} J_\alpha{\gamma_a^\pm(\lambda_\alpha)}^*g
=\chi_1\xi^+_jG_{d_j}J_\alpha{\gamma_a^\pm(\lambda_\alpha)}^*g+\chi_2\xi^+_jG_{d_j}J_\alpha{\gamma_a^\pm(\lambda_\alpha)}^*g;\quad
  g\in\vL\parb{L^2(C_a)}.
\end{align*}

For the first term we estimate using \eqref{eq:genFou2}, cf. the proof of
\cite[Lemma 4.5]{Ya1},
\begin{align*}
  \norm{\chi_1\xi^+_jG_{d_j}&J_\alpha{\gamma_a^\pm(\lambda_\alpha)}^*g}^2\leq
  C_1\int_{\abs{x^a}\geq \rho}\d
  x^a \parb{\abs{u^\alpha}^2+\abs{p^au^\alpha}^2}\\&\abs{x^a}^{-1}\int_{\abs{x_a}\leq
  \epsilon^{-1}\abs{x^a}} \parb{\abs{\gamma_a^\pm(\lambda_\alpha)^*g}^2+\abs{p_a{\gamma_a^\pm(\lambda_\alpha)}^*g}^2}\,\d
                                                     x_a\\
&\leq
  C_2\int_{\abs{x^a}\geq \rho}\,\parb{\abs{u^\alpha}^2+\abs{p^au^\alpha}^2}\,\d
  x^a ;
\end{align*} here the constant $C_2$ is independent of  $\lambda\in
\Lambda$. Consequently the right-hand side can be taken arbitrary
  small uniformly in $\lambda$ by taken $\rho>1$ big enough (now fixed).  

For the the second term  note that  $\abs{x}\leq
  2\epsilon^{-1}\rho$ on the   support of $\chi_2$. Hence this term is continuous in
  $\lambda\in \Lambda$ 
   thanks to the established continuity for $Q^\pm(a,k)= Q_0$. 
     Whence we have treated  the case
  $Q_1(a,j)=\xi^+_jG_{d_j}$ for  any $j \leq J$.

{\bf III.} For  $Q_2^+$ (as  in   Lemma
  \ref{lemma:Hbnds} \ref{item:H2}, leaving out a similar discussion
  for $Q_2^-$)  it suffices by  \eqref{eq:Tbnd} to show continuity in
  $\lambda$ of
  \begin{align*}
  Q_2^+f_1(\brH_a)f_1(\brh_\alpha)J_\alpha{\gamma_a^+(\lambda_\alpha)}^*g;\quad
      g \in C^\infty (C_a).
\end{align*}  For $\lambda$ close to a fixed $\lambda'\in \Lambda$ we have
correspondingly that $\lambda_\alpha=\lambda-\lambda^\alpha$ is close to
$\lambda'_\alpha=\lambda'-\lambda^\alpha$, and we can to some extent
use  that for $s>1/2$
\begin{align*}
  \norm{\inp{x_a}^{-s}\phi(\lambda)}\to 0\text{ for }\lambda\to \lambda';\quad \phi(\lambda):={\gamma_a^+(\lambda_\alpha)}^*g-{\gamma_a^+(\lambda'_\alpha)}^*g.
\end{align*}
By  arguments from Step II, in particular    it suffices to show that   for some 
  $\epsilon>0$ 
  \begin{align*}
  \norm{Q_2^+\chi^2_-\parb{\abs{x^a}/(\epsilon
    r)}f_1(\brH_a)f_1(\brh_\alpha)
    J_\alpha\phi(\lambda)}\to 0
\text{ for }\lambda\to \lambda'.
      \end{align*} In turn, possibly seen  by using one-body versions of Lemmas \ref{lemma:Sommerfeld} and
      \ref{lemma:poiss-oper-geom} and \cite[Theorem 1.8]{AIIS}
      (more precisely by using resolvent bounds), it suffices to
show that  $\norm
{\tilde{\varphi}(\lambda)}=o(\abs{\lambda-\lambda'}^0)$, where for 
$\epsilon, \epsilon_0>0$ taken  small
\begin{align*}
  \tilde{\varphi}(\lambda)&:=\tilde{\varphi}_+(\lambda)+\tilde{\varphi}_-(\lambda),\\
\tilde {\varphi}_\pm(\lambda)&:=
\tilde f_1(\brH_a)\sqrt{(\chi^2_+)'}(
    B/\epsilon_0)\tilde f_1(\brH_a)r^{-1/2}T^\pm_\epsilon
  f_1(\brH_a)f_1(\brh_\alpha)J_\alpha\phi(\lambda),\\\quad  T^\pm_\epsilon&:=\chi^2_-\parb{\abs{x^a}/(\epsilon
    r)}
  \chi_+\parb{\pm{\Re{\parb{(x_a/r)\cdot
  p_a}}}/(4\epsilon_0)}.
\end{align*} (Recall from the discussion at the end of Subsecton
\ref{subsec:Limiting absorption principles} that we allow ourselves
the freedom of considering only sufficiently small values of $\epsilon_0$.)
 We can represent 
 \begin{align}\label{eq:Bdecomp}
  \begin{split}
  B&=2\Re{\parb{(x/r)\cdot p+\vO(r^{-1})\cdot
  p}}\\&=2\Re \parb{{(x_a/r)\cdot p_a}+{(x^a/r)\cdot p^a}+{\vO(r^{-1})\cdot
  p}}+\vO(r^{-1}),  
  \end{split}
\end{align} and by using this 
 and   commutation (cf.  Remark 
\ref{remark:comm-with-a_3yy}),   we estimate (with $\inp{T}_\varphi=\inp{\varphi,T\varphi}$)
\begin{align}\label{eq:commsmall}
  \begin{split}
  &\tfrac12 \norm {\tilde{\varphi}_\pm(\lambda)}^2\\
&\leq \inp{I\mp {B}/(4\epsilon_0)}_{\tilde{\varphi}_\pm(\lambda)}+o(\abs{\lambda-\lambda'}^0)\\
&\leq \inp{(1+C\epsilon) I\mp\Re{\parb{(x_a/r)\cdot
  p_a}/(2\epsilon_0)}}_{\tilde{\varphi}_\pm(\lambda)}+o(\abs{\lambda-\lambda'}^0)\\&\leq
                                                                             \inp{(C\epsilon-1)
                                                                             I}_{\tilde{\varphi}_\pm(\lambda)}+o(\abs{\lambda-\lambda'}^0)
\\&\leq o(\abs{\lambda-\lambda'}^0).  
  \end{split}
\end{align}
{\bf IV.} As for $Q_3^\pm(a)$ and $ Q_4^\pm(a)$ we introduce
$Q^\pm:=\sqrt{\abs{Q_3^\pm(a)}^2+\abs{Q_4^\pm(a)}^2}$. It  suffices (thanks to
\eqref{eq:Tbnd}  and the
weak continuity) to
show that
\begin{align}
  \label{eq:weakStrong}
  \begin{split}
&\norm{Q^\pm
  f_1(\brH_a)f_1(\brh_\alpha)J_\alpha{\gamma_a^\pm(\lambda_\alpha)}^*g}^2\\&\quad\to 
\norm{Q^\pm
  f_1(\brH_a)f_1(\brh_\alpha)J_\alpha{\gamma_a^\pm(\lambda'_\alpha)}^*g}^2\text{
                                                                           for
                                                                           }\lambda\to\lambda'
;\\&\quad \quad \quad\quad
  g\in C^\infty(C_a).   
  \end{split}
\end{align}

Letting  $\chi_\rho=\chi_-(\abs{x_a}/\rho)$ for $\rho>1$ and $\phi^\pm(\lambda)=J_\alpha{\gamma_a^\pm(\lambda_\alpha)}^*g$,
 we compute as
in Step I 
(now with $\Phi^\pm_3=\Phi^\pm_3(a)$) 
\begin{align*}
  0&=\inp{\i [\brH_a,\chi_\rho\Phi^\pm_3\chi_\rho]}_{\phi^\pm(\lambda)}\\
&=\inp{\i [\brH_a,\Phi^\pm_3]}_{\chi_\rho\phi^\pm(\lambda)}+2\Re{\inp{\i [\brH_a,\chi_\rho]\Phi^\pm_3\chi_\rho}_{\phi^\pm(\lambda)}}.
\end{align*} We take the limit $\rho\to \infty$. For the first term we
get in the limit $\mp\norm{f^2_1(\lambda)Q^\pm\phi^\pm(\lambda)}^2$ plus a term that we 
 know is  continuous in $\lambda$ by Steps II and III. Here we used
 commutation to replace $\norm{f^2_1(\lambda)Q^\pm\chi_\rho\phi^\pm(\lambda)}^2$ by
 \begin{align*}
   \inp{Q_3^\pm(a)^*\chi^2_\rho Q_3^\pm(a)}_{f_1(\brH_a)f_1(\brh_\alpha)\phi^\pm(\lambda)}+\inp{Q_4^\pm(a)^*\chi^2_\rho Q_4^\pm(a)}_{f_1(\brH_a)f_1(\brh_\alpha)\phi^\pm(\lambda)}
 \end{align*} before taking the limit.

 As for the second term we compute  as follows,   using again 
 conveniently 
  one-body versions of  Lemma 
      \ref{lemma:poiss-oper-geom} and \cite[Theorem 1.8]{AIIS},  and
      using in
      the second step also  \ref{item:1a}
and \ref{item:2a} from Section \ref{sec:Yafaev's construction}. (The
computation overlaps partially one in Appendix \ref{sec:AppendixO}.) A
further elaboration is needed, it is given after the computation.  We compute for any 
$\epsilon,\epsilon_0>0$ taken small
\begin{align}\label{eq:cruci}
 &\lim_{\rho\to \infty}2\Re\inp{\i
    [\brH_a,\chi_\rho]\Phi^\pm_3\chi_\rho}_{\phi^\pm(\lambda)}\\\nonumber
&=\pm 2\sqrt{\lambda_\alpha}f^4_1(\lambda)\lim_{\rho\to
  \infty}\rho^{-1}\Re\inp [\big]{(\chi^2_-)'\parb{\tfrac{\abs{x_a}}\rho}\chi_-\parb{\tfrac{\abs{x^a}}{\epsilon
    r}}M_a^2}_{\chi_+\parb{\pm{\Re{\parb{(x_a/r)\cdot
  p_a}}}/(4\epsilon_0)}\phi^\pm(\lambda)}\\\nonumber
&=\pm \parbb{2\sqrt{\lambda_\alpha}}^3f^4_1(\lambda)\lim_{\rho\to
  \infty}\rho^{-1}\inp [\big]{(\chi^2_-)'\parb{\tfrac{\abs{x_a}}\rho}}_{m_a(\hat x_a)\chi_+\parb{\pm{\Re{\parb{(x_a/r)\cdot
  p_a}}}/(4\epsilon_0)}\phi^\pm(\lambda)}\\\nonumber
&=\mp\tfrac {2\lambda_\alpha}\pi
f^4_1(\lambda)\norm{{m_a}g}_{L^2(C_a)}^2.
\end{align}

 Note that in the first step of \eqref{eq:cruci} we used  that
 effectively the factor $A_3^a=A_{3\pm}^a$ of $\Phi^\pm_3$ can be removed. This is 
essentially due to the fact that 
 $B_{\delta,\rho_1}^a=0$ in $\set{\abs{x^a}<c r^\delta}$,
 cf. \ref{item:2c} of Section \ref{sec:Derezinski's
  construction}, since for the complement we have 
 \begin{align*}
2\lim_{\rho\to
  \infty} \rho^{-1}\Re\inp [\big]{p_a\cdot \hat x_a(\chi^2_-)'\parb{\tfrac{\abs{x_a}}\rho}\chi_-\parb{\tfrac{\abs{x^a}}{\epsilon
    r}}\chi^2_+\parb{\tfrac{2\abs{x^a}}{c r^\delta}}\Phi^\pm_3}_{\phi^\pm(\lambda)} 
   =0.
 \end{align*} In fact (using commutation and   Remark 
\ref{remark:comm-with-a_3yy})
 \begin{align*}
&2\lim_{\rho\to
  \infty} \rho^{-1}\Re\inp [\big]{p_a\cdot \hat x_a(\chi^2_-)'\parb{\tfrac{\abs{x_a}}\rho}\chi_-\parb{\tfrac{\abs{x^a}}{\epsilon
    r}}\chi^2_-\parb{\tfrac{2\abs{x^a}}{c r^\delta}}\Phi^\pm_3}_{\phi^\pm(\lambda)} 
   \\&=2 f^4_1(\lambda)\lim_{\rho\to
  \infty} \rho^{-1}\Re\inp [\big]{p_a\cdot \hat x_a(\chi^2_-)'\parb{\tfrac{\abs{x_a}}\rho}\chi_-\parb{\tfrac{\abs{x^a}}{\epsilon
    r}}\chi^2_-\parb{\tfrac{2\abs{x^a}}{c r^\delta}}A_3^aM_a\tilde f_2(\brH_a)^2A_1^2M_a}_{\phi^\pm(\lambda)},
 \end{align*} and
 \begin{align}\label{eq:aligeen}
  \chi^2_-\parb{\tfrac{2\abs{x^a}}{c r^\delta}}\parb{A_3^a-I}= \pm 
\int _{\C}
\chi^2_-\parb{\tfrac{2\abs{x^a}}{c r^\delta}}B_{\delta,\rho_1}^a(\pm B_{\delta,\rho_1}^a-z)^{-1}\,z^{-1}\mathrm
d\mu_{\chi_-}(z)=0.
 \end{align} Once  $A_3^a$ is removed we can also remove the remaining factor
 $\chi^2_-\parb{\tfrac{2\abs{x^a}}{c r^\delta}}$. Similarly  in the first step we also replaced   the factors of
 $ A_1= A_{1\pm}$ by   factors of $\chi_+\parb{\pm{\Re{\parb{(x_a/r)\cdot
  p_a}}}/(4\epsilon_0)}$. Using the one-body version of \cite[Theorem
1.8]{AIIS} as in Step III this is
justified by the assertions 
\begin{align*}
  &2f^4_1(\lambda)\lim_{\rho\to
  \infty} \rho^{-1}\Re\inp [\big]{p_a\cdot \hat x_a(\chi^2_-)'\parb{\tfrac{\abs{x_a}}\rho}\chi_-\parb{\tfrac{\abs{x^a}}{\epsilon
    r}} M_a\tilde f_2(\brH_a)\parb{A_1^2-I}\tilde f_2(\brH_a)M_a}_{ \phi_1^\pm}=0,\\
 &2 f^4_1(\lambda)\lim_{\rho\to
  \infty} \rho^{-1}\Re\inp [\big]{p_a\cdot \hat x_a(\chi^2_-)'\parb{\tfrac{\abs{x_a}}\rho}\chi_-\parb{\tfrac{\abs{x^a}}{\epsilon
    r}}M_a\tilde f_2(\brH_a)A_1^2\tilde f_2(\brH_a)M_a}_{\phi_2^\pm}=0;\\&\quad \quad \phi_1^\pm=\chi_+\parb{\pm{\Re{\parb{(x_a/r)\cdot
  p_a}}}/(4\epsilon_0)}\phi^\pm(\lambda),\\&\quad\quad  \phi_2^\pm=\chi_+\parb{\mp{\Re{\parb{(x_a/r)\cdot
  p_a}}}/(4\epsilon_0)}\phi^\pm(\lambda).
\end{align*} In turn these assertions may be verified  essentially  as in
Step III (more precisely  by commutation along the lines of  \eqref{eq:commsmall})  by showing
\begin{subequations}
\begin{align}\label{eq:good}
  \begin{split}
   &\tilde f_2(\brH_a) \chi_-\parb{\tfrac{\abs{x^a}}{\epsilon
    r}}\chi_+\parb{\pm{\Re{\parb{(x_a/r)\cdot
  p_a}}}/(4\epsilon_0)}^2\parb{A_{1\pm}^2-I}\tilde f_2(\brH_a)\\&\quad
                                                                \quad
                                                                \quad\quad
                                                                \quad
                                                                \quad
                                                               =
                                                                \vO(r^{-1/2})=\vO(r^{0_-}), 
  \end{split}
\end{align} and 
\begin{align}\label{eq:lessgood}\begin{split}
 &\tilde f_2(\brH_a) \chi_-\parb{\tfrac{\abs{x^a}}{\epsilon
    r}}\chi_+\parb{\mp{\Re{\parb{(x_a/r)\cdot
  p_a}}}/(4\epsilon_0)}^2A^2_{1\pm}\tilde f_2(\brH_a)\\&\quad \quad \quad\quad \quad \quad = \vO(r^{-1/2})=\vO(r^{0_-}). \end{split}
\end{align} 
  
\end{subequations}

Obviously the right-hand side of \eqref{eq:cruci}, i.e.  $\mp\tfrac
{2\lambda_\alpha}\pi f^4_1(\lambda)\norm{{m_a}g}_{L^2(C_a)}^2$, is continuous in
$\lambda$. We conclude that the function 
\begin{align*}
  \lambda\to \norm{Q^\pm f_1(\brH_a)f_1(\brh_\alpha)\phi^\pm 
(\lambda)}^2
\end{align*}
  is  continuous,  as wanted.

{\bf V.} We  mimic Step IV using  now 
\begin{align*}
  Q^\pm:=\sqrt{\abs{Q_5^\pm(a)}^2+\abs{Q_6^\pm(a)}^2+\abs{Q_7^\pm(a)}^2+\abs{Q_8^\pm(a)}^2},
\end{align*}   $\Phi^\pm_4=\Phi^\pm_4(a)$ (replacing $\Phi^\pm_3(a)$) 
and using 
again Steps II and III,
yielding similarly  the strong continuity for $Q_5^\pm(a),
Q_6^\pm(a),Q_7^\pm(a)$ and $ Q_8^\pm(a)$. In fact this  case is simpler
in that one readily shows the following analogue of \eqref{eq:cruci},
\begin{align*}
 \lim_{\rho\to \infty}2\Re\inp{\i
    [\brH_a,\chi_\rho]\Phi^\pm_4\chi_\rho}_{\phi^\pm(\lambda)}=0,
\end{align*} which obviously is continuous in $\lambda$. \end{proof}

\subsection{Conclusion of argument, the weak
              continuity}\label{subsec:Conclusion of argument, the
              weak continuity}
 The  expression to the right in \eqref{eq:SreP}  can (and should)  be written, cf. the proof of
              Lemma \ref{lemma:exist_integr-tpm_-},  as a sum of terms on
              the form
              \begin{align*}
                 \parbb{f^2_2(\lambda)&{\gamma_b^+(\lambda_\beta)}J^*_\beta
                Q^+(b,l)^*B^*_+}\\
&\parbb{Q^+(b,l)f_1 (H)\delta(H-\lambda)f_1 (H)Q^-(a,k)^*}\parbb{B_-Q^-(a,k)J_\alpha{\gamma_a^-(\lambda_\alpha)}^*},
              \end{align*} where $B_+=B_+(b,l)$ and $B_-=B_-(a,k)$ are bounded
              operators. The middle factor is weakly continuous by
              Remark \ref{remark:breveBNDs} \ref{item:P2}. The factor to the right is
              strongly continuous by Lemma \ref{lemma:strongCont}. The
              adjoint of the   factor to the left is
              strongly continuous by the same result. Consequently the
              entire product is weakly continuous, and  Theorem
              \ref{thm:mainstat-modif-n} follows (cf. the
              discussion given before the statement). \qed

\section{Channel wave matrices and scattering at fixed
  energy}\label{sec:Exact channel wave-matrices}
In this concuding section we derive  various consequences of our proof
of Theorem \ref{thm:mainstat-modif-n}. Whence   we concretely
construct the open channel wave
  matrices and the scattering matrix for  all energies in the   small
  neighbourhood $I_0 \ni \lambda_0$ fixed  in Section
  \ref{sec:Stationary modifier}. Away
  from a null set of
 energies, more precisely at any stationary scattering
  energy in $I_0$, the scattering matrix  is unitary and strongly
  continuous and it is characterized by asymptotic properties of minimum
  generalized eigenfunctions.  The discussion
relies heavily on Appendix \ref{sec:AppendixB}. 

  For all $b\in \vA'$ we pick an arbitrary  increasing sequence
  $(C_{b,k})_k $ of open reflection symmetric subsets of  $C_b'$ with closure
    $\overline C_{b,k}\subset C_b'$ and $\cup_k C_{b,k}=C_b'$. We pick for each such set a
    function $m_b=m_{b,k}$  as in  Section \ref{sec:Yafaev's construction} such
    that $m_b(\hat \xi)=1$ for all $\hat \xi\in C_{b,k}$.   Let for any channel  $\beta
    =(b,\lambda^\beta, u^\beta)$
    \begin{align*}
      \chi_{\beta,k}&=\chi_{\beta,k}(p_b)=F_\beta^{-1}
    1_{C_{b,k}}(\hat \xi_b)F_\beta,\\
\tilde\chi_{\beta,k}&=\tilde\chi_{\beta,k}(p_b)=F_\beta^{-1}(4\lambda_\beta)^{-1}
    1_{C_{b,k}}(\hat \xi_b)F_\beta.
\end{align*}
    
\begin{subequations}
 By combining \eqref{eq:wave_opc2}, \eqref{eq:conjF}  and
 \eqref{eq:Wfinal} we then obtain  (by an elementary approximation argument)
\begin{align}\label{eq:Wfinal2}
  \begin{split}
    &W^+_\beta
  (f_2g)(k_\beta)\chi_{\beta,k}\varphi
=\widetilde W^+_\beta
  g(k_\beta)\tilde\chi_{\beta,k}\varphi\\
    &=\int \,\tfrac \pi {2\lambda_\beta} (f_2g)(\lambda) \parb{f_1(H)\delta(H-\lambda) T^+_{\beta,k} J_\beta 
      \gamma_{b}^+(\lambda_\beta)^*}1_{C_{b,k}}\gamma_{b,0}(\lambda_\beta)\varphi
    \,\d \lambda;\\
&\quad \quad\quad\text{ for  any  }\varphi\in L^2_s(\mathbf  X_b) \text{ with   }s>1/2. 
  \end{split}
              \end{align} Note that indeed $
              T^+_{\beta}=T^+_{\beta,k}$ depends on $k$. 

              We are lead to introduce the \emph{approximate  channel wave matrices}
\begin{align}\label{eq:AppExa}
 \Gamma^+_{\beta,k}(\lambda)^*=\tfrac \pi {2\lambda_\beta}  f_1(H){\delta(H-\lambda) T^+_{\beta,k} J_\beta 
      \gamma_{b}^+(\lambda_\beta)^*}1_{C_{b,k}};\quad k\in
  \N,\,\lambda \in  I_0\,(\ni \lambda_0).
\end{align} 
 
\end{subequations}

\subsection{Channel wave matrices }\label{subsec: Channel wave matrices }

By \eqref{eq:Wfinal2} the right-hand side of \eqref{eq:AppExa} 
 is independent of  details of the
construction of $T^+_{\beta,k}$. Considered as an $\vL(L^2(C_b), L_{-s}^2(\bX))$-valued function of   $\lambda\in  I_0$ 
for
any $s>1/2$, it is strongly continuous. Clearly 
\begin{align*}
  1_{C_{b,k}}\Gamma^+_{\beta,k+1}(\lambda)=\Gamma^+_{\beta,k}(\lambda)\in \vL( L_{s}^2(\bX),L^2(C_b)).
\end{align*} We can then  for any  $\lambda\in  I_0$  and $\psi\in
L_{s}^2(\bX)$ with  $s>1/2$  introduce   a measurable function
$\Gamma^+_{\beta}(\lambda)\psi$  on $C_b$ by 
\begin{align}\label{eq:GAM}
  1_{C_{b,k}}\Gamma^+_{\beta}(\lambda)\psi=\Gamma^+_{\beta,k}(\lambda)\psi;\quad
  k\in \N.
\end{align}
We argue that $\Gamma^+_{\beta}(\lambda)\psi\in L^2(C_b)$ as follows. 
With \eqref{eq:Wfinal2} it follows from Stone's formula that
\begin{align*}
  \int_\Delta\,\norm{\Gamma^+_{\beta,k}(\lambda)\psi}^2\,\d \lambda\leq
  \int_\Delta\,\inp{\delta(H-\lambda)}_\psi\,\d \lambda\quad \text{for
  any   interval }
  \Delta\subseteq I_0,
\end{align*}  leading
to the
conclusion  that 
\begin{align}\label{eq:semiBND}
  \norm{\Gamma^+_{\beta,k}(\lambda)\psi}^2\leq
  \inp{\delta(H-\lambda)}_\psi \text{ for a.e. } \lambda \in I_0.
\end{align}
 Since $\Gamma^+_{\beta,k}(\cdot)\psi\in L^2(C_b)$ is
weakly continuous and the functional $\norm{\cdot}^2$  is weakly lower semi-continuous we
conclude that 
\begin{align*}
  \norm{\Gamma^+_{\beta,k}(\lambda)\psi}^2 \leq
  \inp{\delta(H-\lambda)}_\psi  \text{ for all } \lambda \in I_0.
\end{align*} In particular also  $\norm{\Gamma^+_{\beta}(\lambda)\psi}^2=\lim _{k\to
  \infty}\norm{\Gamma^+_{\beta,k}(\lambda)\psi}^2  \leq
\inp{\delta(H-\lambda)}_\psi $, showing that  indeed $\Gamma^+_{\beta}(\lambda)\psi\in L^2(C_b)$.  

We can now record,  using  Lemma \ref{lemma:strongCont} and
\eqref{eq:LAPbnda},  that for any $ s>1/2$
\begin{align*}
  \Gamma^+_{\beta}(\lambda)\in \vL( L_{s}^2(\bX),L^2(C_b))
 \text{
  with a weakly continuous  dependence  of }  \lambda \in I_0,
\end{align*} and in fact  that the  \emph{outgoing channel wave matrix}
\begin{align*}
  \Gamma^+_{\beta}(\lambda)^*\in \vL(L^2(C_b),  L_{-s}^2(\bX))
 \text{
  depends strongly continuously of }  \lambda \in I_0.
\end{align*}

By taking $k\to \infty$  in \eqref{eq:Wfinal2} we conclude that  for any $\psi\in L_{s}^2(\bX)$,  $s>1/2$,
\begin{align*}
  \parb{F_\beta
  (W^+_\beta)^*\psi}(\lambda)=\Gamma_\beta^+(\lambda)\psi\quad\text{for a.e.
  }
  \lambda\in I_0.
\end{align*}

We can argue similarly for $W^-_\beta$, cf. \eqref{eq:Wfinal2a},  and
introduce the \emph{incoming  channel wave
  matrix} $\Gamma^-_{\beta}(\lambda)^*$, concluding by  the same
reasoning that 
\begin{align*}
  &\Gamma^-_{\beta}(\lambda)\in \vL( L_{s}^2(\bX),L^2(C_b))\text{
  depends weakly continuously of }  \lambda \in I_0,\\
  &\Gamma^-_{\beta}(\lambda)^*\in \vL(L^2(C_b),  L_{-s}^2(\bX))
 \text{
  depends strongly continuously of }  \lambda \in I_0,\\
&\forall \,\psi\in L_{s}^2(\bX), \, s>1/2:\quad \parb{F_\beta
  (W^-_\beta)^*\psi}(\lambda)=\Gamma_\beta^-(\lambda)\psi\quad\text{for a.e.
  }
  \lambda\in I_0.
\end{align*}

We summarize as follows.
\begin{thm}\label{thm:chann-wave-matr}  Let $\beta$ be a  given channel  $\beta
    =(b,\lambda^\beta, u^\beta)$,     $f:I^\beta=(\lambda^\beta,\infty)\to \C$ be
  bounded and 
  continuous, and let  and $s>1/2$. For any $\varphi\in
  L^2_s(\mathbf  X_b)$   
  \begin{subequations}
  \begin{align}\label{eq:wav1}
  W^\pm_\beta
  (f1_{I_0})\paro{ k_\beta}\varphi=\int_{I_0} \,f(\lambda)
    \Gamma^\pm_{\beta}(\lambda)^* \gamma_{b,0}(\lambda_\beta)\varphi
    \,\d \lambda\in  L_{-s}^2(\bX), 
              \end{align} 
More generally for any $\varphi
\in1_{I_0}(k_\beta) L^2(\bX_b)$
\begin{align}\label{eq:wav2}
  W^\pm _\beta
  f\paro{ k_\beta}\varphi=\int_{I_0} \,f(\lambda)
  \Gamma^\pm_{\beta}(\lambda)^* \paro{F_\beta \varphi)(\lambda,
  \cdot} \,\d \lambda \in L_{-s}^2(\bX), 
              \end{align} 

In \eqref{eq:wav1} the integrand is a bounded and continuous
$L_{-s}^2(\bX)$-valued function. For   \eqref{eq:wav2} the integral
has  the weak interpretation of a    measurable 
$L_{-s}^2(\bX)$-valued function.
  \end{subequations}
\end{thm}

 For $f(\lambda)= \e^{-\i t\lambda}$,
$t\in\R$,  the formulas \eqref{eq:wav1} and \eqref{eq:wav2} represent
Schr\"odinger wave packets of energy-localized states outgoing to  or incoming  from
the channel $\beta$, cf. a discussion in Subsection \ref{subsec:A principle example, atomic $N$-body
  Hamiltonians}.  

\subsection{Parseval identities and   construction of the scattering matrix}\label{subsec:Parseval identity}
Let $Q=Q_0$ be given as in \eqref{eq:smoothBnd} (for an arbitrary $\epsilon>0$).
By the orthogonality and completeness of channel wave operators we deduce the Parseval
formulas  
\begin{align}\label{eq:parseval}
  \begin{split}
   &\int_\Delta\,\,\sum_{\lambda^\beta< \lambda_0}\,\norm{\Gamma^\pm_{\beta}(\lambda)Q_0^*\varphi}^2\,\d \lambda=
  \int_\Delta\,\inp{\delta(H-\lambda)}_{Q_0^*\varphi }\,\d \lambda;\\&\quad\quad\quad\quad\text{for
   all intervals }
  \Delta\subseteq I_0 \mand \varphi\in \vH. 
  \end{split}
\end{align} Moreover, arguing more or less as before, we can deduce
from \eqref{eq:parseval}  that
\begin{subequations}
 \begin{align}\label{eq:pointw1}
 &\sum_{\lambda^\beta<
  \lambda_0}\,\norm{\Gamma^\pm_{\beta}(\lambda)Q_0^*\varphi }^2\leq 
  \inp{\delta(H-\lambda)}_{Q_0^*\varphi } \quad\text{ for all } \lambda \in I_0,\\
&\sum_{\lambda^\beta<
  \lambda_0}\,\norm{\Gamma^\pm_{\beta}(\lambda)Q_0^*\varphi }^2= 
  \inp{\delta(H-\lambda)}_{Q_0^*\varphi } \quad\text{  a.e.  in }
  I_0,\text{  say for } \lambda \in I_0
\setminus \vN_0.   \label{eq:pointw2}
\end{align}   
\end{subequations} By first taking first $\varphi$ from any  dense countable
subset $\vC\subseteq \vH$ we then  
obtain \eqref{eq:pointw1} and \eqref{eq:pointw2}
for all vectors  from  $\vC$ 
with the null set $\vN_0$ in   \eqref{eq:pointw2} being independent of
$\varphi\in \vC$. Secondly we extend by continuity  and conclude that  \eqref{eq:pointw2} holds with the same null
set $\vN_0$ for  all
$\varphi\in\vH$. In particular we obtain  
that for some null subset $\vN_0\subseteq I_0$
\begin{subequations}
 \begin{align}\label{eq:Finalgood0}
&\forall \psi\in L^2_\infty:\,\,\sum_{\lambda^\beta<
    \lambda_0}\,\norm{\Gamma^\pm_{\beta}(\lambda)\psi }^2\leq 
  \inp{\delta(H-\lambda)}_{\psi}\text{ for all }  \lambda
  \in I_0,\\
 &\forall \psi\in L^2_\infty:\,\,\sum_{\lambda^\beta<
    \lambda_0}\,\norm{\Gamma^\pm_{\beta}(\lambda)\psi }^2= 
  \inp{\delta(H-\lambda)}_{\psi}\text{ for all }  \lambda
  \in I_0\setminus \vN_0.  \label{eq:Finalgood}
\end{align} 
\end{subequations}

We are lead to introduce the following concept.
\begin{defn}\label{defn:scatEnergy}  
A {stationary scattering energy} is  a point  $\lambda\in I_0$ for which  
\begin{align}\label{eq:ScatEnergy}
  \forall \psi\in  L^2_\infty:\,\, \sum_{\lambda^\beta<
  \lambda_0}\,\norm{\Gamma^\pm_{\beta}(\lambda)\psi}^2= 
  \inp{\delta(H-\lambda)}_{\psi}.
\end{align} 
  \end{defn}

 We have shown that almost all energies in $ I_0$ are stationary scattering
energies. The property missing for possibly having this conclusion for all
energies in $ I_0$ is clearly the continuity in $\lambda$ of the sum to the left in \eqref{eq:ScatEnergy}.

Recall that the above operator $Q_0$ is only one out of many
`$Q$-operators' from our procedure. More
precisely  for each
              $k$ the expansions of $
              T^\pm_{\beta}$ used in Sections \ref{sec:Computation of
                T} and \ref{sec:Formula for widetilde} involved
              several (but 
              finitely many)  `$Q$-operators' either on the above form
              $Q_0$ or being one of those from  Lemma \ref{lemma:Hbnds}. So when varying $k$ we
              obtain in total an infinite  family of `$Q$-operators',
              say  denoted  by $\vQ$.

For any  $Q\in \vQ$ and $\varphi \in L^2_\infty$   we  record that
$\tilde f_2(H) Q^*\varphi \in  L^2_\infty$, cf. Remark
\ref{remark:breveBNDs} \ref{item:P0}, and 
\begin{align*}
  \Gamma^\pm_{\beta}(\lambda)Q^*\varphi
  =\Gamma^\pm_{\beta}(\lambda)\tilde  f_2(H)  Q^*\varphi =\lim_{k\to \infty}
  \Gamma^\pm_{\beta,k}(\lambda)\tilde  f_2(H) Q^*\varphi \in L^2(C_b);\,\lambda\in I_0.
\end{align*} 
 However (more generally) for any finite sum $\psi=\sum_j Q_j^*\varphi_j$, where now 
 $\varphi_j$ is  arbitrary  from $\vH$, the quantity  $ \Gamma^\pm_{\beta,k}(\lambda)\psi=\Gamma^\pm_{\beta,k}(\lambda)\tilde
 f_2(H) \psi$ is well-defined. Thanks to
 \eqref{eq:Wfinal2} it is  unambiguously given   as
 \begin{align*}
   \Gamma^\pm_{\beta,k}(\lambda)
\psi=\sum_j 1_{C_{b,k}}\Gamma^\pm_{\beta,k}(\lambda)
  Q_j^*\varphi_j\in L^2(C_b),
 \end{align*} and as used  before
 $1_{C_{b,k}}\Gamma^\pm_{\beta,k}(\lambda)=1_{C_{b,k}}\Gamma^\pm_{\beta,k+1}(\lambda)$,
 allowing us to compute
 \begin{align*}
   \Gamma^\pm_{\beta}(\lambda) \psi:&=\lim_{k\to \infty} \Gamma^\pm_{\beta,k}(\lambda) \psi\\&=\sum_j \lim_{k\to \infty} 1_{C_{b,k}}\Gamma^\pm_{\beta,k}(\lambda)
   Q_j^*\varphi_j\\&
=\sum_j \lim_{k\to \infty} 1_{C_{b,k}}\Gamma^\pm_{\beta}(\lambda)
    Q_j^*\varphi_j\\&=\sum_j \Gamma^\pm_{\beta}(\lambda)
    Q_j^*\varphi_j \in L^2(C_b).
 \end{align*} We  used the following  analogue of \eqref{eq:semiBND},
\begin{align}\label{eq:SE}
  \norm{\Gamma^+_{\beta,k}(\lambda)Q_j^*\varphi_j}^2 \leq
  \inp{\delta(H-\lambda)}_{f_1(H)Q_j^*\varphi_j}  \text{ for all } \lambda \in I_0.
\end{align} This reasoning yields that  
\begin{align*}
   \forall Q\in \vQ:\,\Gamma^\pm_{\beta}(\lambda)Q^*\in\vL(
\vH,L^2(C_b)),
\end{align*} in fact with a weakly continuous dependence of
$\lambda \in I_0$. 

We are lead  to introduce the  subspace $\vV\subseteq \vH^{-2}$ of  vectors 
$\psi$ as above, i.e.  vectors being  given as a finite
sum  of terms on  the form $Q^* \varphi$, where $Q\in \vQ$ and
$\varphi\in\vH$. Mimicking the proof of \eqref{eq:pointw1}   we
conclude the following extension of \eqref{eq:SE},
\begin{align}\label{eq:Finalgood02a}
\forall \psi\in \vV:\,\,\sum_{\lambda^\beta<
    \lambda_0}\,\norm{\Gamma^\pm_{\beta}(\lambda)\psi }^2\leq 
  \inp{\delta(H-\lambda)}_{f_1(H)\psi}\text{ for all }  \lambda
  \in I_0.
\end{align} We are lead  to define
\begin{align*}
  \Gamma^\pm(\lambda):\vV\to  \vG:=\Sigma^\oplus_{\lambda^\beta<
    \lambda_0}L^2(C_b);\quad\Gamma^\pm(\lambda)=\parb{\Gamma^\pm_{\beta}(\lambda)}_{\lambda^\beta
  <\lambda}.
\end{align*} By \eqref{eq:Finalgood02a} and the above computation, 
$\Gamma^\pm(\lambda)=\Gamma^\pm(\lambda)\tilde  f_2(H)$ are  well-defined linear maps.  Using the natural 
  Hilbert space structure on  $\vG$, 
\eqref{eq:Finalgood02a} may be phrased as
\begin{subequations}
\begin{align}\label{eq:Finalgood023}
\forall \psi\in \vV:\,\,\norm{\Gamma^\pm(\lambda)\psi }^2\leq 
  \inp{\delta(H-\lambda)}_{f_1(H)\psi}\text{ for all }  \lambda
  \in I_0.
\end{align}

The 
  expression $\inp{\delta(H-\lambda)}_{f_1(H)\psi}$  appearing to the right in
  \eqref{eq:Finalgood023}   is continuous in $\lambda \in I_0$,
  cf. Remark \ref{remark:breveBNDs} \ref{item:P2}. The
  quantities on the left-hand side might not have this property,
  however we can  record  that $\Gamma^\pm(\cdot)\psi $
  are  weakly continuous and note
the following version af \eqref{eq:parseval},
\begin{align}\label{eq:parseval3}
  \forall \psi\in \vV,\forall \text{ intervals }
  \Delta\subseteq I_0:\,\int_\Delta\,\norm{\Gamma^\pm(\lambda)\psi}^2\,\d \lambda=
  \int_\Delta\,\inp{\delta(H-\lambda)}_{f_1(H)\psi}\,\d \lambda. 
  \end{align}  
  \end{subequations}
 Since almost all energies in $ I_0$ are stationary scattering
energies the Parseval formulas \eqref{eq:parseval3} may be seen as a
consequence of the following extension of \eqref{eq:ScatEnergy}.

\begin{lemma}\label{lemma:pars-ident-constrSvat}
   Let  $\lambda\in I_0$ be
  any  stationary scattering energy. Then 
\begin{align}
 \forall \psi\in \vV:\,\,\norm{\Gamma^\pm(\lambda)\psi }^2=\sum_{\lambda^\beta<
    \lambda_0}\,\norm{\Gamma^\pm_{\beta}(\lambda)\psi }^2= 
  \inp{\delta(H-\lambda)}_{f_1(H)\psi}.  \label{eq:Finalgood2}
\end{align} 
\end{lemma}
\begin{proof}
  Fix an arbitrary $\psi=\Sigma_j Q^*_j \varphi_j\in \vV$,  let
  $\chi_m=\chi_-(r/m)$ for $m\in\N$ and then
 $\psi_m=\Sigma_jQ^*_j \chi_m\varphi_j$. Since
 $\chi_m\varphi_j\in  L^2_\infty$ and $\tilde  f_2(H) Q^*_j=\vO(r^0)$
 (cf. Remark \ref{remark:breveBNDs} \ref{item:P0}), also $\tilde  f_2(H) \psi_m\in
 L^2_\infty$, and   for such  vector  we know  the identities \eqref{eq:ScatEnergy}. By \eqref{eq:Finalgood023} the sequences
 $\Gamma^\pm(\lambda)\tilde  f_2(H) \psi_m =\Gamma^\pm(\lambda) \psi_m\to  \Gamma^\pm(\lambda)\psi$  in
 $\vG$. Whence also $\norm{\Gamma^\pm(\lambda)\tilde  f_2(H) \psi_m}^2 \to
 \norm{\Gamma^\pm(\lambda)\psi}^2$. Similarly $
 \inp{\delta(H-\lambda)}_{f_1(H)\psi_m}\to
 \inp{\delta(H-\lambda)}_{f_1(H)\psi}$. Since the identities in \eqref{eq:ScatEnergy}
 hold  for all elements of two sequences we conclude  the limiting identities,
 agreeing with \eqref{eq:Finalgood2}. 
\end{proof}

\begin{lemma}\label{lemma:exact-channel-wave} Let  $\lambda\in I_0$. Then the subspaces  
  \begin{align*}
    \vG^\pm_\infty=\set{{\Gamma^\pm(\lambda)\psi}=\parb{\Gamma^\pm_{\beta}(\lambda)\psi}_{{\lambda^\beta<
    \lambda_0}}\in
    \vG\,|\,\psi   \in L^2_\infty}\text{ are  dense in  }\vG.
  \end{align*} 
  \end{lemma}
  \begin{proof} It suffices to prove  density of the subspaces $\vG^\pm_\vV=\set{{\Gamma^\pm(\lambda)\psi}\in
    \vG\,|\,\psi   \in \vV}$, cf. the proof of Lemma \ref{lemma:pars-ident-constrSvat}.
    Let us only
    consider the `minus case'.  For any `open channel'  $\alpha$,  $k\in \N$ and $g_a\in L^2(C_a)$  we consider
    \begin{subequations}
    \begin{align}\label{eq:gPsi}
      \psi_{\alpha}=\psi_{\alpha,k,g_a}=T^- _{\alpha,k}
    J_\alpha{\gamma_a^- (\lambda_\alpha)}^*1_{C_{a,k}}(\hat \xi_a)g_a.
    \end{align}   
  We can write $ \psi_{\alpha}$ as a finite sum of vectors on
    the form $Q^*\varphi$ where $Q\in \vQ$ and $\varphi\in\vH$. Whence
    $\psi_{\alpha}\in \vV$ and 
  $\Gamma^-_{\beta}(\lambda)\psi_{\alpha}$ is a
    well-defined vector in $ L^2(C_b)$. We compute for any $l\in \N$
    the vector 
    $1_{C_{b,l}}(\hat \xi_b)\Gamma^-_{\beta}(\lambda)\psi_{\alpha}$
      using  Remark \ref{remark:proof-eqrefeq:srep}, concluding
 that $1_{C_{b,l}}(\hat \xi_b)\Gamma^-_{\beta}(\lambda)\psi_{\alpha}=0$ if $\beta\neq
\alpha$, while for $\beta=
\alpha$
\begin{align*}
  1_{C_{b,l}}(\hat \xi_b)\Gamma^-_{\beta}(\lambda)\psi_{\alpha}=1_{C_{a,l}}(\hat \xi_a)\Gamma^-_{\alpha,l}(\lambda)\psi_{\alpha}= -
\tfrac {2\lambda_\alpha}\pi 1_{C_{a,l}\cap {C_{a,k}}}g_a.
\end{align*} By taking $l\to \infty$ it follows that
\begin{align}\label{eq:fundCOMP}
  \vG^-_\vV\ni \Gamma^-(\lambda)\psi_{\alpha}=( \delta_{ \beta \alpha} 
\tfrac {-2\lambda_\alpha}\pi 1_{{C_{a,k}}}g_a)_{{\lambda^\beta<
    \lambda_0}}.
\end{align}\end{subequations}  By varying over $\alpha$ (with $\lambda^\alpha<\lambda_0$), $k\in \N$ and $g_a\in L^2(C_a)$ 
the
corresponding vectors to the right in \eqref{eq:fundCOMP} span  a
dense  subset  of $\vG$.
\end{proof}

\begin{defn}\label{defn:U} 
  Let  $\lambda\in I_0$ be
  any  stationary scattering energy. Then we let  $U(\lambda)\in \vL(\vG)$ be
  the uniquely defined unitary
  operator $U(\lambda)$ on $\vG$ such that 
  \begin{align*}
    U(\lambda){\Gamma^-(\lambda)\psi}={\Gamma^+(\lambda)\psi};
\quad \psi   \in \vV.
  \end{align*}  (Note that Lemma \ref{lemma:exact-channel-wave} is used.)
\end{defn}
\begin{defn}\label{defn:scattering matrix} 
  Let  $\lambda\in I_0$. The  scattering matrix  $S(\lambda)$ is  the unique bounded operator  on $\vG$  with each $\beta\alpha$-entry
  given by  $S_{\beta\alpha}(\lambda)$ as specified in Theorem
  \ref{thm:mainstat-modif-n}. (Note that $\norm {S(\lambda)} \leq 1$
  by  the orthogonality of channels and the weak continuity.)
\end{defn}

\begin{thm}\label{thm:mainTHM} Let $\lambda\in
  I_0$ be   any  stationary scattering energy. Then  $S(\lambda)$ is unitary. 
 Moreover 
    the $\vL(\vG)$-valued function $S(\cdot )$ is strongly
  continuous at $\lambda$, for 
  any $\psi \in \vV$ the $\vG$-valued functions 
  $\Gamma^\pm(\cdot)\psi=(\Gamma^\pm_\alpha(\cdot)\psi)_{{\lambda^\alpha<
    \lambda_0}}$ are  norm-continuous
  at $\lambda$, and
  \begin{align}\label{eq:scaSS}
    S(\lambda )\Gamma^-(\lambda)\psi=\Gamma^+(\lambda)\psi.
  \end{align}
\end{thm}
\begin{proof}
  It suffices to check that
  $S_{\beta\alpha}(\lambda)=U_{\beta\alpha}(\lambda)$. By 
\eqref{eq:gPsi},  \eqref{eq:fundCOMP},
\eqref{eq:AppExa}, \eqref{eq:SreP}  and \eqref{eq:locScata}
\begin{align*}
  &1_{C_{b,l}}(\hat \xi_b)U_{\beta\alpha}(\lambda)1_{{C_{a,k}}}g_a\\&=-\tfrac {\pi}{2\lambda_\alpha}\Gamma^+_{\beta,l}(\lambda)T^- _{\alpha,k}
    J_\alpha{\gamma_a^- (\lambda_\alpha)}^*1_{C_{a,k}}(\hat\xi_a)g_a\\& 
=-\tfrac {\pi}{2\lambda_\alpha}\tfrac \pi {2\lambda_\beta} 1_{C_{b,l}}(\hat \xi_b){\gamma_b^+(\lambda_\beta)}J^*_\beta
  \parb{T^+_{\beta,l}}^* \delta(H-\lambda)T^- _{\alpha,k}
    J_\alpha{\gamma_a^-
                                                                        (\lambda_\alpha)}^*1_{C_{a,k}}(\hat\xi_a)g_a\\&
=\tfrac{(2\pi\i)^2}{16\lambda_\beta \lambda_\alpha} 1_{C_{b,l}}(\hat \xi_b){\gamma_b^+(\lambda_\beta)}J^*_\beta
  \parb{T^+_{\beta,l}}^* \delta(H-\lambda)T^- _{\alpha,k}
    J_\alpha{\gamma_a^-
                                                                        (\lambda_\alpha)}^*1_{C_{a,k}}(\hat\xi_a)g_a\\&
=\parb{16\lambda_\beta \lambda_\alpha}^{-1} 1_{C_{b,l}}(\hat \xi_b) \widetilde
  S_{\beta\alpha}(\lambda)1_{C_{a,k}}(\hat\xi_a)g_a\\& = 1_{C_{b,l}}(\hat \xi_b) S_{\beta\alpha}(\lambda)1_{C_{a,k}}(\hat\xi_a)g_a.
 \end{align*} We take $k,l\to \infty$ and conclude that indeed
 $U_{\beta\alpha}(\lambda)g_a=S_{\beta\alpha}(\lambda)g_a$ for any  $g_a\in L^2(C_a)$.
\end{proof}

The free energy $K:=\diag(k_\alpha)$  is diagonalised on the neighbourhood 
$I_0\ni \lambda_0$ by 
\begin{align*}
  F_{I_0}: \sum_{\lambda^\alpha < \lambda_0} &\oplus 1_{I_0}(k_\alpha) L^2(\bX_a)\to  \sum_{\lambda^\alpha < \lambda_0} \oplus L^2(I_0,L^2(C_a))=\int_{I_0}^\oplus \vG \,\d \lambda;\\&F_{I_0}=\diag(F_\alpha).
\end{align*} Recall from Theorem \ref{thm:mainstat-modif-n} the
notation $S_{\beta\alpha}=\parb{W_{\beta}^+}^*W_{\alpha}^-$.
\begin{corollary}\label{cor:direct} The scattering matrix  considered
  as a map 
\begin{align*}
I_0\ni \lambda\to S(\lambda)\in\vL \paro{\vG}\text{ is 
 weakly continuous.}
\end{align*}   Within the class of such maps it is uniquely determined by the identity
  \begin{align}
  F_{I_0}SF^{-1}_{I_0}&=F_{I_0}\paro{S_{\beta\alpha}}_{\lambda^\beta,\,\lambda^\alpha < \lambda_0}\,F^{-1}_{I_0}=\int^\oplus_{I_0}
  S(\lambda)\,\d \lambda.
\end{align} For almost all energies in $I_0$ the scattering matrix
is unitary and strongly continuous.
\end{corollary}
\begin{remark*}
 It is stated in  \cite [Theorem 6.7]{Ya1} that for short-range
systems the  scattering
matrix is strongly continuous away from thresholds and eigenvalues. The proof
appears  incorrect  to the author (however the   assertion  of weak continuity
is correctly derived).    
\end{remark*}

 \subsection{Besov space setting and  minimum  generalized eigenfunctions}\label{subsubsec:Besov space setting at stationary scattering energies}
 In Lemma \ref{lemma:pars-ident-constrSvat} we considered $\Gamma^\pm(\lambda)=\parb{\Gamma^\pm_{\beta}(\lambda)}_{\lambda^\beta
  <\lambda}$  as mappings 
 $\Gamma^\pm(\lambda):\vV\to \vG$; 
 $\lambda\in I_0$. The extended space $\vV\supseteq
 L^2_\infty$ was convenient  for our constructions, however there is
 a different extension which  appears 'more natural',  namely the Besov
 space $\vB$ in  which 
 $L^2_\infty$ is densely embedded. Using 
 \eqref{eq:BB^*a}, \eqref{eq:Finalgood0} and extension by continuity it follows that
 $\Gamma^\pm(\lambda)\in \vL(\vB,\vG)$ with a uniform bound in
 $\lambda\in I_0$. 
Clearly $\Gamma^\pm(\cdot)\in\vL(\vB,\vG)$ are
 strongly continuous at  any stationary scattering energy $\lambda\in I_0$ in
 which case \eqref{eq:scaSS} also holds for $\psi\in \vB$. Due to
  Theorem  \ref{prop:besov-space-setting} (stated below) the vector
  $\Gamma^-(\lambda)\psi$ with 
 $\psi\in \vB$  (appearing 
to
 the left in  \eqref{eq:scaSS}) in this case is  the generic form of vectors in $\vG$.

Introducing the space of  minimum  generalized eigenfunctions,
\begin{align*}
  \vE_\lambda=\set{\phi\in \vB^*|\, (H-\lambda)\phi=0},
\end{align*}
 we note  that $\delta(H-\lambda)\psi\in\vE_\lambda$ for
$\psi\in \vB$ (or possibly for $\psi\in f_1(H)\vV$). However if $\lambda$ is a
stationary scattering energy, in fact  any
$\phi\in\vE_\lambda$ has this form. This is a 
 consequence
of  the following more general 
result. We remark that  there are   analogous assertions   for various one-body type Hamiltonians in the literature, see for
example \cite{IS2}.  
\begin{thm}\label{prop:besov-space-setting} Let $\lambda\in
  I_0$ be   any  stationary scattering energy.
  Then the  operators $\Gamma^\pm(\lambda): \vB\to \vG$,
  $\Gamma^\pm(\lambda)^*: \vG\to \vE_\lambda$   and $\delta(H-\lambda)
  : \vB \to \vE_\lambda$  all map 
 onto. Moreover the wave matrices
  $\Gamma^\pm(\lambda)^*: \vG\to \vE_\lambda(\subseteq \vB^*)$ are bi-continuous linear isomorphisms.
\end{thm}
\begin{proof} The ranges of $\Gamma^\pm(\lambda)$
  are dense due to Lemma
  \ref{lemma:exact-channel-wave}, and    clearly $\vE_\lambda$ is closed in
  $\vB^*$. 
  Thanks to   the open
  mapping and closed range theorems  (see for example see \cite{Yo})
  it then suffices 
  to show that  $\Gamma^\pm(\lambda)^*$ map onto $\vE_\lambda$.  (Here
  and below we use the identities 
$\Gamma^\pm(\lambda)^*\Gamma^\pm(\lambda)=\delta(H-\lambda)$ as
quadratic forms on $\vB$.)

For convenience we only show that  $\Gamma^+(\lambda)^*$ maps
 onto
$\vE_\lambda$.  Let $B$ is given by \eqref{eq:B} (possibly rescaled),
and let  $h$ be the   function $h(e)=e\tilde f_1(e)$  ($h(e)$ is
 denoted by $g(\lambda)$  in
Section \ref{sec:Computation of T}) and note that  $h(H)-\lambda$ annihilates the spaces $\ran \Gamma^+(\lambda)^*\subseteq\vE_\lambda$. Due to \cite[Corollary
    1.10]{AIIS}  there exists  $\epsilon>0$ such that 
    \begin{align}\label{eq:MicroB}
      \chi_-(\pm B/\epsilon) R(\lambda\pm\i0)\psi\in \vB_0^*
      \text{ for any }\psi\in \vB.
    \end{align}  (The fact  that $\chi_-(\pm B/\epsilon)=\vO(r^0)$ implies that
    $\chi_-(\pm B/\epsilon)\in \vL(\vB)$, cf.  \cite[Theorem 14.1.4]{H1}.)
   For  any given $\phi\in \vE_\lambda$ we introduce for $\epsilon>0$
   taken small
   enough
   \begin{align*}
     \phi_+=\chi_+(
   B/\epsilon)\phi,\quad \phi_-=\phi-\phi_+, \quad g_\rho=\Gamma^+(\lambda)\chi_\rho \parb{h(H)-\lambda}\phi_+;
   \end{align*}
 here 
 $\chi_\rho=\chi_-(r/\rho)$, $\rho>1$.  By commuting the factor $h(H)-\lambda$ through the factor $\chi_\rho$
we deduce that $\sup_{\rho>1}\,\norm
{g_\rho}<\infty$ (note that only the  commutator contributes).
  Let $g_\infty=\wlim_{n\to \infty}g_{\rho_n}\in \vG$, the limit
taken  along a  suitable
sequence $\rho_n\to \infty$, and let $\breve
\phi=\Gamma^+(\lambda)^*g_\infty$. We can  now mimic the proof
of \cite[Lemma 3.12]{IS2},  using the identity 
$\Gamma^+(\lambda)^*\Gamma^+(\lambda)=\delta(H-\lambda)$, the
decomposition $\phi=\phi_++\phi_-$  and
the bounds \eqref{eq:MicroB} in combination with \eqref{eq:firstHB},  and conclude that $\phi=\breve
\phi$. In particular $\phi\in \ran \Gamma^+(\lambda)^*$ as wanted.
\end{proof}

The following functions $w^{\pm}_{\lambda,\beta} [g]$ are identical
with the functions $v^{\pm}_{\lambda,\beta} [g]$ of Subsection
\ref{Radial limits} up to  trivial factors.

\begin{lemma}\label{lemma:besov-space-setting} Let $\beta=(b,\lambda^\beta, u^\beta)$ be any
  channel with $\lambda^\beta < \lambda_0$. For $\epsilon>0$ taken
  small enough the following bounds hold for any $\lambda\in
  I_0$ and any $g\in L^2(C_b)$.
  \begin{align*}
    \chi^2_-(\mp B/\epsilon)\Gamma^\pm_{\beta}(\lambda)^*g-J_\beta
                            w^{\pm}_{\lambda,\beta} [g]\in \vB_0^*,
\end{align*} where
\begin{align*} 
  w^{\pm}_{\lambda,\beta} [g]( x_b)&= d_\beta^\pm \lambda_\beta^{-1/4} \chi_+(\abs{x_b}) \abs{x_b}^{(1-n_b)/2}
                                      \e^{\pm \i
                                        K_b(x_b,\lambda_\beta)}g(\hat x_b);\\&\quad\quad\quad\quad d_\beta^\pm=\mp\tfrac \i 2 \e^{\mp
  \i(n_b-3)/4}\pi^{-1/2}.
 \end{align*}
  \end{lemma}
  \begin{proof} 
 By \eqref{eq:AppExa} and its `minus-analogue'
\begin{align}\label{eq:AppExa2}
  \Gamma^\pm_{\beta,k}(\lambda)^*=\pm \tfrac \pi {2\lambda_\beta} f_1(H){\delta(H-\lambda) T^\pm_{\beta,k} J_\beta 
      \gamma_{b}^\pm(\lambda_\beta)^*}1_{C_{b,k}}.
\end{align}  The identity \eqref{eq:AltForm}  may  formally be seen as
a  replacement  
 of the factor
 $\pm 2\pi\delta(H-\lambda)
 T^\pm_{\beta,k}$ from \eqref{eq:Wfinal}  by   ${\Phi_{\beta,k}^\pm+\i
   R(\lambda\mp\i 0)T^\pm_{\beta,k}} $, and from 
the derivation of
\eqref{eq:Wfinal} we see that this is legitimate considering the image
space of the operators as being 
$L^2_{-1}(\bf  X)$. 
 In particular there exist  in $\vL(L^2(C_b),L^2_{-1}(\bf  X))$
\begin{align*}
 f_1(H)R(\lambda\mp\i 0)T^\pm_{\beta,k}J_\beta 
      \gamma_{b}^\pm(\lambda_\beta)^*1_{C_{b,k}}:=\wlim_{\epsilon\to 0_+}f_1(H)R(\lambda\mp\i \epsilon)T^\pm_{\beta,k}J_\beta 
      \gamma_{b}^\pm(\lambda_\beta)^*1_{C_{b,k}} .
\end{align*}  
By  combining  \eqref{eq:kato10} and Lemma
 \ref{lemma:strongCont} we obtain  that
\begin{align*}
  \sup_{\epsilon>0}\,\norm{R(\lambda\mp\i \epsilon)T^\pm_{\beta,k}J_\beta 
      \gamma_{b}^\pm(\lambda_\beta)^*1_{C_{b,k}}}_{
  \vL(L^2(C_b),\vB^*)}<\infty.
\end{align*} 
 Whence  in fact  
\begin{align*}
   S^\mp:=f_1(H)R(\lambda\mp\i 0)T^\pm_{\beta,k}J_\beta 
      \gamma_{b}^\pm(\lambda_\beta)^*1_{C_{b,k}}\in\vL(L^2(C_b),\vB^*).
\end{align*} 
Next we expand  for each sign $T^\pm_{\beta,k}$ in terms of the $Q$-operators and
consider the approximation of $ S^\mp$ given by replacing   $R(\lambda\mp\i 0)f_1(H) Q^*_j$  by 
 $R(\lambda\mp\i 0)f_1(H) Q^*_j \chi_\rho$, say denoted by
 $S^\mp_\rho$ (as in the proof of Theorem
 \ref{prop:besov-space-setting}, here  
 $\chi_\rho=\chi_-(r/\rho)$ with   $\rho$ considered large).  Now $S^\mp_\rho \to S^\mp$ strongly in
 $\vL(L^2(C_b),\vB^*)$, and by \eqref{eq:MicroB} 
$\chi^2_-(\mp
 B/\epsilon)S^\mp_\rho\in\vL(L^2(C_b),\vB_0^*)$. Whence also $\chi^2_-(\mp
 B/\epsilon)S^\mp\in\vL(L^2(C_b),\vB_0^*)$.

 By density and continuity   we are left with showing 
\begin{align*}
       { \tfrac 1 {4\lambda_\beta}
  \chi^2_-(\mp
 B/\epsilon)f_1(H)\Phi_{\beta,k}^\pm J_\beta 
      \gamma_{b}^\pm(\lambda_\beta)^*1_{C_{b,k}}g_k-J_\beta
    w^{\pm}_{\lambda,\beta} [g_k]}\in \vB_0^*
    \end{align*} for any $g_k=1_{C_{b,k}}g\in  C^\infty
  (C_b)$, $k\in \N$. In turn, thanks to the form of $\Phi_{\beta,k}^\pm
  $ (and commutation), it remains to show that 
\begin{align*}
       { \tfrac 1 {4\lambda_\beta}
  f_1(H)\Phi_{\beta,k}^\pm J_\beta 
      \gamma_{b}^\pm(\lambda_\beta)^*1_{C_{b,k}}g_k-J_\beta
    w^{\pm}_{\lambda,\beta} [g_k]}\in \vB_0^*
    \end{align*} for any such  $g_k$.

By using the  analysis of Steps IV and V  in
    the proof of Lemma \ref{lemma:strongCont} (not to be repeated
    here) the first term
    simplifies 
    exactly as the subtracted  second term modulo a term in $\vB_0^*$.
\end{proof}

We can now give a `geometric interpretation' of the operator
$\Gamma^-_{\alpha}(\lambda)^*$ needed for a wave packet description of
an incoming  $\alpha$-channel  experiment, energy-localized in
$I_0$, cf. a discussion in Subsection \ref{subsec:A principle example, atomic $N$-body
  Hamiltonians}. More precisely we  characterize for any stationary scattering
  energy  the space  of
  generalized eigenfunctions  in  the range of
  $\Gamma^-_{\alpha}(\lambda)^*$ in terms of asymptotics expressed by
  the quasi-modes $J_\beta
                            w^{\pm}_{\lambda,\beta} [\cdot]$ of Lemma \ref{lemma:besov-space-setting}. The result
  may also be vieved as a  characterization of the    incoming 
  $\alpha$-channel  part of the
  scattering matrix.
We equip the space $\vB^*/\vB_0^*$ with the quotient-norm, thus
making it  a Banach space.
  \begin{thm}\label{Cor:besov-space-setting}  Let $\lambda\in
  I_0$ be   any  stationary scattering energy and $\alpha=(a,\lambda^\alpha, u^\alpha)$ be any
  channel with $\lambda^\alpha < \lambda_0$. Then the following
  existence and uniqueness result holds for any  $
  g\in L^2(C_a)$.
  \begin{enumerate}[1)]
  \item \label{item:As10} Let   $
    u=\Gamma^-_{\alpha}(\lambda)^* g$,  and let 
    $(g_\beta)_{\lambda^\beta<\lambda_0}\in \vG$ be   given by
    $g_\beta=S_{\beta\alpha}(\lambda) g$. Then,  as an
    identity in  $\vB^*/\vB_0^*$,
    \begin{align}\label{eq:as}
       u=J_\alpha
    w^{-}_{\lambda,\alpha} [ g]+\sum_{\lambda^\beta<\lambda_0} J_\beta
    w^{+}_{\lambda,\beta} [g_\beta].
    \end{align} 
  \item \label{item:As20} Conversely, if  \eqref{eq:as}
  is fulfilled for some  $ u\in\vE_\lambda$
  and some $(g_\beta)_{\lambda^\beta<\lambda_0}\in \vG$, 
  then  $ u=\Gamma^-_{\alpha}(\lambda)^* g$ and
  $g_\beta=S_{\beta\alpha}(\lambda) g$ for all 
$\lambda^\beta<\lambda_0$.
\end{enumerate}
\end{thm}
\begin{proof}
  {\bf I.} We note that with $F_\rho:=F(\set{x\in X|\,\abs{x}<\rho})$, $\rho>1$, the expression
  \begin{align*}
    \norm{u}_{\rm quo}=\limsup_{\rho\to \infty}\,\rho^{-1/2}\norm{F_\rho u}
  \end{align*} is a norm on $\vB^*/\vB_0^*$ equivalent with the
  quotient-norm. Using this norm the Cauchy criterion assures that for
  any  $(g_\beta)_{\lambda^\beta<\lambda_0}\in \vG$ the right-hand
  side of  \eqref{eq:as} is well-defined. The cross-terms do not
  contribute for the following reason: For given open different channels
  $\beta=(b,\lambda^\beta, u^\beta)$ and $\beta'=(b',\lambda^{\beta'},
  u^{\beta'})$ with $b=b'$,    Fubini's theorem and the fact that
  $\inp{u^\beta, u^{\beta'}}_{\vH^b}=0$ clearly implies
  \begin{align}\label{eq:ort}
    \inp{F_\rho J_\beta
    w^{+}_{\lambda,\beta} [g_\beta],F_\rho J_{\beta'}
    w^{+}_{\lambda,\beta'} [g_{\beta'}]}_{\vH}=\int_{\bX} \,\overline{F_\rho u^\beta
    w^{+}_{\lambda,\beta} [g_\beta]}\, F_\rho u^{\beta'}
    w^{+}_{\lambda,\beta'} [g_{\beta'}]\,\d x=0.
  \end{align}
 If $b\neq b'$ we let $c\in \vA$ be given by $
 \bX^c=\bX^b+\bX^{b'}$ and $n_c=\dim \bX_c$, and  note that in this
 case 
 $n_c\leq (n_b+n_{b'}-1)/2$. By approximation, and by using  Fubini's
 theorem and \caS,  we can assume that the functions
 $g_\beta,g_{\beta'}, u^\beta$ and $u^{\beta'}$ are all  bounded with
 compact support. With this assumption we claim that the  integral in \eqref{eq:ort} is $\vO(\rho^{1/2})$ (and hence in particular
  $o(\rho)$  as wanted). On the support of the integrand, 
$\abs{x^c}\leq C_1(\abs{x^b}+\abs{x^{b'}}) \leq C_2$ and
\begin{align*}
  \abs{x_b}^{(1-n_b)/2}\abs{x_{b'}}^{(1-n_{b'})/2}\leq C_3
  \abs{x}^{1-(n_b+n_{b'})/2} \leq C_4 \abs{x}^{(1-n_c)} \abs{x}^{-1/2}.
\end{align*} If $n_c=0$ it follows that $\abs{x}\leq C_2$,  and
consequently the  integral is  $
\vO(\rho^{0})=\vO(\rho^{1/2})$. If  $n_c\geq 1$ we estimate  $\abs{x}^{(1-n_c)}
\abs{x}^{-1/2}\leq \abs{x_c}^{(1-n_c)} \abs{x_c}^{-1/2}$ and use
spherical coordinates in $\bX_c$ and the fact that $\abs{x_c}\leq
\abs{x}< \rho$, indeed yielding the desired bound  $\vO(\rho^{1/2})$.

 We conclude that 
  \begin{align}\label{eq:pars}
    \norm[\big]{u-J_\alpha
    w^{-}_{\lambda,\alpha} [ g]}_{\rm
    quo}^2=\sum_{\lambda^\beta<\lambda_0} \parb{4\pi
    \lambda_\beta^{1/2}}^{-1}\,\norm{g_\beta}^2.
  \end{align}

{\bf II.} For \ref{item:As10} we decompose $u=\Gamma^-_{\alpha}(\lambda)^* g$ as 
\begin{align*}
  u=\chi^2_-( B/\epsilon)u+\chi^2_+( B/\epsilon)u.
\end{align*} The first term corresponds to the first
term to the  right in  \eqref{eq:as} thanks to  Lemma
\ref{lemma:besov-space-setting}. For the second  term we substitute,  cf. \eqref{eq:linkForm} and \eqref{eq:scaSS}, 
\begin{align*}
     \Gamma^-_{\alpha}(\lambda)^*g=\sum_{\lambda^\beta<\lambda}
  \Gamma^+_{\beta}(\lambda)^*g_\beta;\quad g_\beta= S_{\beta\alpha}(\lambda)g.
  \end{align*} The series converges in $\vB^*$. By continuity we can
  take the factor
  $\chi^2_+( B/\epsilon)$ inside the summation. Then we  obtain
  \ref{item:As10} by  applying  Lemma
\ref{lemma:besov-space-setting} to the  terms of  the (convergent)  series $\sum_{\lambda^\beta<\lambda}
  \chi^2_+( B/\epsilon)\Gamma^+_{\beta}(\lambda)^*g_\beta$.

{\bf III.}  For \ref{item:As20} we note that
$ u':=u-\Gamma^-_{\alpha}(\lambda)^* g$ is in $\vE_\lambda$ and
obeys
\begin{align*}
  \chi_-( B/\epsilon) u'=0 \text{ in }\vB^*/\vB_0^*.
\end{align*}
  Here we first use \eqref{eq:as} to $\Gamma^-_{\alpha}(\lambda)^* g $ to obtain a similar representation of
  $u'\in \vB^*/\vB_0^*$, then  we multiply by  $\chi_-( B/\epsilon)$ using 
the fact
that 
this operator  is continuous in $\vB^*/\vB_0^*$  and finally we use Appendix
\ref{sec:AppendixO} to the terms of the resulting series. (Note that
\eqref{eq:bminus} for any $g\in L^2(C_a)$  follows by approximation  
since we know the assertion for $ g\in
  C_\c^\infty (C_a')$.)  This means that $\chi_-( B/\epsilon) u'\in
\vB_0^*$ and then in turn,   thanks to \cite[Corollary
    1.10]{AIIS}, that $ u'=0$. Finally 
\eqref{eq:pars} applied to the mentioned representation of  $u'$ 
    implies that also $g_\beta-S_{\beta\alpha}(\lambda) g=0$, as wanted.
\end{proof}

\appendix

\section{Completion of the proof  of Lemma
  \ref{lemma:Sommerfeld}}\label{sec:AppendixO}

We complete the proof of Lemma \ref{lemma:Sommerfeld}
\ref{item:1s}.  For convenience we only show that $\check
P^+_{\lambda,\alpha} [g]=0$. By  \cite[Corollary
    1.10]{AIIS} it suffices to show that for  some (small)
$\epsilon>0$ 
    \begin{align}\label{eq:bminus}
      \chi_-(B/\epsilon) u^\alpha\otimes v^{+}_{\lambda,\alpha}
                              [g]\in \vB_0^*, 
    \end{align}  cf. \eqref{eq:MicroB}. Our proof of \eqref{eq:bminus}
    does not rely on \eqref{eq:decay}. 
    We recall that 
    \begin{align*}
       v^{+}_{\lambda,\alpha} [g]( x_a)= \chi_+(\abs{x_a})\abs{x_a}^{(1-n_a)/2}
                                      \e^{ \i
                                        K_a(x_a,\lambda_\alpha)}g(\hat x_a) ;\quad g\in
  C_\c^\infty (C_a').
    \end{align*}
 Depending on the support of $g$ we can pick $\xi^+_a$ as in
 \eqref{eq:partM} such that 
\begin{align*}
  \parb{1-\xi^+_a}u^\alpha\otimes v^{+}_{\lambda,\alpha}
                              [g]\in \vB_0^*.
\end{align*} Let $\brH_a$ be given by \eqref{eq:brevH} and $E<\min\sigma(\brH_a)$. By commutation, cf. Lemma \ref{lem:171113b}, it suffices
to find $\epsilon>0$ such that
\begin{align}\label{eq:goal}
  \begin{split}
  &\rho^{-1}\norm{\varphi_\rho}^2\to 0\text{ for }\rho\to \infty;\\
&\varphi_\rho:=\chi_-(B/\epsilon) T_\rho u^\alpha\otimes
  v^{+}_{\lambda,\alpha},\quad T_\rho=\chi_-(r/\rho)\xi^+_a(\brH_a-E)^{-1}.  
  \end{split}
\end{align} Note at this point that
\begin{align}\label{eq:e-lam}
  \begin{split}
   &(\brH_a-E)\parb{u^\alpha\otimes
  v^{+}_{\lambda,\alpha}}-(\lambda-E)\parb{u^\alpha\otimes
  v^{+}_{\lambda,\alpha}}\in \vB_0^*,\\
&{u^\alpha\otimes
  v^{+}_{\lambda,\alpha}}-(\lambda-E)(\brH_a-E)^{-1}\parb{u^\alpha\otimes
  v^{+}_{\lambda,\alpha}}\in \vB_0^*, 
  \end{split}
\end{align} cf. \eqref{eq:basiC0}. As we will below the introduced factor
$(\brH_a-E)^{-1}$ in $T_\rho$ facilitates  commutation in the following estimation:

We estimate (with $\inp{T}_\varphi:=\inp{\varphi,T\varphi}$),
commute repeatedly (using Lemma \ref{lem:fHB}), apply
\eqref{eq:self-similar} and (in the last step) \eqref{eq:r0},
\begin{align*}
  &-2\rho^{-1}\norm{\varphi_\rho}^2\\
&\leq-\rho^{-1}\inp{B/\epsilon}_{\varphi_\rho}\\&=-(\rho \epsilon)^{-1}\inp{T^*_\rho
  B T_\rho}_{\chi_-(B) \parb{u^\alpha\otimes
  v^{+}_{\lambda,\alpha}}}+o(\rho^{0})
\\&=-(\rho \epsilon)^{-1}\Re\inp{T^*_\rho
  r^{-1}\parb{2{x_a}\cdot p_a+{\mathop{\mathrm{grad}} (r^a)^2}\cdot p^a} T_\rho}_{\chi_-(B) \parb{u^\alpha\otimes
  v^{+}_{\lambda,\alpha}}}+o(\rho^{0})
\\&=-(\rho \epsilon)^{-1}\inp{T^*_\rho
  r^{-1}\parb{2\sqrt{\lambda_\alpha}\abs{x_a}} T_\rho}_{\chi_-(B) \parb{u^\alpha\otimes
  v^{+}_{\lambda,\alpha}}}+o(\rho^{0})
\\&=-(\rho
    \epsilon)^{-1}{2\sqrt{\lambda_\alpha}\inp[\big]{\inp{x}/r}}_{\varphi_\rho}+o(\rho^{0})
\\&=-(\rho \epsilon)^{-1}{2\sqrt\lambda_\alpha}\norm{\varphi_\rho}^2+o(\rho^{0}).
\end{align*}
 Note  for the fourth step that we need to commute $r^{-1}{x_a}\cdot
 p_a$ through a factor $\chi_-(B) $  (before we can replace $r^{-1}{x_a}\cdot
 p_a\approx \sqrt\lambda_\alpha r^{-1}\abs{x_a}$ and commute back). After a
 little commutation this may in turn be accomplished by showing 
 \begin{align*}
   \comm[\big]{\brH_a-E)^{-1}r^{-1}{x_a}\cdot p_a(\brH_a-E)^{-1},\chi_-(B) }=\vO(r^{-1})
 \end{align*} in the sense of \eqref{eq:1712022}. However  with the
 two factors of  resolvents in place the result follows from the Helffer--Sj\"ostrand
formula and  Lemma \ref{lem:fHB}, cf. Remark
\ref{remark:comm-with-a_3yy}.  

Clearly the above estimation yields \eqref{eq:goal} for any  positive 
$\epsilon<{\sqrt\lambda_\alpha}$ (by a subtraction). \qed

\section{Proof of \eqref{eq:lim00}
  \eqref{eq:wave_opc2} and \eqref{eq:conjF}}\label{sec:AppendixA}

For convenience we consider only the assertions for $t\to
+\infty$. Our procedure relies on standard stationary phase analysis
on which we omit  details, see for
example \cite{H0, II} for  elaborated accounts.

\underline{{ \eqref{eq:lim00}:}} \quad We need to show that 
\begin{align}\label{eq:simp1}
  \begin{split}
 &\forall \varphi \in L^2(\mathbf X_a):\quad \lim_{t\to +\infty}\norm{\parb{I-A_1 A_2^aA_3^aA_2^aA_1}\parb{u^\alpha\otimes  \varphi(t)}}=0;
\\&\quad \quad\quad\quad\varphi(t)=\e^{-\i S_a (p_a,t)}f_2(k_\alpha)\varphi.   
  \end{split}
\end{align} We can assume that $\varphi$ is smooth in momentum
space. Then also $f_2(k_\alpha)\varphi$ is  smooth  in momentum space,
and 
there it is 
compactly supported 
away from  origo.  Recalling  that $S_a (p_a,t)= t
\xi_a^2+\vO(t^{1-\mu})$ the stationary phase
method then yields the effective localization $\abs{x_a}\geq \epsilon'
t$ for some (computable) $\epsilon'>0$. On the other hand $x^a$ is a localized
variable due to the presence of the factor $u^\alpha$. This means more
precisely that  we may  freely insert to the left of the tensor product
in \eqref{eq:simp1} the   factor $\chi_-\parb{{
    r^{-\delta'}}{\abs{x^a}}}$ for any  $\delta'>0$. We do this for a $\delta'<\delta$.

Noting  that we can also freely insert the  factor $\chi_+\parb{{\Re{\parb{(x_a/r)\cdot
  p_a}}}/(4\epsilon_0)}$ to the left of the tensor product in
\eqref{eq:simp1}  (thanks to stationary phase analysis),  it follows essentially 
from \eqref{eq:good} that $A_1^2=A_{1+}^2$
can be replaced by $I$. Similarly the factors $A_2^a$ and
$A_3^a=A_{3+}^a$ can one by one be replaced by
$I$. This is 
due to the fact that 
  in $\set{\abs{x^a}<c r^\delta}$,
 with $c>0$ given in  the property \ref{item:2c} of Section \ref{sec:Derezinski's
  construction} (this is for $r$ replaced by $r^a$),
$r^a_\delta=r^\delta r^a(0)$ and $B_{\delta,\rho_1}^a=r^{\rho_1/2}B_\delta^a
                         r^{\rho_1/2}=0$. Clearly we
can then remove   $A_2^a$, and the
justification   of  removing $A_3^a$  is  provided by  \eqref{eq:aligeen}.  This completes the proof of \eqref{eq:lim00}.
 
\underline{{ \eqref{eq:wave_opc2} and \eqref{eq:conjF}:}} \quad
%Parallel to \eqref{eq:simp1b}
 Due to 
\eqref{eq:simp1} it suffices to show (recalling the notation
\eqref{eq:primes})  that for all $\varphi  \in L^2(\mathbf X_a)$
with $\hat \varphi\in C^\infty_\c(\mathbf X'_a):$ 
\begin{align}
   \label{eq:simp1bcde}
  \begin{split}
 &\lim_{t\to
   +\infty}\norm{{M_a\parb{u^\alpha\otimes
       \varphi(t)}-u^\alpha\otimes  \parb{{m_\alpha^+}\varphi(t)}}}=0,\\
&\lim_{t\to
   +\infty}\norm{\parb{f_1(H)-f_1(\brH_a)}\parb{u^\alpha\otimes  \parb{\parb{m_\alpha^+}^2\varphi(t)}}}=0;
\\&\quad \quad\quad\quad\varphi(t)=\e^{-\i S_a (p_a,t)}f_2(k_\alpha)\varphi.   
  \end{split}
\end{align}
  The first assertion of \eqref{eq:simp1bcde} follows by combining
  the properties \ref{item:1a}
and \ref{item:2a} from Section \ref{sec:Yafaev's construction} with arguments from our
  proof  of  \eqref{eq:simp1} (including  in particular stationary
  phase analysis). The second part follows easily from stationary
  phase analysis.  \qed

\section{Proof  of \eqref{eq:SreP} and formulas for channel wave matrices}\label{sec:AppendixB} The
 derivation of the formula \eqref{eq:SreP} for $\widetilde
  S_{\beta\alpha}(\lambda)$, interpretated correctly in
  Subsection \ref{subsec:Conclusion of argument, the
              weak continuity},  will be given using
  smoothness   bounds from Section \ref{sec:Computation of T}, arguments from Subsections \ref{subsec:Trace estimates} and \ref{subsec:Conclusion of argument, the
              weak continuity}  and  finally bounds from \cite{AIIS}. We follow essentially the scheme of \cite[Appendix
  A]{DS} (for a similar issue, see    for example the proof of \cite[Proposition 7.2]{Ya2}).

Formally \eqref{eq:formwOp}  leads with \cite[Corollary
  1.10]{AIIS} to the formula
\begin{align}\label{eq:SreP4}
  \begin{split}
  \widetilde
  S_{\beta\alpha}(\lambda)&= -2\pi f^2_2(\lambda){\gamma_b^+(\lambda_\beta)}J^*_\beta\parb{\Phi_\beta^+}^*\parb{T^-_\alpha
  -\i(T^+_\beta)^* R(\lambda+\i 0)T^-_\alpha}
    J_\alpha{\gamma_a^-(\lambda_\alpha)}^*\\&=
 (2\pi \i)^2f^2_2(\lambda){\gamma_b^+(\lambda_\beta)}J^*_\beta
  \parb{T^+_\beta}^* \delta(H-\lambda)T^-_\alpha
    J_\alpha{\gamma_a^-(\lambda_\alpha)}^*.  
  \end{split}
\end{align} Note that we have not given an independent justification
 for
the above middle  term (with the factor $R(\lambda+\i 0)$), including its well-definedness. However
this will not be 
needed, and the stated expression to the far right   does indeed have a clean
interpretation, cf. Remark \ref{remark:breveBNDs}
\ref{item:P2} and Subsections \ref{subsec:Trace estimates} and \ref{subsec:Conclusion of argument, the
              weak continuity}.

            Let for $\epsilon>0$ and $\lambda\in \Lambda$
            \begin{align*}
              \delta_{\epsilon,\beta}(\lambda)= (2\i\pi)^{-1}\parb{(k_\beta -\lambda-\i \epsilon)^{-1}-(k_\beta -\lambda+\i \epsilon)^{-1}}.
            \end{align*} The outset for our analysis is the following
            two formulas which may be derived exactly as in \cite[Appendix
  A]{DS}. In the first formula $g$ is any complex continuous function
   on $\R$ vanishing at infinity. 
  \begin{align}
    \begin{split}
\label{eq:W}
    &\widetilde W^+_\beta (g1_{\Lambda})(k_\beta)\varphi\\
    &\quad=\lim_{\epsilon\to 0}\int_{\Lambda}\,g(\lambda)f_1(H)\parb{\Phi_\beta^+
      +\i R(\lambda-\i \epsilon)T^+_\beta}J_\beta \breve w^+_b f_2(k_\beta)\delta_{\epsilon,\beta}(\lambda)
      \varphi \,\d \lambda, \end{split}\\
    \begin{split}
&\inp{\varphi_b,\widetilde
  S_{\beta\alpha}\varphi_a}\\&\quad=-2\pi \lim_{\epsilon\to 
                                   0}\int^\infty_{-\infty}\,\inp{\varphi_b,\delta_{\epsilon,\beta}(\lambda')\parb{
                                   \widetilde W^+_\beta }^*T^-_\alpha J_\alpha \breve w^-_a f_2(k_\alpha)\delta_{\epsilon,\alpha}(\lambda')
      \varphi_a} \,\d \lambda'.\label{eq:S}  
    \end{split}
  \end{align} Note that \eqref{eq:S} is based on the fact that
  \begin{align*}
    \slim_{t\to +\infty}\e^{\i tH}\tilde{J}_\alpha^-\e^{-\i
   tk_\alpha}f_2(k_\alpha)=0,
  \end{align*} which in  turn is a consequence of \eqref{eq:lessgood} and stationary phase analysis, cf. Appendix
  \ref{sec:AppendixA}.

Although \eqref{eq:W} is valid for any $\varphi\in L^2(\mathbf
X_b)$ we take below  
$\varphi=\varphi_b\in L^2_s(\mathbf  X_b)$ for  $s>1/2$,  and compute by taking
$\epsilon\to 0$    using \cite[Corollary
  1.10]{AIIS}, here done formally only,
  \begin{subequations}
\begin{align*}
  \begin{split}
  \widetilde W^+_\beta g(k_\beta)\varphi
    &=2\pi\int_{\Lambda}\,g(\lambda) f_1(H)\delta(H-\lambda)T^+_\beta
    J_\beta \breve w^+_b f_2(k_\beta)\delta_{0,\beta}(\lambda)
      \varphi \,\d \lambda;\\
\delta_{0,\beta}(\lambda)&=(2\i\pi)^{-1}\parb{(k_\beta -\lambda-\i
  0)^{-1}-(k_\beta -\lambda+\i
  0)^{-1}}\\&=\gamma_{b,0}(\lambda_\beta)^* \gamma_{b,0}(\lambda_\beta).  
  \end{split}
\end{align*} See  Section \ref{sec:one-body  matrices} for
notation, and note that  this formula   has the following precise
meaning. We substitute, referring again to Section \ref{sec:one-body
  matrices}, 
\begin{align*}
  \breve w^+_b f_2(k_\beta)\delta_{0,\beta}(\lambda)
       =f_2(\lambda) \gamma_{b}^+(\lambda_\beta)^*\gamma_{b,0}(\lambda_\beta).
\end{align*}
    Then by combining the proof of
              Lemma \ref{lemma:exist_integr-tpm_-}, 
              Remark \ref{remark:breveBNDs} \ref{item:P2} and Lemma
              \ref{lemma:strongCont} we conclude that the integral to
              the right  is well-defined,
              \begin{align}\label{eq:Wfinal} \begin{split}
                &
                \widetilde W^+_\beta  g(k_\beta)\varphi\\&\quad
    =2\pi\int \,(gf_2)(\lambda)  \parb{ f_1(H)\delta(H-\lambda) T^+_\beta J_\beta 
      \gamma_{b}^+(\lambda_\beta)^*}\gamma_{b,0}(\lambda_\beta)\varphi \,\d \lambda, \end{split}
              \end{align} where the correct interpretation of the product in
              parentheses involves the `$Q$-operators' as in
              Subsection \ref{subsec:Conclusion of argument, the
              weak continuity} (more precisely  obtained by expanding
            $T^+_\beta$ into terms on the
            form \eqref{eq:Ex1} given below). We call \eqref{eq:Wfinal} a `channel
            wave matrix' representation (examined closer in Section \ref{sec:Exact channel wave-matrices}). Although the below proof of
            \eqref{eq:Wfinal} reveals a different representation,  more
            closely related to \eqref{eq:W}, the formula
            \eqref{eq:Wfinal} is more useful for most of our  purposes
            (with exceptions  only in Subsection \ref{subsubsec:Besov space setting at stationary scattering
    energies}). The
            alternative formula
            reads
\begin{align}\label{eq:AltForm}
  \begin{split}
                &\widetilde W^+_\beta  g(k_\beta)\varphi\\
    &\quad =\int \,(gf_2)(\lambda)   f_1(H)\parb{\Phi_\beta^++\i R(\lambda-\i 0) T^+_\beta }J_\beta 
      \gamma_{b}^+(\lambda_\beta)^*\gamma_{b,0}(\lambda_\beta)\varphi \,\d \lambda.
  \end{split}
\end{align}  
\end{subequations}
\begin{subequations} 

Now there are two assertions
               that need to be checked to justify  \eqref{eq:Wfinal}  by
               taking the  `$\epsilon\to 0$'--limit  in
              \eqref{eq:W},  in the precise meaning of taking  limit
              in the weak topology of $L^2_{-1}(\mathbf
X)$:  
\begin{align}
\i \lim_{\epsilon\to 0}&\int_{\Lambda}\,g(\lambda) f_1(H)\parb{R(\lambda-\i \epsilon)-
      R(\lambda+\i \epsilon)}T^+_\beta J_\beta \breve w^+_b f_2(k_\beta)\delta_{\epsilon,\beta}(\lambda)
      \varphi \,\d \lambda \nonumber
\\&=\i \lim_{\epsilon\to 0} \int_{\Lambda}\,g(\lambda) f_1(H)\parb{R(\lambda-\i \epsilon)-
      R(\lambda+\i \epsilon)}T^+_\beta J_\beta f_2(\lambda) \gamma_{b}^+(\lambda_\beta)^*\gamma_{b,0}(\lambda_\beta)
      \varphi \,\d \lambda \nonumber\\&
=2\pi\int_{\Lambda}\,g(\lambda) f_1(H)\delta(H-\lambda)T^+_\beta
    J_\beta f_2(\lambda) \gamma_{b}^+(\lambda_\beta)^*\gamma_{b,0}(\lambda_\beta)
      \varphi \,\d \lambda,\label{eq:W3}
         \end{align} and  
              \begin{align}
                \label{eq:W4}
                \lim_{\epsilon\to 0}\int_{\Lambda}\,g(\lambda) f_1(H)\parb{\Phi_\beta^+
      +\i R(\lambda+\i \epsilon)T^+_\beta}J_\beta \breve w^+_b f_2(k_\beta)\delta_{\epsilon,\beta}(\lambda)
      \varphi \,\d \lambda=0.
              \end{align}
  \end{subequations}

For \eqref{eq:W3}  we write with $\delta_{\epsilon,\lambda}(\lambda'):=\pi^{-1}\epsilon/\parb{(\lambda-\lambda')^2+\epsilon^2}$
for any  fixed  $\lambda\in \R$ and $\epsilon>0$,
\begin{align}\label{eq:delta}
  \begin{split}
  \breve w^+_b f_2(k_\beta)\delta_{\epsilon,\beta}(\lambda)
  \varphi &=\breve w^+_b \int^\infty _{\lambda^\beta}
  f_2(\lambda')\delta_{\epsilon,\lambda}(\lambda')\gamma_{b,0}(\lambda'-\lambda^\beta)^*
  \gamma_{b,0}(\lambda'-\lambda^\beta)\varphi \,\d \lambda'\\
&=\int^\infty_{\lambda^\beta}
  f_2(\lambda')\delta_{\epsilon,\lambda}(\lambda')\gamma^+_{b}(\lambda'-\lambda^\beta)^*
  \gamma_{b,0}(\lambda'-\lambda^\beta)\varphi \,\d \lambda',
  \end{split}
\end{align} substitute  and replace by the limit (when taking  $\epsilon\to
0$). This is doable thanks to Remark \ref{remark:breveBNDs}
\ref{item:P2}, Lemma \ref{lemma:strongCont} and the continuity of $\gamma_{b,0}(\lambda'_\beta)\varphi$ justifying  the first equality, and the second is a
consequence of Remark \ref{remark:breveBNDs}
\ref{item:P2} and Lemma \ref{lemma:strongCont} too. Recall from Subsection \ref{subsec:Conclusion of argument, the weak
continuity}  that $T^+_\beta$ is a finite sum of terms 
expressed as 
\begin{align}\label{eq:Ex1}
  f_1(H)Q^+(b,l)^*B_+Q^+(b,l)f_1(\brH_b)f_1(\brh_\beta),\quad
  B_+\text{ bounded},
\end{align} to be used below too.

For \eqref{eq:W4} we will invoke \cite[Theorem 1.8 and Corollary
1.9]{AIIS}.  Indeed thanks to these results and the presence of the factors
$A_{1+}$ in $\Phi^+_\beta$ we can compute the left-hand side to be
equal
  \begin{align}\label{eq:EXP}
\int_{\Lambda}\,g(\lambda) f_1(H)\parb{\Phi_\beta^+ +\i R(\lambda+\i
0)T^+_\beta}J_\beta f_2(\lambda)
\gamma_{b}^+(\lambda_\beta)^*\gamma_{b,0}(\lambda_\beta) \varphi \,\d
\lambda.
  \end{align} To see this, considering only the `difficult term'
  \begin{align*} \int_{\Lambda}\,g(\lambda) f_1(H)R(\lambda+\i
\epsilon)T^+_\beta J_\beta \breve w^+_b
f_2(k_\beta)\delta_{\epsilon,\beta}(\lambda) \varphi \,\d \lambda,
  \end{align*} we insert $I=\chi^2_+( 2B/\epsilon_0)+\chi^2_-(
2B/\epsilon_0)$ to the right of the factor $R(\lambda+\i
\epsilon)$. Using the factors of $A_{1+}$ and a little commutation
the contribution from the second term $\chi^2_-( 2B/\epsilon_0)$
allows the insertion of the weight $\inp{x}^{-s}$ for some  $s>1/2$  to the
right of $R(\lambda+\i \epsilon)$ and we can use \eqref{eq:LAPbnd} and
\eqref{eq:delta} to compute the `$\epsilon\to 0$'--limit for this
term (more precisely we use that $\inp{x}^{s}\chi^2_-(
2B/\epsilon_0)f_1(H)Q^+(b,l)^*B_+$ is bounded). As for the contribution from the first term $\chi^2_+(
2B/\epsilon_0)$ we write
  \begin{align*} R(\lambda+\i \epsilon)\chi^2_+(
2B/\epsilon_0)=R(\lambda+\i \epsilon)\chi^2_+(
2B/\epsilon_0)\inp{x}^{-s}\inp{x}^{s};\quad s>1/2-\mu/2.
  \end{align*} Then by \cite[Theorem 1.8 and Corollary 1.9]{AIIS}
  \begin{align*} \inp{x}^{-1}f_1(H)R(\lambda+\i \epsilon)\chi^2_+(
2B/\epsilon_0)\inp{x}^{-s}\text{ is uniformly bounded},
\end{align*} and there exists 
\begin{align*} \Lambda\ni \lambda\to\inp{x}^{-1}f_1(H)&R(\lambda+\i
0)\chi^2_+( 2B/\epsilon_0)\inp{x}^{-s}\\&=\lim_{\epsilon\to 0}
\,\inp{x}^{-1}f_1(H)R(\lambda+\i \epsilon)\chi^2_+(
2B/\epsilon_0)\inp{x}^{-s},
  \end{align*} constituting a continuous $\vL(\vH)$-valued
function. Consequently, using again \eqref{eq:delta}, it suffices
to check the existence and continuity of the $\vL(L^2(C_b),\vH)$-valued
function
  \begin{align*} \R\ni \lambda \to f_2(\lambda)\inp{x}^{s}T^+_\beta
J_\beta \gamma_{b}^+(\lambda_\beta)^*.
  \end{align*} Expanding $T^+_\beta$ as above, 
 we can use Lemma \ref{lemma:strongCont} to combine for
each  term the factor $Q^+(b,l)$ with the factor $J_\beta
\gamma_{b}^+(\lambda_\beta)^*$ (as done above). Each  remaining factor of a `$Q$-operator'
(more precisely $Q^+(b,l)^*$ appearing to the left)   contributes by a factor 
$\vO(r^{\rho_1/2-\delta/2}) $, cf. Remark \ref{remark:breveBNDs} \ref{item:P0}. Whence we are left with checking that
$s+\rho_1/2-\delta/2\leq 0$ is possible. Due to \eqref{eq:parameters}
it suffices to produce an $s>1/2-\mu/2$ such that $s+ 1/2-2/(2+\mu)\leq 0$,
    and therefore in turn to check that $ 1/2-\mu/2+ 1/2-2/(2+\mu)<0$.
  But the latter condition is fulfilled for all $\mu>0$,
so the check is done. Consequently indeed the left-hand side of \eqref{eq:W4} is given by the
integral \eqref{eq:EXP}.

Next, using the above arguments and  \cite[Corollary
  1.10]{AIIS},  the expression
  \eqref{eq:EXP} may be  computed   with  $\chi_\rho:=\chi_-(r/\rho)$ as
\begin{align*} 
&\int_{\Lambda}\,g(\lambda)f_1(H)\parb{\Phi_\beta^+
      - R(\lambda+\i 0)(H-\lambda)\Phi^+_\beta}J_\beta f_2(\lambda) \gamma_{b}^+(\lambda_\beta)^*\gamma_{b,0}(\lambda_\beta)
      \varphi \,\d \lambda\\
&=\lim_{\rho\to \infty}\int_{\Lambda}\,g(\lambda)f_1(H)\parb{\Phi_\beta^+
      - R(\lambda+\i 0)\chi_\rho (H-\lambda)\Phi^+_\beta}J_\beta f_2(\lambda) \gamma_{b}^+(\lambda_\beta)^*\gamma_{b,0}(\lambda_\beta)
      \varphi \,\d \lambda\\
&=\lim_{\rho\to \infty}\int_{\Lambda}\,g(\lambda) f_1(H)R(\lambda+\i
  0)[H,\chi_\rho]\Phi^+_\beta J_\beta f_2(\lambda) \gamma_{b}^+(\lambda_\beta)^*\gamma_{b,0}(\lambda_\beta)
      \varphi \,\d \lambda\\
&=0,
  \end{align*} proving \eqref{eq:W4}.

Now to prove \eqref{eq:SreP},  we would like to substitute the adjoint expression  of \eqref{eq:Wfinal} with $g(\lambda)=\delta_{\epsilon,\lambda'}(\lambda)$ into
\eqref{eq:S} and then  interchange the order of the two
integrations. This requires of course some modification since the
meaning of the right-hand side of  \eqref{eq:Wfinal} is a vector in  $L^2_{-s}(\mathbf
X)$  for any $s>1/2$ (but unlikely any  smaller),  and there is no
reason to expect  that the quantity  $T^-_\alpha J_\alpha \breve w^-_a f_2(k_\alpha)\delta_{\epsilon,\alpha}(\lambda')
      \varphi_a \in L^2_{s}(\mathbf
X)$ for any such $s$. However, writing    $T^-_\alpha$ as  a finite sum terms on the
form
\begin{align}\label{eq:Ex2}
  f_1(H)Q^-(a,k)^*B_-Q^-(a,k)f_1(\brH_a)f_1(\brh_\alpha),\quad
  B_-\text{ bounded},
\end{align}
   we can introduce 
the  modification, say denoted by 
$ T^-_{\alpha,\rho}$, given by inserting for each such term  the above factor
$\chi_\rho$ to the right of $Q^-(a,k)^*$. Using  Fubini's theorem and
the computation 
\begin{align*}
  \int^\infty_{-\infty}\,\delta_{\epsilon,\lambda'}(\lambda)\delta_{\epsilon,\alpha}(\lambda')\,\d \lambda'=\delta_{2\epsilon,\alpha}(\lambda),
\end{align*} we then obtain 
\begin{align*}
  &(2\pi \i)^{-2}\inp{\varphi_b,\widetilde
  S_{\beta\alpha}\varphi_a}=\lim_{\epsilon\to
                                   0}\,\lim_{\rho\to
                                   \infty}\\&\quad \int_{\Lambda}\,f_2(\lambda) \inp{\gamma_{b,0}(\lambda_\beta)\varphi_b,\gamma^+_{b}(\lambda_\beta)J^*_\beta (T^+_\beta)^*\delta(H-\lambda)T^-_{\alpha,\rho} J_\alpha \breve w^-_a f_2(k_\alpha)
    \delta_{2\epsilon,\alpha}(\lambda)
      \varphi_a} \,\d \lambda,
\end{align*} valid  for any $\varphi_a\in L^2_s(\mathbf  X_a)$ and
$\varphi_b\in L^2_s(\mathbf  X_b)$ for
$s>1/2$.

 With an analogue of \eqref{eq:delta} for $\alpha$ rather
than $\beta$ we compute 
\begin{align*}
  &\lim_{\epsilon\to
                                   0}\int_{\Lambda}\,f_2(\lambda) \inp{\gamma_{b,0}(\lambda_\beta)\varphi_b,\gamma^+_{b}(\lambda_\beta)J^*_\beta (T^+_\beta)^*\delta(H-\lambda)T^-_{\alpha,\rho} J_\alpha \breve w^-_a f_2(k_\alpha)
    \delta_{2\epsilon,\alpha}(\lambda)
      \varphi_a} \,\d \lambda\\& =\int_{\Lambda}\,f_2(\lambda) \inp{\gamma_{b,0}(\lambda_\beta)\varphi_b,\gamma^+_{b}(\lambda_\beta)J^*_\beta (T^+_\beta)^*\delta(H-\lambda)T^-_{\alpha,\rho} J_\alpha f_2(\lambda)
    \gamma_{a}^-(\lambda_\alpha)^*\gamma_{a,0}(\lambda_\alpha)
      \varphi_a} \,\d \lambda,
\end{align*} uniformly in $\rho>1$. Due to  these 
features  we can interchange limits above and conclude that
\begin{align*}
  &(2\pi \i)^{-2}\inp{\varphi_b,\widetilde
  S_{\beta\alpha}\varphi_a}=\lim_{\rho\to
                                   \infty}\\&\int_{\Lambda}\,f_2(\lambda) \inp{\gamma_{b,0}(\lambda_\beta)\varphi_b,\gamma^+_{b}(\lambda_\beta)J^*_\beta (T^+_\beta)^*\delta(H-\lambda)T^-_{\alpha,\rho} J_\alpha f_2(\lambda)
    \gamma_{a}^-(\lambda_\alpha)^*\gamma_{a,0}(\lambda_\alpha)
      \varphi_a} \,\d \lambda
\\&=\int_{\Lambda}\,f_2(\lambda) \inp{\gamma_{b,0}(\lambda_\beta)\varphi_b,\gamma^+_{b}(\lambda_\beta)J^*_\beta (T^+_\beta)^*\delta(H-\lambda)T^-_\alpha J_\alpha f_2(\lambda)
    \gamma_{a}^-(\lambda_\alpha)^*\gamma_{a,0}(\lambda_\alpha)
      \varphi_a} \,\d \lambda,
\end{align*} showing  \eqref{eq:SreP}. Note that indeed the right
interpretation is given by expansion into a (finite) sum,  substituting the expressions
\eqref{eq:Ex1} and \eqref{eq:Ex2}, cf. Subsection \ref{subsec:Conclusion of argument, the
              weak continuity}. \qed

\begin{remark}\label{remark:proof-eqrefeq:srep} Note the following analogues of
  \eqref{eq:Wfinal} and \eqref{eq:SreP}, cf. \eqref{eq:locScata} and \eqref{eq:locScatb}, 
  \begin{subequations}
  \begin{align}\label{eq:Wfinal2a}
\widetilde W^-_\beta g(k_\beta)\varphi
    &=-2\pi\int \,(gf_2)(\lambda)  \parb{ f_1(H)\delta(H-\lambda) T^-_\beta J_\beta 
      \gamma_{b}^-(\lambda_\beta)^*}\gamma_{b,0}(\lambda_\beta)\varphi
    \,\d \lambda,\\
    \begin{split}
 16\lambda_\beta \lambda_\alpha f^2_2(\lambda)&m_b(\pm\hat
  \xi_b)^2m_b(\pm\hat \xi_a)^2\,\delta_{ \beta \alpha}\\&\quad =-(2\pi \i)^2f^2_2(\lambda){\gamma_b^\pm(\lambda_\beta)}J^*_\beta
  \parb{T^\pm _\beta}^*\delta(H-\lambda)T^\pm _\alpha
    J_\alpha{\gamma_a^\pm (\lambda_\alpha)}^*. \label{eq:Wfinal2b}   
    \end{split}
              \end{align} The  quantity to the  left in
              \eqref{eq:Wfinal2b} is meant to be  an operator in
              $\vL\parb{L^2(C_a),L^2(C_b)}$;  the use of the
              Kronecker symbol  $\delta_{
                \beta \alpha}$ specifies that it vanishes unless
              $\beta =\alpha$. The operator  may be considered as   the fiber of $F_\beta \parb{\widetilde W_\beta^\pm}^*\widetilde W^\pm_\alpha
  F_\alpha^{-1}$ at energy $\lambda$, invoking 
the  orthogonality of channels.  The formula \eqref{eq:Wfinal2b} 
results 
              by mimicking 
              the above procedure for showing \eqref{eq:SreP}.
  \end{subequations}
\end{remark}

\end{document}